\documentclass{article}
\usepackage{amsmath,amssymb,amsthm}
\usepackage{graphics}
\usepackage{xfrac}

\DeclareMathAlphabet{\mathsfsl}{T1}{cmss}{bx}{n} 

\def\Ga#1{\Gamma^*_{\!#1}}
\def\GaN{\Gamma^*_{\!N}}
\def\barGamma{\mkern1mu\overline{\hbox{$\mkern-1mu\Gamma\mkern4mu$}}\mkern-4mu}
\def\clGa#1{\barGamma^*_{\!#1}}
\def\GN{\Gamma_{\!N}}
\let\norm=\|%
\def\|{\mkern1mu|\mkern1mu}
\def\down(#1,#2){{\downarrow}^{#1}_{#2}}
\newcommand\sm{\tsp{-}\tsp}
\newcommand\del{\mathbin{\setminus}} 
\newcommand\contr{\mathbin{\mkern -7mu\not\mkern 7mu}} 

\newcommand\eps{\varepsilon}
\newcommand\rest{{\restriction}}
\newcommand\up{{\uparrow}}
\newcommand\dn{{\downarrow}}
\newcommand\subgr{\trianglelefteq}
\DeclareMathOperator\Prob{Pr} 
\DeclareMathOperator\prob{\mathbb P} 
\DeclareMathOperator\E{\mathsf E}
\DeclareMathOperator\dom{dom}

\newcommand\cups{\mathrel{\cup\mkern -1mu^*\mkern-3mu}}
\def\restr{\mathord{\upharpoonright}}

\def\pge{\mathbin{\ge}}

\def\m#1{\mkern 1.5mu{#1}\mkern 1.5mu}
\def\phi{\varphi}
\def\tto{\to\mkern-14mu\to}
\def\<#1,#2\|#3\>{\langle#1,#2\|#3\rangle}

\newtheorem{theorem}{Theorem}
\newtheorem{lemma}[theorem]{Lemma}

\newtheorem{proposition}[theorem]{Proposition}
\newtheorem{claim}[theorem]{Claim}
\theoremstyle{definition}
\newtheorem{definition}[theorem]{Definition}
\newtheorem{problem}{Problem}
\renewenvironment{cases}{\left\{\mkern-6mu\begin{array}{l@{~~~}l}}{\end{array}\right.}
\makeatletter
\def\ssum{\raisebox{1pt}{$\textstyle\sum@$}\mkern 1.5mu}

\def\tsp{\mkern 1mu}
\def\fc{f^{\mathsf c}}
\newcommand\eqdef{\stackrel{\mathrm{def}}{=}}
\def\sharp{{\resizebox{!}{5pt}{\raisebox{2pt}{$\parallel$}}}}
\makeatother
\usepackage[bookmarks=true,hyperindex=true]{hyperref}
\hypersetup{
  pdftitle={Exploring the Entropic Region},
  pdfauthor={Laszlo Csirmaz},
  bookmarksopen = false,
  pageanchor = false,
  plainpages = false,
}
\newcommand\bref[1]{\hyperlink{#1}{#1}}
\def\r{\mathsfsl r}
\def\MEM{\textsf{MAXE}}
\def\GMEM{\textsf{GMAXE}}
\DeclareMathOperator\GAK{\textsf{GAK}}
\def\H{\mathbf H}
\def\I{\mathbf I}

\title{Exploring the entropic region}
\author{Laszlo Csirmaz}
\date{}

\begin{document}
\maketitle

\begin{abstract}

\noindent

The paper explores three known methods, their variants and limitations,
that can be used to obtain new entropy inequalities. The Copy Lemma was
distilled from the original Zhang--Yeung construction which produced the
first non-Shannon inequality. Its iterated version, effects of
symmetrizations, and connections with polyhedral vertex enumeration are
discussed. Another method, derived from the principle of maximum entropy,
has the Copy Lemma as a special case. Nevertheless, none of the two
presented variants is known to generate more inequalities than the iterated
Copy Lemma. Finally, the Ahlswede-K\"orner method is shown to employ a
hidden application of the Copy Lemma -- the underlying lemma alone cannot
generate new inequalities --, which makes this method strictly weaker than
the Copy Lemma. The paper is written in a tutorial style and concludes with
a list of open questions and research problems.


\smallskip\noindent\textbf{AMS Numbers}
06B35, 26A12, 52B12, 90C29
\end{abstract}

\section{Introduction}

Questions about probabilistic representation of structures and structure
classes can often be reduced to the following one: Does there exist a
collection of (finite) random variables with specified entropies? Examples
include, but are not limited to, channel coding \cite{Csiszar.Korner} and
network coding in particular \cite{yeung-network}, estimating the efficiency
of secret sharing schemes \cite{Beimel,Beimel.Orlov, Csirmaz.Ligeti, Farras.Kaced,
G.Rom}, questions about matroid representability \cite{B.B.F.P}, guessing
games \cite{riis}, and finding
conditional independence inference rules \cite{studeny}. The $N$-variable
\emph{entropy region} is the collection of the entropies of the non-empty
subsets (i.e., the marginals) of every possible jointly distributed $N$ discrete
random variables indexed by the non-empty subsets of $N$. Shannon's basic
information inequalities \cite{yeung-book} provide linear inequalities among
these entropies, thus yielding a polyhedral bound of the entropy region. The
long-standing open question of whether this Shannon bound is tight has been
settled by Zhang and Yeung \cite{ZhY.ineq,Zh.gen.ineq} by exhibiting a linear inequality
that holds for entropies but is not a consequence of the basic inequalities.
Their discovery initiated intensive research to further investigate the
shape of the entropy region. The phrase \emph{Copy Lemma} was coined by
Dougherty et al.~\cite{DFZ11} to describe the general method distilled from
the original Zhang--Yeung construction. It has been applied successfully to
generate several hundred sporadic and a couple infinite families of
non-Shannon entropy inequalities \cite{Csirmaz.book,DFZ11,M.infinf}. A
different method using an information-theoretic lemma attributed to Ahlswede
and K\"orner \cite{Ahlswede.Korner} was proposed in \cite{MMRV}; later it
was shown to be equivalent to a special case of the Copy Lemma \cite{Kaced}.
Structures satisfying all instances of the Copy Lemma were named
\emph{self-adhesive} by F.~Mat\'u\v s who characterized self-adhesive
structures for the case when $N$ has at most four elements \cite{M.fmadhe}.
Self-adhesivity was an important ingredient in investigating different
representations of conditional independence models
\cite{Boege.Gauss,Boege.Studeny}.

The Copy Lemma, and its variants, are especially amenable to numerical
computations \cite{B.B.F.P,Beimel.Orlov,multiobj,G.Rom}. Although the Copy
Lemma can utilize previously generated non-Shannon inequalities, the same
effect can be achieved by applying the Lemma \emph{iteratively}. The influential
paper of Dougherty et al.~\cite{DFZ11} also introduces such iterative
variants. Additional tricks, such as exploiting the inherent symmetry and
special structure of the entropy region, can significantly reduce the
otherwise prohibitively huge computational complexity \cite{multiobj}.

The Zhang-Yeung method can be recovered by an ingenious application of
\emph{principle of maximum entropy}. The principle, which was conceived in
its modern form by E.~T.~Jaynes \cite{jaynes}, is easy to formulate: ``If a
probability distribution is specified only partially, take the one with the
largest entropy,'' see \cite{maxentp}. In this case the partial
specification requires certain marginals (those at the ``copied instance''
positions) to be isomorphic to a given, fixed distribution. Since this
requirement imposes linear constraints on the distribution probabilities,
the maximum entropy distribution always exists \cite{convex-optim}, is
unique, and satisfies numerous independence statements, including those
provided by the Copy Lemma.

While it is clear that the Maximum Entropy Method is more general than the
Copy Lemma, it is not known whether it is actually stronger and could
provide additional entropy inequalities. Finally, it is also unknown whether
any of these methods captures the complete entropy region, but there are
indications that this is not the case \cite{Boege.Gauss}.

The remainder of this paper is organized as follows. Section
\ref{sec:prelim} recalls some basic notions used later, such as probability
distribution, Shannon entropy, entropy region, and polymatroids. Basic facts
about the entropy region are proved in Section \ref{sec:entreg}, while
Section \ref{sec:polymatroids} is a quick introduction to polymatroids and
operations on them. Linear polymatroids and the Ingleton expression are
treated in Section \ref{sec:linear}. The Copy Lemma and its refinements are
discussed in Section \ref{sec:copylemma}. Section \ref{sec:MEP} describes
the Maximum Entropy Method and its limitations. The method based on the
Ahlswede-K\"orner lemma is discussed briefly in Section \ref{sec:AK}.
Section \ref{sec:conclusion} summarizes the results and collects some open
questions and research problems. Proofs of two information-theoretic
reductions, thinning and compressing a distribution, are provided in the
\bref{Appendix}.


\section{Preliminaries}\label{sec:prelim}

In this paper all sets are finite. Capital letters, such as $A$, $J$, $K$,
$N$, denote (finite) sets; elements of these sets are denoted by lowercase
letters such as $a$, $i$, $j$, etc. The union sign between sets is
frequently omitted, as well as the curly brackets around singletons. Thus,
$Nij$ denotes the set $N\cup\{i,j\}$. The difference of two sets is
written as $A\sm B$, or $A\sm b$ if the second set is a singleton. The
star in $A\cups B$ emphasizes that the arguments $A$ and $B$
are disjoint sets. A \emph{partition} of $N$ is a collection of non-empty
disjoint subsets of $N$ whose union equals $N$.

A \emph{discrete random variable $\xi$}, taking values from the finite set
$\mathcal X$ called the \emph{alphabet}, is determined by the probability
mass function $\Prob_\xi=\{\Prob_\xi(x):x\in\mathcal X\}$ specifying that
$\xi$ takes the value $x\in\mathcal X$ with probability $\Prob_\xi(x)$.
Being probabilities, we must have $0\le\Prob_\xi(x)\le 1$, and
$\sum_{x\in\mathcal X} \Prob_ \xi(x)=1$. When the random variable $\xi$ is
clear from the context, the index $\xi$ in $\Prob_\xi$ is omitted. The
\emph{support} of $\xi$ is the set of elements of $\mathcal X$ which $\xi$
can take with positive probability. $\xi$ is \emph{uniform} if it takes
every element of its alphabet with equal probability.

Suppose $\xi$ is defined on the direct product $\mathcal X=\prod\{\mathcal
X_i: i\in N\}$ for some finite index set $N$, also called \emph{base}.
For every non-empty subset $A$ of the base the \emph{marginal distribution $\xi_A$} is
defined on the partial product $\mathcal X_A=\prod\{\mathcal X_i:i\in A\}$
so that the probability of each $y\in\mathcal X_A$ is the sum of the
probabilities of those $x\in\mathcal X$ whose projection to $\mathcal X_A$
is just $y$:
$$
   \Prob_{\xi_A}(y)=\ssum \{ \Prob_\xi(x):x\restr A = y \}.
$$
In particular, for each $i\in N$ $\xi_i$ is a random variable defined on
$\mathcal X_i$, and $\xi_N$ is the same as $\xi$. To emphasize that they are
marginals, we write $\xi=\{\xi_i:i \in N\}$, and say that the variables
$\{\xi_i\}$ are \emph{jointly distributed}. Such a distribution is
\emph{quasi-uniform} if each $\xi_i$ is uniform, and the marginals $\xi_A$
are uniform on their own support for all non-singleton subsets $A$ (so
there might be elements of $\mathcal X_A$ that are taken with zero
probability, but all positive probabilities are the same).

Typical examples of quasi-uniform distributions come from groups
\cite{Chan.Yeung,yeung-book}. Let $G$ be a finite group and $G_i\subgr G$ be subgroups
of $G$ for $i\in N$. The domain of $\xi_i$ is the collection of the left
cosets $\{g G_i: g\in G\}$, and the joint distribution $\{\xi_i:i\in N\}$ is
defined as follows. Choose $g\in G$ uniformly and set $\xi_i$ to the left
coset $g G_i$. The probability of a sequence of such left cosets is the
number of group elements that produce these cosets, divided by the size of
the group $G$. For $A\subseteq N$ let $G_A=\bigcap\{G_i:i\in A\}$. As the left
cosets of $G_i$ and $G_j$ are either disjoint or have exactly $|G_{ij}| =
|G_i\cap G_j|$ common elements, the intersections of the elements of the
coset families $\{g G_i:g\in G\}$ and $\{g G_j:g\in G\}$ form a partition of
$G$ where each part has size $|G_{ij}|$. Similarly, the support of
$\xi_A$ has exactly $|G|/|G_A|$ elements, each obtained by the same number
of group elements, thus having the same probability.

\smallskip

The \emph{Shannon entropy} of $\xi$ is 
$$
\H(\xi)=\ssum \big\{\, {-}\Prob_\xi(x)\log_2 \Prob_\xi(x) :~ x\in\mathcal X
\big\},
$$ 
with the convention that $0\cdot\log_2 0 = 0$. If $\xi=\{\xi_i:i\in N\}$ is
jointly distributed on the base $N$ and $A\subseteq N$, then we write
$\H_\xi(A)$ for $\H(\xi_A)$; and even the index $\xi$ is dropped from the
notation and we write $\H(A)$ when $\xi$ is clear from the context. By
convention, $\H(\emptyset)=0$. Other information measures, called
conditional entropy, mutual information, and conditional mutual information,
are defined as follows; see \cite{yeung-book}:
\begin{align*}
   \H(A\|B) &\eqdef \H(A\cup B)-\H(B), \\
   \I(A,B) &\eqdef \H(A)+\H(B)-\H(A\cup B), \\
   \I(A,B\|C) &\eqdef \H(A\|C)+\H(B\|C)-\H(A\cup B\|C).
\end{align*}
Here $A$, $B$, $C$ are arbitrary subsets of $N$. \emph{Shannon
information inequalities} (S1) and (S2) below express, in an equivalent
form, that these information measures are non-negative (which, incidentally,
follows from the concavity of the $\log$ function, see \cite{yeung-book}).
\begin{itemize}\setlength\itemsep{0pt}
\item[(S1)] If $A\subseteq B$ then $0\le\H(A)\le \H(B)$ (monotonicity),
\item[(S2)] $\H(A)+\H(B)\ge \H(A\cap B)+\H(A\cup B)$ (submodularity).
\end{itemize}
Monotonicity in (S1) is equivalent to the non-negativity of the conditional
entropy $\H(B\|A)\ge 0$, while submodularity in (S2) is the non-negativity
of the conditional mutual information $\I(A,B\|A\cap B)\ge 0$ in a different
form. Non-negativity of mutual information is the special case of (S2)
when $A\cap B=\emptyset$. Non-negativity of $\I(A,B\|C)$ for arbitrary
subsets $A$, $B$, $C$ also follows from (S1) and (S2).

\smallskip

Any joint distribution $\xi$ on base $N$ gives rise to the marginal
entropies $\H_\xi(A)$ arranged into a vector indexed by the non-empty
subsets of $N$. The collection of these $(2^{|N|}\m-1)$-dimensional real
vectors forms the \emph{entropy region}, which is denoted by $\GaN$
\cite{yeung-book}. Elements of $\GaN$ are considered interchangeably as
vectors, as points in this high-dimensional Euclidean space, and as
functions assigning non-negative real numbers to non-empty subsets of the
base set $N$. A $(2^{|N|}\m-1)$-dimensional point (or vector, or a real
function on the non-empty subsets of $N$) is \emph{entropic} if it is in
$\GaN$, and it is \emph{almost entropic}, or \emph{aent}, if it is in $\clGa
N$, the closure of $\GaN$ in the usual Euclidean topology. In other words,
$\mathbf x$ is aent, if for each positive $\eps$ there is an entropic
$\mathbf y\in\GaN$ such that $\norm\mathbf x-\mathbf y\norm<\eps$ in the
usual Euclidean norm. The \emph{entropy profile} of the distribution $\xi$
is the vector formed from the marginal entropies of $\xi$, and the vector
$\mathbf x\in \GaN$ (or point, or function) is \emph{represented by} the
distribution $\xi$ if $\mathbf x$ is the entropy profile of $\xi$.

Notions of conditional entropy, mutual information and conditional mutual
information are formally extended to the functional form of entropic points.
If $f$ is a function on subsets of $N$ then the following notations are used
as abbreviations for arbitrary subsets $A,B,C$ of $N$:
\begin{align*}
   f(A\|B) &\eqdef f(AB)-f(B), \\
   f(A,B)  &\eqdef f(A)+f(B)-f(AB), \mbox{ and}\\
   f(A,B\|C) &\eqdef f(AC)+f(BC)-f(C)-f(ABC).
\end{align*}
Although the entropy function $f$ is not defined on the empty set (since there
is no coordinate in $\GaN$ indexed by the empty set), when convenient,
$f(\emptyset)=0$ will be assumed. In particular, $f(A,B\|\emptyset)$ is
considered the same as $f(A,B)$. If $f\in\GaN$, then each of the above
expressions is non-negative as guaranteed by the Shannon inequalities (S1)
and (S2).

\smallskip

Shannon inequalities (S1) and (S2) impose linear constraints on the
entropy region $\GaN$. The collection of vectors satisfying all Shannon
inequalities is denoted by $\GN$, clearly $\GaN\subseteq \GN$. $\GN$ is a
\emph{pointed polyhedral cone} \cite{ziegler} as it is bounded by finitely
many hyperplanes and each bounding hyperplane contains the origin. The
collection of the facets of $\GN$ is the (unique) minimal subset of
Shannon-inequalities -- the so-called \emph{basic inequalities}, see
\cite{yeung-book} -- which imply all other Shannon inequalities. This
minimal set, written in functional form, consists of the following
inequalities:
\begin{itemize}\setlength\itemsep{0pt}
\item[(B1)] \hypertarget{B1}{}$f(N\sm i)\le f(N)$ ~~for all $i\in N$, and
\item[(B2)] \hypertarget{B2}{}$f(aK)+f(bK)\ge f(K)+f(abK)$ ~~for all $K\subseteq N$ and different
$a,b\in N\sm K$ (including $K=\emptyset$ with $f(\emptyset)=0$).
\end{itemize}
It is worth pointing out that the non-negativity of $f$ follows from
(\bref{B1}) and (\bref{B2}) and needs not be specified separately.

\emph{Polymatroids} are the elements of $\GN$, written as functions that
assign
non-negative real numbers to non-empty subsets of $N$. In this sense
polymatroids are axiomatized by the (basic) Shannon inequalities (\bref{B1})
and (\bref{B2}). Polymatroids are discussed in more detail in Section
\ref{sec:polymatroids}.

A \emph{non-Shannon} (linear) \emph{entropy inequality} is a half-space in
the $(2^{|N|}\m-1)$-dimensional Euclidean space (the space of the entropy
region $\GaN$), which goes through the origin (so this half-space can be
specified by the coefficients of its normal), contains the complete entropy
region $\GaN$ on its non-negative side, and cuts into the polymatroidal
region $\GN$ (implying that it is not a consequence of the Shannon
inequalities). The intersection of all valid linear entropy inequalities is
the closure of the conic hull of $\GaN$.

A natural strategy to separate $\GN$ and $\GaN$ by such an inequality is the
following. Find some linear operation on the real vectors that
maps entropic vectors into entropic vectors but maps some polymatroid
outside the polymatroid region. If such an operation is found, the entropic
region sits inside the pullback of the polymatroid region along the inverse
of the operation. Since the operation is linear, the pullback of one of
the bounding hyperplanes provides a linear inequality that separates $\GN$ and
$\GaN$. In fact, this strategy works with the additional twist that $\GaN$ is
mapped to $\Ga M$ with a larger base set $M$.

\section{The entropy region}\label{sec:entreg}

This section explores various properties of the entropy region $\GaN$. Many
of the results and proofs in this section are from the paper \cite{M.twocon}
by F.~Mat\'u\v s. For additional material on Information Theory and on the
entropy region, consult one of the books \cite{cover-thomas} or
\cite{Csiszar.Korner} or \cite{yeung-book}.

For each real number $\alpha\ge 0$, there is a random variable $\xi$ such
that $\H(\xi)=\alpha$. In fact, the range the entropy of a random variable
defined on the alphabet $\mathcal X$ can take is the closed interval
$[0,\log_2 |\mathcal X|]$, and $\xi$ takes the upper bound only when it is
uniform on $\mathcal X$. From this it follows immediately that $\Ga 1$, the
1-dimensional entropy region for a single variable, is the non-negative real
half-line.

\subsection{Independence}

Two random variables $\xi$ and $\eta$ jointly distributed on $\mathcal
X\times\mathcal Y$ are \emph{independent} if for every $x\in\mathcal
X$ and every $y\in\mathcal Y$,
\begin{equation}\label{eq:indep}
  \Prob(\xi=x,\eta=y) = 
    \Prob(\xi=x)\cdot\Prob(\eta=y);
\end{equation}
this occurs if and only if $\H(\xi\eta) = \H(\xi)+ \H(\eta)$. If $\xi$ is
defined on $\mathcal X$ and $\eta$ is defined on $\mathcal Y$, then
\emph{joining them independently} yields the distribution $\xi\times\eta$ on
$\mathcal X\times \mathcal Y$ that has independent marginals identical to
$\xi$ and $\eta$. The probability mass of this distribution is determined
uniquely by the formula (\ref{eq:indep}). If both $\xi$ and $\eta$ are
jointly distributed on the same base set $N$, then $\xi\times\eta$ can also be
considered jointly distributed on the same base set. In this
case the marginal of $\xi\times\eta$ on $A\subseteq N$ is isomorphic to the
independent join of the marginals $\xi_A$ and $\eta_A$, therefore,
$$
   \H_{\xi\times\eta}(A)=\H_\xi(A)+\H_\eta(A)~~~ \mbox { for all }A\subseteq N.
$$
This observation proves the first part of the following claim.
\begin{claim}\label{claim:addition}
Both the entropy region $\GaN$ and the almost entropy region $\clGa N$ are
closed under addition.
\end{claim}

\noindent
Using Minkowski sum notation \cite{ziegler} and the fact that the all-zero vector
is clearly in $\GaN$, this claim can be expressed equivalently as $\GaN + 
\GaN = \GaN$, and $\clGa N + \clGa N = \clGa N$.

\begin{proof}
The claim for $\GaN$ was proved above. The other part follows from this part
and the continuity of vector addition.
\end{proof}

If $\xi$ is multiplied by itself $n$ times, the resulting distribution is
denoted by $\xi^n$, and is referred to as \emph{$n$ i.i.d.} (independent and
identically distributed) copies of $\xi$. Sometimes it is also called the
\emph{tensorization} of $\xi$. If $\xi=\{\xi_i:i\in N\}$ is jointly
distributed on $\mathcal X = \prod_{i\in N} \mathcal X_i$, then $\xi^n$ is a
distribution on $\mathcal X^n$, which is isomorphic to (and frequently
identified with) $\prod_{i\in N} \mathcal X^n_i$.
Therefore, $\xi^n=\{\xi^n_i:
i\in N\}$ is also jointly distributed on $N$; moreover the entropy profile
of $\xi^n$ is the entropy profile of $\xi$ multiplied by $n$.

If $\xi$ and $\eta$ are quasi-uniform, then $\xi\times\eta$ is
quasi-uniform as well. Additionally, if both $\xi$ and $\eta$ are
group-generated by the subgroups $G_i\subgr G$ and $G'_i\subgr G'$,
respectively, then $\xi\times\eta$ is group-generated by the subgroups
$G_i\times G'_i \subgr G\times G'$, where $G_i\times G'_i$ is the usual
group direct product. Claim \ref{claim:qu} below summarizes these facts.

\begin{claim}\label{claim:qu}
Suppose $\xi$ is jointly distributed on $N$, and $\xi^n$ is $n$
i.i.d.~copies of $\xi$.
\begin{itemize}\setlength\itemsep{0pt}
\item[\upshape(i)] For all $A\subseteq N$, $\H_{\xi^n}(A) = n\cdot\H_\xi(A)$.

\item[\upshape(ii)] If $\xi$ is quasi-uniform (or if it is group-generated),
then so is $\xi^n$.
\end{itemize}
\end{claim}

In Claim \ref{claim:timsum}, functional notation is used: $f(A)$
denotes the coordinate of the vector $f$ indexed by $A$.

\begin{claim}\label{claim:timsum}
Suppose $N$ and $M$ are disjoint base sets, $f\in\GaN$ and $g\in\Ga M$.
Then the \emph{direct sum} $f\oplus g$ is also entropic, where $f\oplus g$
is defined on non-empty subsets of the disjoint union $N\cup M$ as
$$
   f\oplus g: A\mapsto f(A\cap N) + g(A\cap M) ~~~A\subseteq N\cup M.
$$
Additionally, if $f,g\in\clGa N$, then $f\oplus g\in\clGa {N\cup M}$.
\end{claim}
\begin{proof}

Choose the distributions $\xi$ and $\eta$ with the entropy profiles $f$ and
$g$, respectively. The independent join $\xi\times\eta$ can be considered to
produce jointly distributed variables on the base set $N\cup M$. The
function $f\oplus g$ expresses that the marginals of $\xi\times\eta$ are
just $\xi$ and $\eta$, and that they are independent.
The claim about closure follows from the fact that the operation
$f\oplus g$ is continuous in $f$ and $g$.
\end{proof}

\subsection{Restricting, factoring, duplicating}\label{subsec:rest-fact-dupl}

\emph{Restricting} a distribution to a subset of variables means
discarding the rest and replacing the probabilities with their marginal values.
\emph{Factoring} is the process when a group of variables is taken as (or
grouped into) a single variable. Finally \emph{duplicating} means adding an
identical copy of one of the variables so that the duplicate always takes
the same value as the original.

These operations on distributions have their corresponding operations on the
entropy region since the entropy of the result depends only on the entropy
of the operands and not on the particular distribution. If $f$ is the
entropy profile of a distribution $\xi$ on $N$ and the distribution is
restricted to the smaller base $M\subseteq N$, then the entropy profile of
the restricted distribution will be denoted by $f\restr M$. Clearly,
$(f\restr M)(A)=f(A)$ for all $A\subseteq M$, and this value does not depend
on the particular distribution $\xi$. Using $f\restr M$ is a slight abuse of
notation as $f$ is restricted not to $M$ but to the (non-empty) subsets of
$M$. The restriction operation is also known as \emph{deleting} the
complement of $M$, and the corresponding notation is $f\del (N\sm M)$.

\emph{Factoring} requires an equivalence relation $\sim$ on $N$; this
relation defines the groups of variables that are to be merged. When lifted
to the entropy profile $f$, the factored profile, denoted $f/{\sim}$, is
defined on subsets of the equivalence classes, and on a subset of
equivalence classes it takes the value that $f$ would do on their union.
Again, the result clearly depends on the profile only and not on the
representing distribution.

\emph{Duplicating at index $a\in N$} is the result of adding a ``twin'' of $a$
to $N$, named $a'$, and extending the base set to $N\cup\{a'\}$. If the
original entropy profile is $f$ and the extended entropy profile is
$f^\sharp_a$, then, as $a$ and $a'$ always take the same value,
$$\begin{array}{l@{\,=\,}l}
   f^\sharp_a(A) & f(A),\\[2pt]
   f^\sharp_a(a'A) & f(aA)
\end{array} ~~~~ \mbox{ for all subsets $A\subseteq N$,}
$$
independently of the representing distribution.
In particular, $f^\sharp_a(a)=f^\sharp_a(a')=f^\sharp_a(aa')=f(a)$.

\begin{claim}\label{claim:Gaoper1}
{\upshape(i)} \hangindent\parindent\hangafter1
Suppose $M$ is a subset of $N$. If $f\in\GaN$, then $f\restr
M\in \Ga M$. Similarly, if $f\in\clGa N$, then $f\restr M \in \clGa M$.

\noindent{\upshape(ii)}
Suppose $\sim$ is an equivalence relation on $N$; let $M = N/{\sim}$ be
the set of equivalence classes. If $f\in\GaN$, then $f/{\sim}\in\Ga M$.
Similarly, if $f\in\clGa N$, then $f/{\sim}\in\clGa M$.

\noindent{\upshape(iii)}
Let $f\in\GaN$, $a\in N$, and $f^\sharp_a$ be the function on $M=N\cup\{a'\}$ 
which duplicates $f$ at $a$. In this case $f^\sharp_a\in\Ga M$. If $f\in\clGa N$, then
$f^\sharp_a\in\clGa M$.
\end{claim}
\begin{proof}
Claims regarding the entropy region follow from the discussion above. The claims
are also true for the aent region as each operation is continuous.
\end{proof}

\subsection{The aent region is a closed cone}\label{subsec:aent-cone}

For $J\subseteq N$ let $\r_J$ be the indicator vector of the
relation $J\cap A\ne\emptyset$, that is,
$$
   \r_J(A) = \begin{cases}
                1 & \mbox{if $J\cap A\ne\emptyset$,}\\
                0 & \mbox{otherwise.}
             \end{cases}
$$
Notice that $\r_{\emptyset}$ is the all-zero vector.
\begin{lemma}\label{lemma:1}
{\upshape(i)} \hangindent\parindent\hangafter1
$\lambda\, \r_J$ is entropic for each $\lambda\ge 0$.

\noindent{\upshape(ii)} The vectors $\{\r_J : \emptyset\ne J\subseteq N\}$ are linearly
independent, thus span the complete $(2^{|N|}\m-1)$-dimensional Euclidean
space.
\end{lemma}
\begin{proof}
(i) Take a random variable $\xi$ on $\mathcal X$ with $\H(\xi)=\lambda$, and
fix some $x_0\in\mathcal X$. Define the jointly distributed variables $\{\xi_i:i
\in N\}$ as follows:
$$
    \xi_i \eqdef \begin{cases} \xi &\mbox{if $i\in J$,}\\
                          x_0 &\mbox{otherwise}.
            \end{cases}
$$
Then $\xi_A$ is identical to $\xi$ if $A\cap J\ne\emptyset$, thus has
entropy $\lambda$. If $A\cap J=\emptyset$, then $\xi_A$ is identically
$x_0$, thus has entropy $0$.

\smallskip\noindent (ii)
This statement can be proved, for example, by induction on the size of the
base set $N$. Another possibility is to observe that
$$
   \r_N - \r_{N\sm J} = \sum_{A\subseteq J} e_A
$$
where $e_A$ is the unit vector with 1 at coordinate $A$ and zero
everywhere else. The M\"obius inversion formula \cite{mobius.inv} applied to
the Boolean lattice gives
$$
   e_J = \sum_{A\subseteq J} (-1)^{|J\sm A|}\,(\r_N-\r_{N\sm A})
$$
for each $J\subseteq N$, 
showing that each unit vector is a linear combination of the vectors $\r_J$.
Since their number equals the dimension of the space, they must also be
linearly independent.
\end{proof}

\begin{lemma}\label{lemma:2}
Suppose $\lambda\ge 0$ and $\eps>0$. If $\mathbf x$ is entropic, then so is
$\lambda\mathbf x+\eps\tsp\r_N$. That is, $\mathbf x\in\GaN$ implies
$\lambda\mathbf x + \eps\,\r_N\in\GaN$.
\end{lemma}
\begin{proof}

Suppose first that $\xi=\{\xi_i\}$ is jointly distributed on the alphabet
$\prod \{\mathcal X_i:i\in N\}$, $\mathcal Y_i=\mathcal X_i\cup\{z\}$ where
$z$ is a new element shared by all $\mathcal Y_i$, and $0\le p<1$. Define
the variables $\eta=\{\eta_i\}$ jointly distributed on $\prod \{\mathcal
Y_i:i\in N\}$ as follows. Let $\eta=\xi$ with probability $p$, and let all
$\eta_i=z$ with probability $(1-p)$. For each $A\subseteq N$ we have
$$
  \H_\eta(A) = \big( p\cdot\H_\xi(A) - p \log_2 p \big)
           - (1-p)\log_2(1-p),
$$
where the first term comes from the cases when $\eta=\xi$ as every
probability there is multiplied by $p$, and the second term is the case when 
all $\eta_i$ are equal to $z$. This construction proves that if $\mathbf x$
is entropic, then so is $p\cdot \mathbf x + h(p)\cdot\r_N$ where
$$
    h(p)=-p\log_2p - (1-p)\log_2(1-p)
$$
is the binary entropy function. Since $h(p)$ is non-negative, continuous
and $h(0)=0$,
there is an integer $n>\lambda$ such that $h(\lambda/n)<\eps$. Let
$p=\lambda/n$, then
$$
  \lambda\mathbf x+\eps\tsp\r_N=\big(p\cdot(n\tsp\mathbf x)
      + h(p)\cdot\r_N \big)+ (\eps-h(p))\r_N.
$$
Since $\GaN$ is closed for addition by Claim \ref{claim:addition},
$n\tsp\mathbf x\in\GaN$, thus, by the remark above, the first term on the
right hand side is entropic as well. Lemma \ref{lemma:1} (i) and
$\eps-h(p)>0$ implies that the second term is also in $\GaN$, proving the
lemma.
\end{proof}

Lemma \ref{lemma:2} does not remain true for $\eps=0$. For a counterexample take the
following distribution on three variables $N=\{a,b,c\}$. Fix the integer $n\ge
2$. Each variable takes values from zero to $n-1$ uniformly so that the sum
$a+b+c$ is divisible by $n$. Then any two of the variables are independent, and any
two determine the third. Since the entropy of the uniform distribution on
$n$ elements is $\log_2 n$, entropies of this distribution are
\begin{align*}
   & \H(a)=\H(b)=\H(c)=\log_2 n, \\
   & \H(ab)=\H(ac)=\H(bc)=\H(abc)=2\log_2 n.
\end{align*}
Thus the entropy profile of this distribution is $(\log_2 n)\mathbf u$,
where
\begin{equation}\label{eq:u}
\mathbf u(J)=\min\{|J|,2\} ~~~ \mbox{ for $J\subseteq N$}.
\end{equation}

\begin{claim}\label{claim:u}
If $\lambda$ is not of the form $\log_2 n$ for some integer $n\ge 2$, then
$\lambda\tsp\mathbf u$ is not entropic.
\end{claim}
\begin{proof}

Take a distribution on $N=\{a,b,c\}$ with entropy profile $\lambda\tsp\mathbf u$. Then
any two of the variables $a$, $b$, $c$ are independent and any two determine
the third one. Denote the alphabets of the variables by $\mathcal A$,
$\mathcal B$, and $\mathcal C$, respectively, keeping only those elements
that occur with positive probability. Create a matrix $\mathbb M$ with rows
indexed by $\mathcal A\cup\{{*}\}$ and columns indexed by $\mathcal
B\cup\{{*}\}$. The matrix entry $\mathbb M[a,b]$ in the row marked by $a$ and
in the column marked by $b$ is the list
$$
  \langle \Prob(a,b,c): c\in \mathcal C \rangle,
$$ 
where $*$ is interpreted as taking the marginal. For example, $\mathbb
M[{*},{*}]$ is the marginal distribution of $c$. Since $b$ and $c$ are
independent, $\mathbb M[{*},b]$ equals $\mathbb M[{*},{*}]$ multiplied by
$\Prob(b)$, and $\mathbb M[a,{*}]$ is the multiple of $\mathbb M[{*},{*}]$
by $\Prob(a)$. In addition, the sum of the probabilities in the list $\mathbb M[a,b]$ is
$\Prob(a,b)$, which is not zero since $a$ and $b$ are independent.

By assumption, $a$ and $b$ together determine the value of $c$. This means
that the list at $\mathbb M[a,b]$ contains exactly one non-zero probability.
Furthermore, since $a$ and $c$ together determine the value of $b$, the same
non-zero position $c$ cannot occur twice in the row indexed by $a$. The sum
of lists in row $a$ is the list at $\mathbb M[a,{*}]$ with all non-zero
probabilities, thus in row $a$ all non-zero positions occur exactly once. By
symmetry, the same holds for each column $b$.

Suppose that the only non-zero entry in $\mathbb M[a,b]$ is labeled by $c$. It is
the same value that appears in $\mathbb M[a,{*}]$ at position $c$; thus it
is $\Prob(a)\Prob(c)$. Similarly, the same number occurs at position $c$ in
$\mathbb M[{*},b]$, thus it also equals $\Prob(b)\Prob(c)$. Therefore
$\Prob(a)=\Prob(b)$ for any $a$ and $b$, which means that all these marginal
probabilities are the same, that is, $a$ takes every value with the same
probability. Thus $\H(a)=\log_2 |\mathcal A|$, which is a logarithm of an
integer number.
\end{proof}

\begin{theorem}\label{claim:3}
$\clGa N$, the closure of $\GaN$, is a full-dimensional convex cone.
Furthermore, internal points of $\clGa N$ are entropic.
\end{theorem}
\begin{proof}
Since non-negative multiples of the vector $\r_A$ are entropic, 
their conic (that is, non-negative linear) combination is also entropic.
The vectors $\r_A$ span the entire space by Lemma
\ref{lemma:1}, thus $\GaN$, and then $\clGa N$ is full-dimensional. To show
that $\clGa N$ is closed for conic combinations, choose $\mathbf x$ and $\mathbf
y$ from $\clGa N$ as well as the non-negative real numbers $\lambda$, $\mu$. There are
entropic points $\mathbf x'$ and $\mathbf y'$ arbitrarily close to $\mathbf
x$ and $\mathbf y$, respectively. By Lemma \ref{lemma:2} there are entropic
points arbitrarily close to $\lambda \mathbf x'+\mu\mathbf y'$, therefore
$\lambda \mathbf x + \mu\mathbf y$ is indeed in the closure of the entropic
points.

For the last statement let $\mathbf y$ be an internal point of $\clGa N$,
which means that an open neighborhood of $\mathbf y$ is also in $\clGa N$. In
particular, there is a positive $\eps$ such that the point
$$
   \mathbf y'= \mathbf y - \sum_{A\subseteq N} \eps\tsp \r_A
$$
is also in $\clGa N$, which means that there are entropic points arbitrarily
close to $\mathbf y'$. Since the vectors $\{\r_A:A\subseteq N\}$ span the
whole space, any vector can be written as their linear combination. So there
is an entropic point $\mathbf x\in\GaN$ such that
$$
   \mathbf x - \mathbf y' = \sum_{A\subseteq N} \mu_A \r_A
$$
where all the scalars $\mu_A$ have an absolute value less than $\eps/2$. Then
$$
   \mathbf y = \mathbf x + \sum_{A\subseteq N} (\eps-\mu_A)\r_A
$$
is a sum of entropic points by (i) of Lemma \ref{lemma:1} (as $\eps-\mu_A$
is positive), therefore $\mathbf y$ is also entropic.
\end{proof}

\subsection{The case of two, three and four variables}\label{subsec:case234}

The entropy region $\Ga 1$ of a single variable is one-dimensional; it is
the non-negative half-line; and it coincides with its closure $\clGa 1$. 
For the two-element base set $N=\{a,b\}$ the region $\Ga 2$ sits in the
three-dimensional space with coordinates labeled $a$, $b$, $ab$,
written in this order. The vectors $\r_J$ are
\begin{align*}
    \r_a &= \langle 1,0,1\rangle, \\
    \r_b &= \langle 0,1,1\rangle, \\
    \r_{ab} &= \langle 1,1,1\rangle.
\end{align*}
By Lemma \ref{lemma:1} any non-negative linear combination of these vectors
is in $\Ga 2$, therefore $\Ga 2$ contains the (infinite) cone $\mathcal C_2$,
which has its tip at the origin and contains all rays (half-lines) starting
at the origin and having a point at the boundary or inside the triangle
marked by the points $\r_a$, $\r_b$, and $\r_{ab}$.

On the other hand, the entropy region is bounded by the basic Shannon
inequalities (\bref{B1}) and (\bref{B2}). For the point $\langle x,y,z\rangle\in \Ga 2$
they give the restrictions
\begin{align*}
   x\le z, ~ y\le z & ~~\mbox{ from (\bref{B1})},\\
   x+y\ge z          &~~\mbox{ from (\bref{B2})}.
\end{align*}
Each of these inequalities specifies a half-space bounded by the plane
with points where equality holds, and $\Ga 2$ is on the non-negative side of
that plane. Each plane contains the origin, the first plane contains
points $\r_a$ and $\r_{ab}$, thus it coincides with one of the faces of
$\mathcal C_2$. Similarly, the second plane contains $\r_b$ and $\r_{ab}$,
the third one contains $\r_a$ and $\r_b$. Thus $\Ga 2$ is bounded by the
three faces of the cone $\mathcal C_2$, which means that the two-variable
entropy region is the cone $\mathcal C_2$. Since $\mathcal C_2$ is closed,
we have $\clGa 2=\Ga 2=\mathcal C_2$. If
$\Ga 2$ were plotted in the coordinate system with axes $\r_a$, $\r_b$ and
$\r_{ab}$, it would be just the non-negative orthant.

The entropy region $\Ga 3$ of three variables $N=\{a,b,c\}$ is seven-dimensional,
and contains the (closed) cone $\mathcal C_3$ spanned by the seven vectors
$\r_J$, $\emptyset\ne J\subseteq N$. $\Ga 3$ also contains the vector
$\mathbf u$ defined in (\ref{eq:u}) (and also its $\log_2 n$ multiples), which
is outside $\mathcal C_3$. This
is because writing $\mathbf u$ as the unique linear combination of the vectors 
$\r_j$, not all coefficients are non-negative:
$$
   \mathbf u = \r_{ab}+\r_{ac}+\r_{bc}-\r_{abc}.
$$
Thus $\clGa 3$, as a convex cone by Theorem \ref{claim:3}, contains the cone
spanned by the eight vectors $\r_J$ and $\mathbf u$. Similarly to the
two-variable case, the basic Shannon inequalities provide an outer bound for
$\clGa 3$. There are three inequalities in (\bref{B1}) and six inequalities
in (\bref{B2}), determining a total of nine half-spaces whose intersection
is the Shannon-bound. As we are in the seven-dimensional space, any six of
the bounding hyperplanes intersect in a line. If there is a point on this
line, different from the origin, which satisfies the remaining three
inequalities (as six out of the nine inequalities takes zero), then 
the corresponding half-line, starting from the origin, is
an \emph{extremal ray} of the Shannon-bound; and that bound is the convex
hull of its extremal rays \cite{ziegler}. It is relatively easy to check
that these extremal rays are just $\r_J$ and $\mathbf u$. Thus $\clGa 3$ is
this cone, and $\Ga 3$ misses only boundary points from $\clGa 3$. The exact
structure of $\Ga 3$ on the boundary is not known, but there are partial
results, see \cite{2faces,M.piece,thakor}.

\begin{table}[!ht]
\centering\begin{tabular}{|rrrrrrr|}
\hline
\strut$|N|$ & \footnotesize\it dim
     & \footnotesize\it   symm 
     & \footnotesize\it  ineq
     & \footnotesize\it ineq/symm 
     & \footnotesize\it rays ~~ 
     & \footnotesize\it ray/symm \\
\hline
\strut
   2     ~  &  3~  &  2 ~  &  3~   &  2 ~~~~      &  3  &   2  ~   \\
   3     ~  &  7~  &  6 ~  &  9~   &  3 ~~~~      &  8  &   4  ~   \\
   4     ~  & 15~  & 24 ~  & 28~   &  4 ~~~~      & 41  &  11  ~   \\
   5     ~  & 31~  & 120 ~  & 85~   & 5 ~~~~      & 117,983 &  1320 ~  \\
   6     ~  & 63~  & 720 ~  &246~   & 6 ~~~~ & ${>}4{\cdot}10^{11}$ &
   ${>}5{\cdot}10^9$ ~ \\
\hline
\end{tabular}
\caption{Dimension and symmetry classes of $\GN$.}\label{table:1}
\end{table}

Permuting elements of the base set $N$ induces a permutation on the
coordinates of the entropy region, and provides symmetries of both $\GaN$
and $\GN$. Table \ref{table:1} lists the dimensions and the number of
symmetries for $|N|\le 6$. The column \emph{ineq} shows the number of basic
Shannon inequalities, which is the same as the number of bounding
hyperplanes of $\GN$. The number of extremal rays of $\GN$ is in column
\emph{rays}; the exact number is known only for $|N|\le 5$ \cite{studeny-kocka},
the lower estimate for $|N|=6$ is from \cite{impossible}. Upper
estimates are doubly exponential in $|N|$. Each symmetry moves a
bounding plane into a bounding plane and an extremal ray into an extremal
ray. The columns \emph{ineq/symm} and \emph{ray/symm} show the number of
symmetry classes.

According to Table \ref{table:1}, in the case $N=\{abcd\}$, the outer bound
$\Gamma_{\!4}$ has 41 extremal rays in 11 permutationally equivalent
classes. Each of these rays, with the exception of $6$ rays (in one symmetry
class), has entropic points; thus these rays are in $\clGa 4$. The
exceptional rays are the non-negative multiples of the so-called
\emph{V\'amos vector}
\begin{equation}\label{eq:vamos}
   \mathbf v_{cd}: J\mapsto \begin{cases}
     4 & \mbox{ if $J=\{cd\}$}, \\
    \min\{4,|J|+1\} &\mbox{ otherwise,}
   \end{cases}
\end{equation}
and its symmetric versions. They are called so because these polymatroids are
factors of the so-called V\'amos matroid \cite{oxley}. Similarly to the
proof of Claim \ref{claim:u}, one can show that no scalar multiple of
$\mathbf v_{cd}$ is entropic. The long-standing open problem whether
$\mathbf v_{cd}$ is
almost entropic was settled negatively by the first non-Shannon entropy
inequality by Zhang and Yeung \cite{ZhY.ineq}. From this fact it follows
that $\clGa N$ is a proper subset of $\GN$ for every $|N|\ge 4$. Later
F.~Mat\'u\v s proved the stronger result \cite{M.infinf} that for $|N|\ge 4$
the region $\clGa N$ is not polyhedral, meaning that it cannot be characterized by
finitely many linear inequalities.

\subsection{Conditioning}\label{subsec:cond}

Conditioning is an important operation in Information Theory. Given a subset
$K\subseteq N$, the conditional information measures the average information
content of $A\subseteq
N\sm K$ as variables in $K$ run over the alphabet $\mathcal X_K$. More
precisely, the conditional entropy of $A$ is the weighted average of the
entropies of the conditional distributions $\xi_A \| \xi_K=z$ as $z$ runs
over the elements of $\mathcal X_K$:
\begin{equation}\label{eq:cond}
\H_\xi(A\|K) \eqdef
   \sum_{z\in\mathcal X_K} \Prob(\xi_K=z)\tsp 
    \H_{\xi_A\|\xi_K=z}(A).
\end{equation}
A simple calculation shows that this value equals $\H_\xi(AK)-\H_\xi(A)$,
in accordance with how conditional entropy was defined earlier. This
conditioning gives rise to an operation on entropy functions.
Suppose $f$ is defined on subsets of $N$, $K\subseteq N$ and $M=N\sm K$ is
the complement of $K$. The operation of \emph{conditioning $f$ on $K$} results
in the function $\fc_M$ defined on subsets of $M$ as
$$
   \fc_M : A \mapsto f(A\|K) ~~~ \mbox{ for } A \subseteq M.
$$
\begin{claim}\label{claim:conditioning}
Suppose $N$ is partitioned into the non-empty sets $K\cups M$, and $f\in\clGa
N$ is almost entropic. Conditioning $f$ on $K$ gives the almost entropic
$\fc_M\in\clGa M$.
\end{claim}
\begin{proof}
Since conditioning is a continuous operation, it suffices to show
that conditioning an entropic function results in an almost entropic
function. First, remark that the result of conditioning by (\ref{eq:cond})
does not depend on the distribution $\xi$, only on its entropy profile. So
assume that $f$ is the entropic function of the distribution $\xi$. 
Then for each $z\in \mathcal X_K$
$$
   f_z: A \mapsto \H_{\xi_A\|\xi_K=z}(A), ~~~~ A \subseteq M
$$
is the entropy profile of the conditional distribution $\xi_M
\|\xi_K{=}z$; thus it is in $\Ga M$. Since $\fc_M$ is a convex linear
combination of the functions $f_z$ by (\ref{eq:cond}) and $\clGa M\supseteq
\Ga M$ is convex by Theorem \ref{claim:3}, we are done.
\end{proof}

Conditioning an entropic function can result in a non-entropic function
(but will be almost entropic by Claim \ref{claim:conditioning}). For an
example, take the
following distribution on $N=\{a,b,c,d\}$. The variable $d$ takes values
$0$ and $1$ with equal probability. If $d=0$ then $\{a,b,c\}$ is
distributed so that its entropy profile is $\mathbf u$ as defined
in (\ref{eq:u}), and if $d=1$ then its entropy profile is $2\cdot\mathbf u$. 
When conditioned on $d$, the function $\fc_{abc}$ is the average $(3/2)\mathbf
u$, which is not entropic according to Claim \ref{claim:u}.

\subsection{Group-based distributions}\label{subsec:group-based}

The celebrated theorem of Chan and Yeung \cite{Chan.Yeung} can be
interpreted by saying that entropy functions arising from group-based
distributions are ``typical''. In many cases if a property holds for those
``typical'' functions, then it also holds for all entropy (and aent)
functions, thus one can rely on properties of distributions that provide
typical functions. We will see such an application in Section
\ref{subsec:fmpe}.

\begin{theorem}[Chan--Yeung]\label{thm:group}
Non-negative multiples of entropy functions arising from group-based distributions
(and, consequently, from quasi-uniform distributions)
form a dense subset of $\clGa N$.
\end{theorem}

\begin{proof}

Let $\xi = \{\xi_i:i\in N\}$ be a joint distribution on the alphabet
$\mathcal X=\prod \{\mathcal X_i:i\in N\}$. Approximate the probabilities by
rational numbers and write every probability mass as a rational number with
the same denominator. Take a set $\Omega$ with as many elements as this
denominator, and take the uniform distribution on it. Consider $\xi$ as a
map from $\Omega$ to $\mathcal X$ which, for any particular $x\in\mathcal
X$, maps so many elements of $\Omega$ to $x$ that their total mass adds up
to $\Prob(x)$. For each $i\in N$ the variable $\xi_i$ determines a partition
of $\Omega$ that group the elements assigned by $\xi_i$ to the same
value in $\mathcal X_i$. The set $\Omega$ with these $N$ partitions is the
\emph{partition representation} of the distribution $\xi$, see
\cite{M.partition}.

Replace each point in $\Omega$ by the same number of points to get a set of
$T$ points, and let $G_i$ be the set of all permutations on these $T$
points that respect the partition induced by the
random variable $\xi_i$. Clearly, $G_i$ is a subgroup of the permutation
group $G$ on these $T$ points. If the set sizes in the partition are $p_1$,
$\dots$, $p_k$ such that $p_1+\cdots+p_k=T$, then
\begin{align*}
 \frac{|G_i|}{|G|} &= \frac{p_1! p_2!\cdots p_k!}{T!} = 
    (1+o(1)) \Big(\frac{p_1}{T}\Big)^{p_1} \cdots
       \Big(\frac{p_k}{T}\Big)^{p_k} ={} \\
       &= (1+o(1))\,2^{-T\H(\xi_i)}.
\end{align*}
As permutations in $G$ that respect the partition induced by $\xi_A$ are
just the permutations that fix the partitions induced by $\xi_i$ for all $i\in
A$, we have similarly
$$
 \frac{|G_A|}{|G|} = (1+o(1))\,2^{-T\H(\xi_A)}.
$$
Let the distribution $\eta$ be based on $G$ and the subgroups $G_i$.
The support of $\eta_A$ has exactly $|G|/|G_A|$ elements, each having the same
probability mass, therefore $\H(\eta_A)=\log_2(|G|/|G_A|)$. Thus
$$
   \frac1T \H(\eta_A)=\frac1T \log_2 \frac{|G|}{|G_A|} = \H(\xi_A)-
\frac{\log_2(1+o(1))}{T}
$$
for all $A\subseteq N$, which proves that the entropy profile of $\xi$ can
be arbitrarily approximated by multiples of entropy functions originating
from group-based distributions.
\end{proof}

\subsection{Thinning and compressing}\label{subsec:fmpe}

Let $\{\xi_i:i\in N\}$ be a quasi-uniform distribution where each alphabet
$\mathcal X_i$ is ``large enough''. Fix the set $Z\subset N$ of the
variables; the subset $A\subseteq N$, however, will vary. Let $\mathcal Z$
be the set of elements from $\mathcal X_Z$ that $Z$ can take with positive
probability, and define $\mathcal A$ similarly. Since the distribution is
quasi-uniform, each element of $\mathcal Z$ has probability $1/|\mathcal
Z|$, and each element of $\mathcal A$ has probability $1/|\mathcal A|$.
Consequently, $\H(Z)=\log_2|\mathcal Z|$, and $\H(A)=\log_2|\mathcal A|$, or
after exponentiation, $|\mathcal Z|=2^{\H(Z)}$ and $|\mathcal
A|=2^{\H(A)}$.

Consider the matrix $\mathbb M$ with columns labeled by $\mathcal Z$ and
rows labeled by $\mathcal A$, where the entry at $\mathbb M[a,z]$ is the
probability $\Prob(\xi_A=a$, $\xi_Z=z)$. Quasi-uniformity implies that
each matrix entry is either zero or $2^{-\H(AZ)}$ (to account for the
total entropy of $AZ$). The probabilities in a row add up to $1/|\mathcal 
A|=2^{-\H(A)}$ (the
probability of that row), and in a column they add up to $2^{-\H(Z)}$.
Therefore, of the $|\mathcal Z| = 2^{\H(Z)}$ many elements in a row exactly
$$
  \frac{2^{-\H(A)}}{~2^{-\H(AZ)}~} = 2^{\H(Z\|A)}
$$ 
are non-zero, and, similarly, of the $|\mathcal A|=2^{\H(A)}$ many
elements in a column exactly $2^{\H(A\|Z)}$ are non-zero.

\medskip

The operation \emph{thinning} creates a new distribution by randomly choosing
$2^\alpha$ many columns from $\mathbb M$ and then renormalizing the
probabilities. Consider the row labeled with $a\in\mathcal A$. When choosing a
column randomly, the probability that this row contains a non-zero value in
that column is
$$
  \frac{\mbox{ non-zero positions }}{\mbox{total length}} = 
  \frac{\; 2^{\H(Z\|A)}\;}{2^{\H(Z)}} = 2^{-\I(A,Z)}.
$$
After choosing $2^\alpha$ columns randomly, we expect $2^\alpha\cdot
2^{-\I(A,Z)}$ non-zero entries in this row. If $\alpha>\I(A,Z)$, then we
expect this many non-zero entries in each row; therefore, after
normalization, every row will have equal probability, and then the new
entropy of the rows will remain $\log_2|\mathcal A| =
\H(A)=\I(A,Z)+\H(A\|Z)$.

If, however, $\alpha<\I(A,Z)$, then we expect that each row contains at most
one non-zero entry in these $2^\alpha$ columns. Each column contains
$2^{\H(A\|Z)}$ non-zero entries, for different columns these are in
different rows. Therefore $2^\alpha\cdot 2^{\H(A\|Z)}$ rows have (equal)
non-zero probabilities (other rows have zero probability), and then the new
entropy of the rows is the $\log_2$ of this value, that is, $\alpha+\H(A\|Z)$.

Combining the two cases, the new entropy of the distribution in the
rows after thinning will be $\min\{\H(A),\alpha+\H(A\|Z)\}$. Since thinning
is performed on the alphabet $\mathcal Z$ and the process itself is oblivious to
what the subset $A\subseteq N$ is, thinning gives a new distribution where
the entropy of $A$ is (approximately) the above value for every $A\subseteq
N$. This heuristic can be made precise (see \bref{Appendix}). Since
quasi-uniform distributions are ``typical'' by Theorem \ref{thm:group},
it proves that the \emph{Generalized Ahlswede-K\"orner} operation, 
abbreviated as $\GAK$ and defined below, preserves the aent functions.

\begin{theorem}\label{thm:AK}
Suppose $Z\subseteq N$, $\alpha\ge0$, and $f\in\clGa N$ is aent. Then the
function $f^*=\GAK(f)$, defined as
$$
   f^*: A \mapsto \min\{ f(A),\alpha+f(A\|Z)\} ~~~\mbox{ for } A\subseteq N,
$$
is also aent.\qed
\end{theorem}

A special case of Theorem \ref{thm:AK} is known as the
\emph{Ahlswede-K\"orner lemma}. This lemma is implicit in
\cite{Ahlswede.Korner} or in the book \cite{Csiszar.Korner}, and explicitly
stated in \cite{Kaced} and \cite{MMRV}. Recall that $\r_z$ is the
(entropic) indicator function of the relation $z\in A$ from Section
\ref{subsec:aent-cone}.

\begin{lemma}[Ahlswede-K\"orner]\label{lemma:AK}
If $f$ is aent, then so is $f - f(z\|N\sm z)\tsp \r_z$ whenever
$z\in N$.
\end{lemma}
\begin{proof}
Apply Theorem \ref{thm:AK} with a single element $Z=\{z\}$, and $alpha$ as
the mutual information
$$
   \alpha=f(z,N\sm z) = f(z)-f(z\|N\sm z)= f(z)+f(N\sm z)-f(N).
$$
The aent function $f^*=\GAK(f)$ can be written as
$$
   f^*: A \mapsto \begin{cases}
       f(A) & \mbox{ if $z\notin A$}, \\[3pt]
       f(A)-f(z\|N\sm z) & \mbox{ if $z\in A$},
   \end{cases}
$$
using that $f(Az)-f(A)\ge f(z\|N\sm z)$ when $z\notin A$. Since $\r_z$ is
the indicator function of $z\in A$, the function claimed to be aent is just
$f^*$.
\end{proof}

\medskip

Returning to the matrix $\mathbb M$, the dual operation of thinning is
\emph{compressing}. While keeping all variables in $N$, a new variable $z'$
is created by making $2^\alpha$ many buckets and placing each of the
$|\mathcal Z|$ columns into one of the buckets randomly. This process
compresses the columns into $2^\alpha$ many buckets. The probability mass of
a bucket is the probability mass of the columns in that bucket. The entropy
of $Az'$ is the entropy of the distribution in which columns in the same bucket
are added. Similarly as before, consider the row labeled $a\in \mathcal
A$. Since there are $2^\alpha$ many buckets, each of the $2^{\H(Z\|A)}$
non-zero elements in this row ends up in one of the buckets, thus each bucket
is expected to contain $2^{\alpha-\H(Z\|A)}$ non-zero entries from this row.
If $\alpha>\H(Z\|A)$, then it means that in each row each bucket contains
the same number of non-zero probabilities, which add up to the same
probability mass. Consequently $A$ and $z'$ are independent, and then
$\H(Az')=\H(A)+\alpha$. If $\alpha<\H(Z\|A)$, then each bucket is expected
to contain at most one non-zero element from each row. This means that all
non-zero elements in the matrix $\mathbb M$ are in different (bucket, row)
pairs; thus $\H(Az')$ is the same as the total entropy of the distribution
in $\mathbb M$, that is, $\H(AZ)$. Combining the two cases into a single formula
yields
$$
    \H(Az') = \min\{\H(A)+\alpha,\H(AZ)\}.
$$
Similarly to the thinning case, the bucketing is oblivious to what the subset
$A$ is, thus the same formula works for every $A\subseteq N$. This
heuristic leads to the following theorem of F.~Mat\'u\v s \cite{M.twocon}.
A formal proof is provided in the \bref{Appendix}.

\begin{theorem}[F. Mat\'u\v s]\label{thm:princ}
Suppose $Z\subseteq N$, $\alpha\ge 0$, and $f\in\clGa N$ is aent. Then the
extension $f'$ of $f$ to subsets of $Nz'$ is also aent, where
$$
  \begin{array}{l@{\,=\,}l}
     f'(A) & f(A), \\[3pt]
     f'(Az') & \min\{f(A)+\alpha,f(AZ)\}
  \end{array} 
   ~~~\mbox{ for } A\subseteq N.
$$
\end{theorem}
In the special case when $Z=\{z\}$ has a single element and $\alpha=f(z\|N\sm
z)$, the aent function $f'$ on subsets containing $z'$ can be computed as
$$\begin{array}{l@{\,=\,}l}
   f'(Az') & f(A)+\alpha, \\[3pt]
   f'(Azz') & f(Az)
\end{array}
  ~~~\mbox{ for } A \subseteq N\sm z,
$$
in particular, $f'(z')=\alpha$. Conditioning (see Section \ref{subsec:cond}) 
$f'$ on $z'$ gives the following aent function defined on subsets of $N$:
$$
   f^{\prime\mathsf c}_N: A \mapsto
    \begin{cases}
       f(A) & \mbox{ if $z\notin A$}, \\[3pt]
       f(A)-f(z\|N\sm z) & \mbox{ if $z\in A$},
    \end{cases}
$$
which recovers the Ahlswede-K\"orner lemma \ref{lemma:AK}.

Operations defined in Theorems \ref{thm:AK} and \ref{thm:princ} can
move entropic points into non-entropic points; however, the result is
always almost entropic.


\section{Polymatroids}\label{sec:polymatroids}

\emph{Polymatroids} are non-negative real functions defined on the
(non-empty) subsets of some finite base $N$ that satisfy the Shannon
inequalities (S1) and (S2). As the basic inequalities imply all Shannon
inequalities, polymatroids are \emph{axiomatized} by the inequalities in
(\bref{B1}) and (\bref{B2}).

Polymatroids, as their name suggests, emerged as generalizations of
\emph{matroids}, which, in turn, were designed to capture fundamental
properties of dependencies and independencies in matrices and linear spaces
\cite{oxley}. Many polymatroid terms originate from matroids. The
intimate connection between polymatroids and Shannon information was
recognized by Fujishige \cite{fujishige}. This connection was used
extensively, either implicitly or explicitly, to investigate numerous
optimization problems related to entropies. For a comprehensive overview,
see \cite{padro}. This section recalls some terminology and results for
polymatroids and investigates how these are related to the notions discussed in
Section \ref{sec:entreg}.

A polymatroid is usually defined by specifying both the function, called the
\emph{rank function}, and the base set $N$, as the tuple $(f,N)$. When convenient, the
rank of the empty set is assumed to be zero. When it is clear from the
context, the base of the polymatroid is omitted and the rank function $f$ alone
defines the polymatroid. A polymatroid is identified with the
$(2^{|N|}\m-1)$-dimensional vector (or point) consisting of the rank values.
$f$ is a polymatroid if and only if this vector, also written as $f$, is an
element of the cone $\GN$. The polymatroid is \emph{entropic} if its rank
function determines an entropic point in $\GaN$, and it is \emph{almost
entropic}, or \emph{aent}, if $f\in\clGa N$.

A \emph{matroid} is an integer valued polymatroid with the additional
requirement that the rank of each singleton is either zero or one, implying
that the rank of $A\subseteq N$ is at most $|A|$. In addition to
defining matroids by their rank functions, matroids have several
combinatorial axiomatizations as well, for details consult \cite{oxley}.
Polymatroids can be considered ``fractional'' matroids, and in contrast to
matroids, they are geometrical objects rather than combinatorial ones.

Many operations on entropic points discussed in Section \ref{sec:entreg}
apply to, or generalize to, polymatroids, and many have the corresponding
operation on matroids. However, the terminology sometimes differs.
Operations such as restricting, factoring, and duplicating in Section
\ref{subsec:rest-fact-dupl} generalize to polymatroids trivially, so we
discuss them only to introduce the terminology.

The polymatroid $(g,M)$ is a \emph{restriction} of $(f,N)$, or $(f,N)$ is an
\emph{extension} of $(g,M)$, if $M\subseteq N$ and $g(A)=f(A)$ for all
$A\subseteq M$. If $N$ has exactly one more element than $M$, then $(f,N)$
is a \emph{one-element}, or \emph{one-point} extension. The restriction of
$(f,N)$ to $M\subseteq N$ is determined uniquely by the base set $M$ and is
written as $f\restr M$. \emph{Deleting} $K\subseteq N$, written as $f\del
K$, is the same as restricting $f$ to the complement of $K$: $f\del K =
f\restr (N\sm K)$.

\emph{Factoring} requires an equivalence relation $\sim$ on $N$; the base of
the factor $f/{\sim}$ is the set of equivalence classes, and the rank is the
$f$-rank of the union of the corresponding equivalence classes. The factor
is clearly polymatroid.

The operation of duplicating the element $a\in N$ in the polymatroid $(f,N)$
is called \emph{parallel extension along $a$}; the resulting one-element
extension is denoted by $f^\sharp_a$, and the added element is denoted by $a'$.
\emph{Contracting} a polymatroid is the analogue of conditioning: the
contract of $(f,N)$ to $M=N\sm K$, is the polymatroid $(\fc_M,M)$ where
$$
   \fc_M: A\mapsto f(A\|K), ~~~ A\subseteq M,
$$
and $f(A\|K)=f(AK)-f(K)$. The contract is also written as $f\contr K$,
and it is easy to see that it is indeed a polymatroid. \emph{Minors}
of a polymatroid are obtained by repeated applications of deletions and
contracts. Actually, at most one application of one operator followed by an
optional application of the other suffices to generate all minors.

In matroid parlance, independence corresponds to \emph{modularity}. The subsets
$A$ and $B$ are \emph{modular} if $f(A)+f(B)=f(AB)+f(A\cap B)$. Disjoint
modular sets are sometimes called \emph{independent}. Setting $K=A\cap B$, $A$ and $B$
are modular if and only if $aK$ and $bK$ are modular for all $a\in A\sm K$
and $b\in B\sm K$, similarly to the entropy case.

The \emph{direct sum} $(f_0,N_0)\oplus (f_1,N_1)$ of polymatroids mimics
the direct sum of distributions. First, the base sets $N_0$ and $N_1$ are
made (or assumed to
be) disjoint. Then the base of the direct sum is $N_0\cup N_1$, and the rank
function is
$$
     g: A \mapsto f_0(A\cap N_0)+f_1(A\cap N_1), ~~~ A \subseteq N_0\cup N_1.
$$
Thus $(f_i,N_i)$ is isomorphic to the restriction of the direct sum to
$N_i$, while $N_0$ and $N_1$ are modular; these properties uniquely 
determine the direct sum. A polymatroid is \emph{connected}
if it is not the direct sum of two of its restrictions. In other words,
$(f,N)$ is connected if $N$ cannot be partitioned into two (non-empty,
disjoint) modular parts.

The subset $A\subseteq N$ is a \emph{flat}, if every subset larger than $A$ has a
strictly larger rank. In particular, the complete set $N$ is always a flat.
The intersection of flats is a flat, and the \emph{closure} of a subset is the
smallest flat containing it. Flats of a polymatroid form a lattice with $N$
as the largest flat, and the collection of all elements with zero rank
(which can be the empty set) as the minimal flat. Elements of zero rank
are called \emph{loops}.

\subsection{Principal extension}\label{subsec:principal}

The \emph{principal extension} is analogue of Mat\'u\v s' Compression Theorem
\ref{thm:princ}. In the context of matroids, a special case of this extension
was treated in \cite{lovasz} and \cite[page 270]{oxley}, while a more general case
appears in \cite{geelen} and \cite{fmpe}.

Let $(f,N)$ be a polymatroid, $Z\subseteq N$, and $\alpha\ge 0$. The
\emph{principal extension of $(f,N)$ along $Z$ by $\alpha$} is the one-point
extension $(f',Nz')$ of $(f,N)$ such that for all $A\subseteq N$,
$$\begin{array}{l@{\,=\,}l}
    f'(A) & f(A), \\ [3pt]
    f'(Az')& \min\{\,f(A)+\alpha, f(AZ)\}.
\end{array}$$

\begin{claim}
The principal extension is a polymatroid.
\end{claim}
\begin{proof}

It is clear that $f'$ is non-negative, monotone, and is an extension of
$f$. For submodularity it suffices to check the basic inequalities
\begin{equation}\label{eq:pp}
   f'(aK)+f'(bK) \ge f'(K)+f'(abK)
\end{equation}
listed in (\bref{B2}). If $z'$ is not an
element of $abK$, then (\ref{eq:pp}) holds as $(f,N)$ is a
polymatroid. If one of $a$ or $b$ is $z'$, say $b=z'$, then $f'(aK)=f(aK)$,
$f'(K)=f(K)$. Using submodularity for $f$ we have
$$\begin{array}{r@{\;}c@{\;}l@{\;\ge\;}l}
  f(aK)+f(K)+\alpha &=& f(K) + \big(f(aK)+\alpha\big) & f'(K)+f'(az'K), \\[3pt]
  f(aK)+f(KZ) &\ge& f(K)+f(aKZ) & f'(K)+f'(az'K),
\end{array}$$
which shows that $f(aK)+f'(z'K)\ge f(K)+f'(az'K)$, as required. Finally,
if $z'$ is in $K$, say $K=Lz'$, then there are four possibilities where the
minimum is taken in the sum $f'(aLz')+f'(bLz')$. Each of the cases can be
handled similarly to the above as
$$\begin{array}{r@{\;\ge\;}l}
  (f(aL)+\alpha) + (f(bL)+\alpha) & (f(L)+\alpha) + (f(abL)+\alpha),\\[3pt]
  (f(aL)+\alpha)+ f(bLZ) & (f(L)+\alpha) + f(abLZ),\\[3pt]
  f(aLZ)+f(bL)+\alpha & (f(L)+\alpha)+f(abLZ),\\[3pt]
  f(aLZ)+f(bLZ) & f(LZ)+f(abLZ),
\end{array}$$
where the inequalities follow from the submodularity of $f$. As in each case
the right hand side is at least $f'(Lz')+f'(abLz')$, submodularity holds for $f'$
as was claimed.
\end{proof}

An interesting application of principal extension is \emph{splitting}: any
base element can be split into two with specified ranks so that
the factor that merges the new elements will be isomorphic to the original.

\begin{claim}\label{claim:splitting}

Let $(f,N)$ be a polymatroid, and $a\in N$. Write the rank of $a$
as $f(a)=\alpha_0+\alpha_1$ with $\alpha_i\ge 0$. There is a
polymatroid $f'$ on $N\sm a\cup\{a_0a_1\}$ such that $f'(a_i)=\alpha_i$, and
the factor of $f'$ merging $a_0$ and $a_1$ is isomorphic to $(f,N)$.
\end{claim}
\begin{proof}
Take the principal extension of $f$ along $a$ by $\alpha_0$ to get the
polymatroid $(f_0,N\cup a_0)$, then take the principal extension of $f_0$
along $a$ by $\alpha_1$ to get $(f_1,N\cup a_0a_1)$. An easy calculation
shows that for each $A\subseteq N$ we have
$$\begin{array}{l@{\;=\;}l}
   f_1(A) & f(A), \\[3pt]
   f_1(Aa_i) & \min\{ f(A)+\alpha_i, f(Aa) \}, \mbox{ and } \\[3pt]
   f_1(Aa_0a_1) & \min\{ f(A)+\alpha_0+\alpha_1,f(Aa)\} = f(Aa),
\end{array}$$
where the last equality holds since $f(A)+f(a)\ge f(Aa)$ by submodularity.
The restriction of $f_1$ to $N\sm a\cup a_0a_1$ gives the required
splitting.
\end{proof}

Observe that each rank in the splitting polymatroid $f'$ is either equal to
some rank in the original polymatroid, or it exceeds some old rank by $\alpha_i$.
Consequently if all ranks in $(f,N)$ are integers and one of the singletons
has rank greater than one, then this singleton can be split into two with
smaller integer ranks, while keeping the polymatroid integer. Continuing the
splitting, after finitely many steps we arrive at an integer polymatroid
where all singletons have rank zero or one so that the original polymatroid
is its factor. As an integer polymatroid with singleton ranks zero or one is
a matroid, we proved the following result of T.~Helgason \cite{helgason}.

\begin{theorem}[Helgason]\label{thm:helgason}
Every integer polymatroid is a factor of a matroid.
\end{theorem}

The generalized Ahlswede-K\"orner operation $\GAK$, defined in Theorem
\ref{thm:AK}, also preserves polymatroids.
\begin{claim}\label{claim:AKpoly}
Suppose $Z\subseteq N$, $\alpha\ge 0$, and $(f,N)$ is a polymatroid. Then so
is $(f^*,N)$, where
$$
    f^*: A \mapsto \min\{ f(A),\alpha+f(A\|Z)\} ~~~\mbox{ for } A\subseteq N.
$$
\end{claim}
\begin{proof}
We show that $(f^*,N)$ is the same as a principal extension
followed by a contraction. As both operations preserve polymatroids, so does
$\GAK$. In more detail, two cases are distinguished.
If $\alpha\ge f(Z)$, then $f^*(A)=f(A)$, and there is nothing to prove.
Otherwise, take the principal extension $(f',Nz')$ along $Z$ by $\beta=f(Z)-
\alpha$, and then take the contraction of $f'$ to $N$. Since for $A\subseteq N$
$$
   f'(Az') = \min\{ f(A)+\beta, f(AZ) \},
$$
we have $f'(z')=\beta\le f(Z)$, and then the contract $f^{\prime\mathsf c}_N$ is
\begin{align*}
   f^{\prime\mathsf c}_N(A) &= f'(Az')-f'(z') = 
      \min\{f(A), f(AZ)-\beta\} ={} \\
      &=\min\{f(A),\alpha+f(A\|Z)\},
\end{align*}
equal to $f^*$, as was claimed.
\end{proof}

\subsection{Tightening}\label{subsec:tightening}

Let $(f,N)$ be a polymatroid and $z\in N$. For a real number $t$ define the
functions $f\up^z_t$ and $f\dn^z_t$ as
$$\left.\begin{array}{r@{\;=\;}l}
   f\up^z_t(A) & \min\{ f(Az),f(A)+t\}\\[5pt]
   f\dn^z_t(A) & \min\{ f(Az)-t, f(A)\}
\end{array}~~\right\}
  \mbox{ for all $A\subseteq N$.}
$$

\begin{claim}
If $0\le t$, then $(f\up^z_t,N)$ is a polymatroid. If $0\le t\le f(z)$, 
then $(f\dn^z_t,N)$ is also a polymatroid.
\end{claim}
\begin{proof}
For operation $\up$, let $(f',Nz')$ be the principal extension (Section
\ref{subsec:principal}) along
$z\in N$ by $t$, and restrict $f'$ to $N-z\cup\{z'\}$. For operation
$\dn$, observe that $t\le f(z)$ implies $f'(z')=t$, and then the contract
of $f'$ on $N$ gives $f\dn^z_t$.
\end{proof}

The \emph{private info} of $z\in N$ is $f(z\|N\sm z) = f(N)-f(N\sm z)$; this
value can be considered as the amount of information that $z$ does not share
with any other member of $N$. \emph{Tightening $f$ at $z$} means to take
away this info from $z$:
$$
   f\dn z : A \mapsto \begin{cases}
        f(A) & \mbox{ if $z\notin A$,} \\
        f(A)-f(z\|N\sm z) & \mbox{ if $z\in A$.}
   \end{cases}
$$
This function can be written equivalently as 
\begin{equation}\label{eq:tightz}
  f\dn z= f-f(z\|N\sm z)\tsp\r_z
\end{equation}
where the (entropic) vector $\r_z$ (the indicator function of the
relation $z\in A$) is defined in Section \ref{subsec:aent-cone}.
If $z\notin A$ then $f(A)\le f(Az)-f(z\|N\sm z)$ by submodularity, thus
$f\dn z = f\dn^z_t$ with $t=f(z\|N\sm z)$. This means that tightening creates
a polymatroid. $f$ is \emph{tight at $z$} if $f\dn z=f$, and $f$ is
\emph{tight} if it is tight at every $z\in N$. Clearly, $f\dn z$ is tight at
$z$, moreover tightening at different elements commutes: $f\dn z_1\dn z_2 =
f\dn z_2\dn z_1$. Consequently the \emph{tight version} of $f$,
$$
   f\dn = f\dn z_1 \dn z_2 \dn \cdots \dn z_n, ~~~ N=\{z_1,\dots,z_n\}
$$
is tight and the result does not depend on the order of the elements.
Additionally,
\begin{equation}\label{eq:tight}
    f = f\dn + \sum_{z\in N} \lambda_z \r_z,
\end{equation}
where $\lambda_z=f(z\|N\sm z)$. It is easy to show that $f\dn$ is the
unique tight polymatroid that satisfies (\ref{eq:tight}) with non-negative
$\lambda_z$ coefficients. The polymatroid $f\dn$ is the \emph{tight part} of
$f$, and the difference $f-f\dn$ is its \emph{modular part}, see
\cite{entreg}.

\bigskip

Tightening is one of the operations that preserves almost entropic
polymatroids in both directions.

\begin{claim}\label{claim:aent-tight}
A polymatroid is aent if and only if its tightened version is aent.
\end{claim}
\begin{proof}
If $f\dn$ is aent, then $f$ is the sum of aent polymatroids by
(\ref{eq:tight}), thus it is aent. In the other direction, $f\dn z$ is just
the aent polymatroid provided by the Ahlswede-K\"orner Lemma \ref{lemma:AK}.
Since $f\dn$ is the result of tightening at each variable, we are done.
\end{proof}

If $f\dn$ is entropic, then by (\ref{eq:tight}), $f$ is also entropic. In the other
direction there exists an entropic $f$ such that $f\dn$ is not entropic. As an
example let $N=\{a,b,c\}$ where each variable takes $0$, $1$ or $2$. Choose
$a$ and $b$ uniformly and independently, and let $c$ be either $a+b$ or
$a-b$ modulo 3 with probability $\sfrac12$. Then $c$ takes each value with
probability $\sfrac13$, each pair is independent, thus
$$
   f(A)=\begin{cases}
      \log_2 3 & \mbox{ if $|A|=1$,}\\[2pt]
      2\tsp\log_2 3 & \mbox { if $|A|=2$, and }\\[2pt]
      1+2\tsp\log_2 3 & \mbox{ if $A=\{abc\}$.}
\end{cases}
$$
Then $f\dn=(-1+\log_2 3)\mathbf u$, which is not entropic by Claim
\ref{claim:u}.


\section{Linear polymatroids}\label{sec:linear}

Let $\mathbb F$ be a field, and consider the $d$-dimensional vector space
over $\mathbb F$, denoted as $\mathbb F^d$. Suppose that for each $i\in N$
we have a set $V_i\subset \mathbb F^d$ of linearly independent vectors. For
$A\subseteq N$, let $f(A)$ be the rank (over the field $\mathbb F$) of the
matrix formed from the vectors in $V_A=\bigcup \{V_i:i\in A\}$ as rows; it
is the same value as the \emph{dimension} of the linear subspace spanned by
the vectors in $V_A$. Since the rank equals the maximal number of
independent vectors, we have $f(i)=|V_i|$. Clearly, $f$ is an integer
polymatroid. \emph{Linearly representable polymatroids} are the ones that
can be obtained in this way. Occasionally these polymatroids are called
\emph{multi-linear} or \emph{folded-linear}, see \cite{B.B.F.P}, to
emphasize that the vector set $V_i$, assigned to the base element $i$, may
contain several vectors. The polymatroid is \emph{$p$-representable} if it
has a vector representation where the underlying field $\mathbb F$ has
characteristics $p$.

As an example, the uniform polymatroid $\mathbf u: J \mapsto \min\{ |J|,2\}$
for all $J\subseteq N$ from (\ref{eq:u}) is linear over every large enough
field $\mathbb F$ independently of its characteristics. To show this, choose
$|N|$ two-dimensional vectors $\mathbf v_i$ from $\mathbb F^2$ so that none
of them is a multiple of another. One can do it when $\mathbb F$ has $|N|-1$
or more elements. Since any two of these vectors span $\mathbb F^2$, setting
$V_i=\{\mathbf v_i\}$ gives
$$
   \dim(V_i) =1, ~~~\dim (V_iV_j)=2, ~~~ \dim(V_J)=2 \mbox{ for $|J|\ge 3$},
$$
as required by $\mathbf u$.

The fact that a collection of vectors forms a linear representation of a
polymatroid can be expressed as an algebraic condition on the vector
coordinates. Actually, this condition says which subdeterminants of the
matrix, formed from the vectors in $V_N=\bigcup\{V_i:i\in N\}$, are zero and
which are non-zero. Consequently, the compactness theorem holds: if a
polymatroid is linearly representable over any field, then it is also
representable over a \emph{finite} field. In particular, $f$ is
$0$-representable if and only if it is $p$-representable for every large
enough prime $p$. Furthermore, since the value of a determinant does not
change when switching from a field to an extension, if a polymatroid is
representable over $\mathbb F$, then the same representation works over
every extension of $\mathbb F$.

When polymatroids $(f,N)$ and $(g,N)$ are linearly representable over the
\emph{same} field $\mathbb F$, then $f+g$ is also representable over the
same $\mathbb F$ (using the direct product of the representing vector
spaces). As a consequence, integer conic combinations of $p$-representable
polymatroids are also $p$-representable. It is so because the underlying fields
can be assumed to be finite, say $|\mathbb F_i|=p^{n_i}$, and then the
finite field on $p^{n_1 n_2\cdots}$ elements is a common extension of all
$\mathbb F_i$. The sum of polymatroids representable over fields of
different characteristic is not necessarily linearly representable over any
field \cite{pena}.

\begin{definition}[Linear polymatroids]\label{def:linear}
A polymatroid is \emph{linear} if it is in the closure of the multiplies of
linearly representable polymatroids. 
\end{definition}

Linear polymatroids do not necessarily form a cone, but contain the
closure of the conic hull of the $p$-representable polymatroids for each
$p$.

For $|N|\le 3$ every polymatroid on $N$ is linear. It is so as in these
cases every extremal ray of $\GN$ contains a polymatroid that has a linear
representation over every field. A representation of the
exceptional extremal ray $\mathbf u\in\Gamma_{\!3}$ from Section
\ref{subsec:case234} was discussed earlier. Non-negative integer
combinations of these polymatroids are also representable over every field, 
and multiples of these combinations form a dense subset of $\GN$.

Linear polymatroids allow for the extraction of \emph{common information}. An
element $z\in N$ represents the common information of $A$ and $B$ if both
$A$ and $B$ determine $z$ (meaning $f(Az)=f(A)$, $f(Bz)=f(B)$), and $f(z)$
takes the maximum possible value under this condition, which is $f(A,B)=
f(A)+f(B) -f(AB)$. The common information of $A$ and $B$ can be extracted if
the polymatroid $f$ has an extension with such an element $z$. Linear
polymatroids even have such a linear extension. Clearly it suffices to show
that this is the case when $f$ is represented by the vector sets $\{V_i:i\in
N\}$ over the field $\mathbb F$. The new element $z$ is assigned a
maximal independent set of vectors (a base) from the intersection of the linear
span of $V_A$ and the linear span of $V_B$. Using this extractability
property, linear polymatroids for $|N|=4$ and $|N|=5$ can be characterized
\cite{Ma.Stud,DFZ10}. In both cases, linear polymatroids form a closed
polyhedral cone whose extremal rays contain polymatroids that are representable over
every field. For $|N|=4$ this cone is the conic hull of the $35$ non-V\'amos
extremal rays of $\Gamma_{\!4}$ (see Section \ref{subsec:case234}). The
cone has $34$ bounding hyperplanes. Among these hyperplanes $28$ come from
the basic Shannon inequalities (\bref{B1}) and (\bref{B2}), the remaining six are
determined by the non-negativity of the \emph{Ingleton expression}
\begin{equation}\label{eq:ing}
    f[a,b,c,d\,] \eqdef -f(a,b)+f(a,b|c)+f(a,b|d)+f(c,d),
\end{equation}
and its permuted versions. They determine six hyperplanes as (\ref{eq:ing})
is symmetric for swapping $a$ and $b$ as well as swapping $c$ and $d$.

\begin{claim}\label{claim:4ing}
Let $N=\{abcd\}$. The polymatroid $f\in\GN$ is linear if and only if all six
instances of the Ingleton expression {\upshape(\ref{eq:ing})} are non-negative.
\end{claim}
\begin{proof}

As noted above, all non-V\'amos extremal rays of $\GN$ have points
linearly representable over all fields, thus all polymatroids with
non-negative Ingleton values are linear indeed. To finish the proof, it
suffices to show that a linearly representable polymatroid $f$ satisfies the
Ingleton inequality $f[a,b,c,d\,]\ge 0$; non-negativity of the other
expressions follows by symmetry. Extend $f$ by the element $z$ to
a (still linear) polymatroid so that $z$ represents the common information of
$a$ and $b$. This means $f(az)=f(a)$, $f(bz)=f(b)$, and $f(z)=f(a,b)$.
A simple calculation shows that in this case we also have
$$
   f(a,b\|z)=f(a,z\|b)=f(b,z\|a)=f(cd,z\|ab) = 0.
$$
Next, we quote a five-variable Shannon inequality from \cite{MMRV} which we
will call (\ref{eq:MMRV}) from the acronym of the authors. To simplify the notation the
function denoting the polymatroid is omitted from before the opening
brackets.
\begin{equation}\label{eq:MMRV}\tag{MMRV}
  [a,b,c,d\,]+(a,b\|z)+(a,z\|b)+(b,z\|a)+3\tsp(cd,z\|ab) \ge 0.
\end{equation}
It can be checked by an automated tool, such as \cite{ITIP}, that it is
indeed a consequence of the Shannon inequalities. Actually, (\ref{eq:MMRV})
can be rearranged to be the sum of $12$ basic inequalities from (\bref{B2}).
Plugging in the one-point extension $f$ into (\ref{eq:MMRV}) we get
$f[a,b,c,d\,]\ge 0$, as required.
\end{proof}

The magical Shannon inequality in (\ref{eq:MMRV}) was apparently pulled out
of a magic hat. Another possibility is to use an LP solver to check that the
basic Shannon inequalities (\bref{B1}) and (\bref{B2}) for the base set
$\{abcdz\}$ together with the constraints $f(az)=f(a)$, $f(bz)=f(b)$,
$f(z)=f(a,b)$ and $f[a,b,c,d\,]\le-1$ (to take care of the homogeneity of the
system) do not have a feasible solution.

The value of the Ingleton expression $[a,b,c,d\,]$ for the V\'amos polymatroid
$\mathbf v_{cd}$ is $-1$, thus $\mathbf v_{cd}$ is not linear by Claim
\ref{claim:4ing}.

The archetypal example of a distribution with a negative Ingleton value is
the \emph{ringing bells} distribution. We have two ropes, $c$ and $d$,
pulled independently. If both are pulled, bell $a$ rings, and if any of them
is pulled, then bell $b$ rings. Thus $c$ and $d$ are independent 0--1
values, $a=\max\{c,d\}$ and $b=\min\{c,d\}$. If $c$ or $d$ is fixed, then
either $a$ or $b$ is constant, thus
$$
   (a,b\|c)=0, ~~~ (a,b\|d)=0, ~~~ (c,d)=0.
$$
Since $a$ and $b$ are not independent, the Ingleton value in (\ref{eq:ing})
is negative. This means that, in general, entropic polymatroids do not allow
the extraction of common information. This fact was observed much earlier by
G\'acs and K\"orner \cite{gacs-korner}.

\begin{claim}\label{claim:lin-aent}
Linearly representable polymatroids have group-based entropic multiples.
As a consequence, every linear polymatroid is aent.
\end{claim}
\begin{proof}

Suppose $f$ is represented by the vectors $V_i\subseteq \mathbb F^d$ where
$\mathbb F$ is a finite field. Define the joint distribution
$\xi=\{\xi_i:i\in N\}$ as follows. Choose the vector $\mathbf x\in \mathbb
F^d$ randomly with uniform distribution (this can be done because there are only
finitely many elements in $\mathbb F^d$). The value of $\xi_i$ is the
sequence of the scalar products $\langle\mathbf x\cdot\mathbf y: \mathbf
y\in V_i\rangle$. Since $V_A=\bigcup\{V_i:i\in A\}$ has $f(A)$ many linearly
independent elements (this is the dimension of the linear span of the
vector set $V_A$), $\xi_A$ can take $|\mathbb F|^{f(A)}$ different values, each
with the same probability. Thus the entropy profile of $\xi$ is $(\log_2|
\mathbb F|)\tsp f$, which is a multiple of $f$. This distribution is clearly
quasi-uniform, but it is also group-based. The set
$$
   G_i=\{\mathbf z\in \mathbb F^d: \mathbf z\cdot\mathbf y=0
~~ \mbox{ for all } \mathbf y\in V_i\}
$$
is a subgroup of the additive group $G$ of the vector space $\mathbb F^d$.
For $\mathbf x\in\mathbb F^ d$ the left coset $\mathbf x+G_i$ is uniquely
represented by the scalar product sequence $\langle\mathbf x\cdot\mathbf
y:\mathbf y\in V_i\rangle$. Indeed, $\mathbf x$ and $\mathbf x+\mathbf z$
give the same sequence if and only if $\mathbf z\in G_i$, proving that this
distribution is group-based.
\end{proof}

\begin{claim}\label{claim:lin-closed}
Linear polymatroids are closed under the following operations:
deleting, contracting, factoring, principal extension, the $\GAK$
operation, and tightening.
\end{claim}
\begin{proof}

All operations are homogeneous and continuous, thus it suffices to show that
they result in a (limit of multiplies of) linearly representable
polymatroid(s) when applied to a linearly representable polymatroid.

Deleting and factoring a linearly representable polymatroid is clearly
linearly representable. By Claim \ref{claim:AKpoly} the $\GAK$ operation is
equivalent to a principal extension followed by a contraction. Tightening is a
special case of the $\GAK$ operation. Thus we need to show that contraction
and principal extension preserve linear polymatroids.

For the contraction, suppose that the base is partitioned as $N=M\cups K$, and
assume that $(f,N)$ is represented by the vector sets $\{V_i:i\in N\}$. Choose
$S\subseteq V_K=\bigcup\{V_i:i\in K\}$ consisting of $f(K)$ independent
vectors, and let $T$ be another set such that $S\cup T$ is a base of the
complete vector space $\mathbb F^d$. Every vector $\mathbf x$ can be written
uniquely as the sum $\mathbf x_S+\mathbf x_T$, where $\mathbf x_S$ is in the
linear span of $S$ and $\mathbf x_T$ is in the linear span of $T$. For $i\in
M$ let $V^{\mathsf c}_i=\{\mathbf x_T: \mathbf x\in V_{iK}\}$. We claim that
it is a linear representation of the contraction. Since for $A\subseteq M$
the linear span of $V_{AK}$ has dimension $f(AK)$, and that linear span
contains the linear span of the vectors in $S$ (with dimension $f(K)$), the
rest is provided by the linear span of $V^{\mathsf c}_A$, thus having
dimension $f(AK)-f(K)$, providing the required representation of the
contraction.

For the principal extension, we can assume that $\alpha$ is rational (linear
polymatroids form a closed set), then take the multiple by the denominator
making $\alpha$ integer (linear polymatroids are closed under
multiplication). From this point on the proof uses \emph{generic} vectors.
We may assume that the underlying field $\mathbb F$ is large enough since
the same vector representation works over every field extension. Let $V_N =
\bigcup \{ V_i:i\in N\}$ be the set of all vectors in the representation of
the polymatroid $(f,N)$, and let $L$ be the linear span of the vectors in
$V_Z=\bigcup\{ V_i:i\in Z\}$, the linear subspace $L$ has dimension $f(Z)$.
Choose $\mathbf g_1$ from $L$ so that it is not in the linear span of any
subset of the vectors from $V_N$ except when this linear span contains the
whole $L$. Such a ``generic'' vector exists assuming $\mathbb F$ is large
enough as each constraint excludes a lower dimensional subspace of $L$. Then
let $\mathbf g_2$ be generic for $V_N\cup\{\mathbf g_1\}$, $\mathbf g_3$ be
generic for $V_N\cup\{\mathbf g_1,\mathbf g_2\}$, etc. Finally let the
representation of the new variable $z'$ be $V_{z'}=\{\mathbf
g_1,\dots,\mathbf g_\alpha\}$. Since these vectors were chosen to be
generic, we clearly have
$$
   \dim(V_{z'}\cup V_A)=\min\{\dim(V_A)+\alpha,\dim(V_{ZA})\}
$$
for every $A\subseteq N$, as was required.
\end{proof}

\section{Copy Lemma}\label{sec:copylemma}

Section \ref{sec:entreg} discussed several operations on the entropy profile
of probabilistic distributions. As was shown in Section
\ref{sec:polymatroids}, when applied to a polymatroid, all of them produce
another polymatroid. A perhaps more surprising fact is the converse: many
traditional polymatroid operations preserve almost entropic polymatroids,
and quite a few of them even preserve entropic polymatroids. Table
\ref{table:2} summarizes the operations discussed so far, and indicates which
polymatroid classes they preserve. None of them separates almost entropic
and general polymatroids, consequently none of them can produce non-Shannon
inequalities.

\begin{table}[!ht]
\def\mstr#1{\rule{0pt}{1#1pt}}%
\def\yes{$\checkmark$}%
\def\no{$\times$}%
\centering\begin{tabular}{lcccc}
\mstr2 \bf Operation & \bf linear &\bf entropic &\bf aent &\bf polymatroid \\
\hline
\mstr3 sum $f{+}g$             & \no  & \yes & \yes & \yes \\
\mstr1 direct sum $f\oplus g$  & \no  & \yes & \yes & \yes \\
\mstr1 scaling $\lambda\tsp f$ & \yes & \no  & \yes & \yes \\
\mstr1 factoring $f/{\sim}$    & \yes & \yes & \yes & \yes \\
\mstr1 deletion $f\del K$      & \yes & \yes & \yes & \yes \\
\mstr1 contraction $f\contr K$ & \yes & \no  & \yes & \yes \\
\mstr1 tightening $f\dn$       & \yes & \no  & \yes & \yes \\
\mstr1 principal extension     & \yes & \no  & \yes & \yes \\
\mstr1 Ahlswede-K\"orner ($\GAK$) & \yes & \no  & \yes & \yes \\
\mstr1 extracting CI           & \yes & \no  & \no  & \no  \\
\end{tabular}
\caption{Polymatroid classes preserved by different operations}\label{table:2}
\end{table}

An operation that preserves entropic (and aent) polymatroids, but not
general polymatroids, was discovered by Zhang and Yeung
\cite{ZhY.ineq,Zh.gen.ineq}.
Their method was later formalized and named \emph{Copy Lemma} by Dougherty
et al.~\cite{DFZ11}. The operation could have been called equally
``conditional reflection'' since it creates a conditional mirror image of a
subset of the random variables. After formally defining the operation we
prove that it indeed preserves aent polymatroids. Then we show that the Copy
Lemma maps the V\'amos vector $\mathbf v_{cd}$, defined in (\ref{eq:vamos}),
outside the polymatroid region, thus proving that this vector is not aent.

Let us first introduce some notation that will be used in connection with
the Copy Lemma. The base set $N$ is partitioned into two non-empty sets as
$N=E\cups D$. For $A\subseteq E$ a copy of the points in $A$, disjoint from
$N$, is denoted by $A'$. The permutation $\pi_A$ of the set $AA'D$, which
swaps the corresponding elements of $A$ and $A'$ and is the identity on $D$,
is called \emph{canonical permutation} or \emph{canonical map}. If
$A=E$, then the copy is denoted by $E'$, and the canonical map on $E'ED$ is
denoted by $\pi$ without the index $E$.

\begin{definition}[An $A$-copy over $D$]\label{def:copy}

Let $(f,N)$ be a polymatroid, $N=E\cups D$, $A\subseteq E$, and $\pi_A$ be the
canonical map on $AA'D$. The polymatroid $(f^*,A'N)$ is
an \emph{$A$-copy of $f$ over $D$} if the following conditions hold:

\begin{itemize}\setlength\itemsep{0pt}\setlength\parskip{0pt}\setlength\parsep{0pt
plus 2pt minus 0pt}
\item[(i)] $f^*$ is an extension of $f$, that is, $f^*(J)=f(J)$ for all
$J\subseteq N$.

\item[(ii)]
The canonical map $\pi_A$ provides an isomorphism between 
$f^*\restr A'D$ and $f^*\restr AD$ as $f^*(\pi_A(J))=f^*(J)$ for all
$J\subseteq AD$.

\item[(iii)] $f^*(A',E\|D)=0$, that is, $A'$ and $E$ are conditionally
independent over $D$. Using (i) and (ii) this condition can be written
equivalently as
$$
   f^*(A'N)=f(A\|D)+f(N).
$$
\end{itemize}
$f^*$ is a \emph{full copy}, or just a \emph{copy of $f$ over $D$} if $f^*$ is
an $E$-copy over $D$.
\end{definition}

Observe that requirements (i), (ii), and (iii) do not completely specify
$f^*$. For example, there are no conditions on those subsets of $A'A$ that
contain elements from both $A'$ and $A$. A polymatroid extension $f^*$ of
$f$ that satisfies the first two conditions always exists: simply take parallel
extensions along all elements of $A$. This parallel extension will be an
$A$-copy if and only if $A$ is determined by $D$, that is, $f(A\|D)=0$.

Since in a polymatroid $f^*(E',E\|D)=0$ implies $f^*(A',E\|D)=0$ for all $A'
\subseteq E'$, the restriction of a full copy of $f$ to $A'N$ is
an $A$-copy. A more general reduction also applies. Recall that $f$ is a
minor of $g$ if $f$ can be obtained from $g$ by a sequence of deletions
(restrictions) and contractions.

\begin{proposition}\label{prop:cp-red}
The $A$-copy of a minor is a minor of a full copy. Similarly,
an $A$-copy of a factor is a minor of a factor of a full copy.
\end{proposition}
\begin{proof}
It clearly suffices to prove the statement for the three operations
separately, and, by the remark above, for the case when $A=E=N\sm D$, that
is, for a full copy.

\smallskip
\noindent (i) If $f$ is the restriction of $(g,M)$ to $N$ and the copy is
over $D\subseteq N$, then let $F=D$. The restriction of
the full copy $g^*$ of $g$ over $D$ restricted to $E'N$ will be the required
copy of $f$.

\smallskip
\noindent (ii) If $f$ is the contract $g\contr Z$, then $N=M\sm Z$ and
$f(J)=g(JZ)-g(Z)$ for all $J\subseteq N$. Let $g^*$ be the full copy of $g$ over
$D\cup Z$, and set $f^*=g^*\contr Z$. Conditions (i) and (ii) of Definition
\ref{def:copy} clearly hold. Because $f^*$ is the contract of $g^*$ we have
$$\begin{array}{l@{\,=\,}l@{}l}
  f^*(ED) & g^*(EDZ) & {}-g^*(Z),\\[2pt]
  f^*(E'D) & g^*(E'DZ) & {}-g^*(Z), \\[2pt]
  f^*(D)   & g^*(DZ) & {}-g^*(Z),   \\[2pt]
  f^*(EE'D) & g^*(EE'DZ) & {} -g^*(Z).
  \end{array}
$$
Consequently $f^*(E,E'\|D)= g^*(E,E'\|DZ) = 0$, thus condition (iii) also
holds.

\smallskip\noindent (iii)
Finally, if $f$ is the factor $g/{\sim}$, then $F\subseteq M$, the union of the
equivalence classes in $D$, trivially works.
\end{proof}

By Proposition \ref{prop:cp-red}, it suffices to consider only full copies.
In practice, however, choosing a smaller $A$ to be copied, and applying
additional tricks, can significantly reduce the computational complexity
when harvesting the consequences of the existence of an $A$-copy.

The existence of a copy polymatroid is oblivious to tightening which was
defined in Section \ref{subsec:tightening}.

\begin{proposition}\label{prop:copydown}
The polymatroid $f$ has an $A$-copy over $D$ if and only if so does $f\dn$.
\end{proposition}
\begin{proof}
Clearly it suffices to consider tightening at a single point $z\in N$.
According to (\ref{eq:tightz}), $f\dn z = f-\lambda\tsp\r_z$, where
$\lambda=f(z\|N\sm z)$. The polymatroid $\r_z$ always has an $A$-copy
$r^*_z$: it is $\r_z$ (with the larger base $A'N$) if $z\notin A$, and
$\r_z+\r_{z'}$ if $z\in A$. 

For the only if part, let $f^*$ be an $A$-copy of $f$ over $D$. We claim that $f^*-\lambda
\tsp\r^*z$ is an $A$-copy of $f\dn$. Clearly all requirements in Definition
\ref{def:copy} hold (as each of them is linear); we only need to show that
it is a polymatroid. Submodularity and monotonicity follow from linearity,
so we need to show only that it is non-negative. Since $f^*(z)=f(z)\ge
f(z\|N\sm z) = \lambda$, non-negativity is clear when $z\notin A$. If $z\in
A$, then $z$ and $z'$ are independent over $D$, thus
\begin{align*}
f^*(zz') &\ge f^*(zz'\|D) = {}\\
         & = f^*(z\|D)+f^*(z'\|D) \ge{}\\
         & \ge f(z\|N\sm z) + f(z\|N\sm z) = 2\lambda,
\end{align*}
which proves the non-negativity for the $z\in A$ case.

For the if part, let $f^*$ be an $A$-copy of $f\dn$ over $D$. Then
$f^*+\lambda \r^*_z$ is clearly the required $A$-copy of $f$.
\end{proof}

Note that an $A$-copy of a tight polymatroid need not be tight. The full
copy of a tight polymatroid, if it exists, is always tight.

\begin{lemma}[Copy Lemma]\label{lemma:copy}
If $(f,N)$ is entropic (or aent, or linear), then it has an entropic (or aent,
or linear, respectively) $A$-copy over $D$.
\end{lemma}

\begin{proof}

Since these polymatroid classes are closed for restriction, it suffices to
prove the existence of a full copy. Also, requirements in Definition
\ref{def:copy} are both continuous and homogeneous, thus the lemma for aent
polymatroids follows from the entropic case, and for linear polymatroids it
follows from the linearly representable case.

So, first let $(f,N)$ be the entropy profile of the distribution $\xi =
\{\xi_a : a\in N\}$, and $N=E\cups D$. The distribution $\xi^*$ will be
defined on the variables $\{\xi^*_a: a\in E'\cup E\cup D\}$ and alphabet
$\mathcal X_E\times \mathcal X_E\times\mathcal X_D$ as follows:
$$
   \Prob_{\xi^*}(e'ed) =
\frac{~\Prob_\xi(\pi(e'd))\cdot\Prob_\xi(ed)}{\Prob_\xi(d)},
$$
where $e'\in E$, $e\in E$, $d\in D$, and $\pi$ is the canonical permutation
on $E'ED$. Requirements in Definition \ref{def:copy} clearly hold for the
entropy profile of $\xi^*$.

For the linear case assume $f$ is represented by the vector sets $\{V_a :
a\in N\}$, and let $L_J$ be the linear span of the vectors $V_J=
\bigcup\{V_a:a\in J\}$. Then $\dim(L_J)=f(J)$, in particular,
$\dim(L_N)=f(N)=d$. Set $k=f(D)$, $k+\ell=d$ and choose a $d$-element
basis $\{\mathbf b_i, \mathbf c_j:i\le k, j\le \ell\}$ of $L_N$ such that the
linear span of the vectors $\{\mathbf b_i\}$ is $L_D$. The linear span of
$\{\mathbf c_j\}$ is the complement space of $L_D$, denoted by $L^\top_D$. For each
$a\in N$ the dimension of $L_a$ is $f(a)$. Choose an $f(a)$-element basis
$V^0_a\cup V^1_a$ of $L_a$ such that $V^0_a$ contains $\dim(L_a\cap L_D)$
many vectors from $L_D$, and $V^1_a$ contains the remaining $\dim(L_a\cap
L^\top_D)$ many vectors from $L^\top_D$.

The linear representation of $f^*$ will be over the vector space spanned by
$\ell+\ell+k$ independent vectors $\mathbf c'_j$, $\mathbf c_j$, and
$\mathbf b_i$ for $j\le\ell$ and $i\le k$. For $a\in N$ the vector set
$V^*_a$ consists of the ($f(a)$ many) vectors in $V^0_a$ and $V^1_a$ written
as linear combinations of $\mathbf b_i$ and $\mathbf c_j$, respectively. For
$a'\in E'$ let $V^*_{a'}$ be the set of vectors from $V^0_{\pi(a')}$ plus
the linear combinations from $V^1_{\pi(a')}$ where each $\mathbf c_j$ is
replaced by the corresponding $\mathbf c'_j$. The map $\mathbf
c_j\leftrightarrow\mathbf c'_j$ provides the isomorphisms required in (i)
and (ii) of Definition \ref{def:copy}, and the conditional independence in
(iii) follows from the fact that the vectors $\mathbf c'_j$, $\mathbf c_j$
and $\mathbf b_i$ are independent.
\end{proof}

The V\'amos vector $\mathbf v_{cd}$ was defined in Section
\ref{subsec:case234}. As the first application we show that the Copy Lemma
maps $\mathbf v_{cd}$ outside the entropy region. This means that $\mathbf
v_{cd}$ is not almost entropic, proving that $\clGa 4$ is a proper subset of
$\Gamma_{\!4}$. Denoting the base set by $N=\{abcd\}$, and, for simplicity,
leaving out the index $cd$, the V\'amos vector is defined as
$$
  \mathbf v(J)=\begin{cases}
    4 &\mbox{ if $J=\{cd\}$},\\
    \min\{4,|J|+1\} & \mbox{ otherwise},
  \end{cases}
$$
see (\ref{eq:vamos}). A simple calculation shows that for this polymatroid
\begin{align*}
  &\mathbf v(a,b\|c)=\mathbf v(a,b\|d)=\mathbf v(a,c\|b)=\mathbf v(b,c\|a)=0, \\
  &\mathbf v(a,b)=1, ~~\mathbf v(c,d)=0,
\end{align*}
and the Ingleton value from (\ref{eq:ing}) is 
$$
  \mathbf v[a,b,c,d\,]= -(a,b)+(a,b\|c)+(a,b\|d)+(c,d) = -1.
$$
Assume, by contradiction, that $\mathbf v$ is almost entropic. Take a $c$-copy
of $\mathbf v$ over $ab$ and denote the aent extension on $abcdc'$ also by $\mathbf
v$. By (ii) of Definition \ref{def:copy}, $\mathbf v$ restricted to $abc$ and
$abc'$ are isomorphic, thus
$$
  \mathbf v(a,b\|c')=\mathbf v(a,b\|c)=0,
$$
and, similarly, $\mathbf v(a,c'\|b)=0$ and $\mathbf v(b,c'\|a)=0$. By (iii)
of Definition \ref{def:copy}, we also have $\mathbf v(c',cd\|ab)=0$. Since
the extension is almost entropic, the inequality (\ref{eq:MMRV}) must
hold in it. Plugging in $z=c'$ the Ingleton value is $-1$, all other
terms are zero, which gives the required contradiction.

An alternative way to reach a contradiction without applying the magic
(\ref{eq:MMRV}) inequality could be to use an LP solver to check that the
following moderate size linear programming problem has no feasible
solution (it is unsatisfiable):

\def\bbx#1{\strut\hbox to 1.1\parindent{\hfil #1}}
\medskip

\textbf{Variables}: $31$ variables indexed by non-empty subsets of $abcdc'$.

\smallskip
\textbf{Constraints}:

\hangindent2.4\parindent\hangafter1
\bbx{(i)} Basic Shannon inequalities in (\bref{B1}) and
(\bref{B2}) for $abcdc'$ ($85$ inequalities).

\hangindent2.4\parindent\hangafter1
\bbx{(ii)} Variables indexed by subsets of $abcd$ are set to equal the values in
the vector $\mathbf v_{cd}$ ($15$ equalities).

\hangindent2.4\parindent\hangafter1
\bbx{(iii)} Equality constraints that say that $abc'$ and $abc$ are
distributed identically ($4$ equalities).

\hangindent2.4\parindent\hangafter1
\bbx{(iv)} Independence of $c'$ and $cd$ over $ab$ ($1$ equality).

\medskip\noindent
This LP instance can be reduced to $9$ variables from $31$ by using the equality
constraints to eliminate variables. Beyond stating that the problem has no
feasible solution, the LP solver can also provide a reason for the
unsatisfiability by solving the dual problem. The solution of the dual is a
linear combination of the constraints where inequality constraints have
non-negative coefficients (while there is no restriction on the coefficients of
equality constraints), so that the combination is a trivially invalid inequality
with no variables. From this combination one can extract the Shannon 
inequalities that can be used to derive the contradiction, leading eventually
to the (\ref{eq:MMRV}) inequality.


\subsection{New entropy inequalities}\label{subsec:new-ineq}

The Copy Lemma \ref{lemma:copy} can be used to derive new entropy
inequalities. 

\begin{claim}[Zhang-Yeung inequality]
The following inequality holds in every aent polymatroid on $N=\{abcd\}$:
\begin{equation}\label{eq:ZY}
  [a,b,c,d\,] + (a,b\|c)+(a,c\|b)+(b,c\|a) \ge 0.
\end{equation}
\end{claim}
\begin{proof}

Take a $c$-copy of the polymatroid over $ab$, this extension has base
$abcdc'$. By property (iii) of Definition \ref{def:copy} we have
$(c',cd\|ab)=0$. By (ii) we also have
$$
   (a,b\|c')=(a,b\|c), ~~~(a,c'\|b)=(a,c\|b), ~~~ (b,c'\|a)=(b,c\|a).
$$
The inequality (\ref{eq:MMRV}) holds in the extension with $z$ replaced by
$c'$. Replacing terms containing $c'$ 
with equal terms computed above results in inequality (\ref{eq:ZY}).
\end{proof}

Let us now consider the problem of extracting \emph{all} non-Shannon
inequalities implied by a particular application of the Copy Lemma. So let
$(f,N)$ be an aent polymatroid, $A,D\subseteq N$, and the extended function
$f^*$ on $A'N$ be the $A$-copy of $f$ over $D$. Similarly as it was done
earlier, an LP instance will be generated that captures all
information on the copy extension $f^*$. The LP variables are indexed by
the non-empty subsets of $A'N$, they represent the values that $f^*$ assigns
to these subsets. These LP variables will be separated into two vectors. The
first one, denoted as $\mathbf x$, contains the \emph{main} variables whose
index is a subset of $N$; the second vector, denoted by $\mathbf y$,
contains the remaining \emph{auxiliary} variables.

Condition (ii) in Definition \ref{def:copy} stipulates that $f^*$ restricted
to $A'D$ and $AD$ are isomorphic. Expressed differently, variables mapped by
the canonical map $\pi_A$ have equal values. This allows us to eliminate those
auxiliary variables whose index is a subset of $A'D$ but is not a subset of
$N$; there are $(2^{|A|}\m-1)\cdot 2^{|D|}$ many of them. Condition (iii)
says that $A'$ and $E$ are conditionally independent given $D$. This implies
that every $J'\subseteq A'$ and $K\subseteq E$ are conditionally independent
given $D$; this allows eliminating additional
$(2^{|A|}\m-1)\allowbreak(2^{|E|}\m-1)$ auxiliary variables corresponding to
the subsets $J'KD$; their values are equal to a linear combination of three
main variables.

The copy $(f^*,A'N)$ is a polymatroid if and only if all basic Shannon
inequalities in (\bref{B1}) and (\bref{B2}) hold in it. After incorporating
the above reduction of auxiliary variables, each of these inequalities
can be written as the non-negativity of a scalar product $\langle \mathbf
p_i,\mathbf q_i \rangle \cdot \langle\mathbf x,\mathbf y\rangle\ge 0$, or as
$$
   \mathbf p_i\cdot\mathbf x + \mathbf q_i\cdot\mathbf y \ge 0,
$$ 
where the constant vectors $\mathbf p_i$ and $\mathbf q_i$ are determined by
the particular Shannon inequality. Collecting the
vectors $\mathbf p_i$ and $\mathbf q_i$ into matrices $P$ and $Q$,
$(f^*,A'N)$ is a polymatroid if and only if the condition
\begin{equation}\label{eq:PQ}
       P\tsp \mathbf x^{\mathsf T} + Q\tsp \mathbf y^{\mathsf T} \ge \mathbf 0
\end{equation}
holds for the vectors $\mathbf x$ and $\mathbf y$. The linear inequality $\mathbf e\cdot
\mathbf x\ge 0$ is a consequence of this instance of the Copy Lemma if it is
a consequence of the inequalities in (\ref{eq:PQ}), namely, if the following
implication holds:
\begin{quote}
\emph{if} all inequalities in (\ref{eq:PQ}) are true for 
$\mathbf x$ and $\mathbf y$,
\emph{then} $\mathbf e\cdot\mathbf x\ge 0$.
\end{quote}
According to the Farkas' lemma \cite{ziegler}, the inequality $\mathbf e\cdot
\mathbf x \ge 0$, which is the same as $\langle\mathbf e,\mathbf 0\rangle
\cdot \langle\mathbf x, \mathbf y\rangle \ge 0$, is a consequence of the
inequalities in (\ref{eq:PQ}) if and only if it is a non-negative linear
combination of them. Consequently the set
\begin{equation}\label{eq:cone-Q}
    \mathcal Q = \{ \mathbf h P: \mathbf hQ=\mathbf 0,~ \mathbf h\ge 0 \}
\end{equation}
contains exactly the coefficients $\mathbf e$ of those linear 
inequalities $\mathbf e\cdot
\mathbf x\ge 0$ which are consequences of this Copy Lemma instance. From this
formulation it is clear that $\mathcal Q$ is a polyhedral cone in the
$(2^{|N|}\m-1)$-dimensional space. The minimal set of inequalities (that
is, vectors) in $\mathcal Q$ of which every other is a consequence (by
taking their non-negative linear combination) is the set of the
\emph{extremal rays} of the cone $\mathcal Q$. Computing all extremal rays
(and all bounding hyperplanes simultaneously) of this implicitly defined
cone is the task of \emph{vertex enumeration algorithms}, see
\cite{fukuda-prodon,hamel-loehne-rudloff} for such algorithms in general,
and \cite{multiobj} for the special case discussed here. The result of
running such an algorithm on the cone $\mathcal Q$ is the list of the
minimal set of inequalities implied by this Copy Lemma instance. This
set, however, might also contain basic Shannon inequalities for $N$;
these should be filtered out in a post-processing step.

The enumeration algorithm outlined above can easily incorporate additional
constraints on the copy polymatroid. Since $f^*$ is not only polymatroid but
it is also aent, it satisfies additional non-Shannon inequalities. These
inequalities can be added without much ado as new rows to the matrices $P$
and $Q$. Also, rows $\langle\mathbf p_i, \mathbf q_i\rangle$ where $\mathbf
q_i$ is the all zero vector (including rows obtained from Shannon
inequalities on the base set $N$) can be deleted safely as they do not
contribute to the condition $\mathbf h\,Q=\mathbf 0$. A deleted
$\langle\mathbf p_i,\mathbf 0\rangle$ row contributes the inequality
$\mathbf p_i\cdot \mathbf x\ge 0$ to the consequence set. The inequality
specified by the coefficients in $\mathbf p_i$ can be checked to be a
Shannon inequality, or a consequence of other inequalities in the
obtained list of extremal rays. If neither is the case, $\mathbf p_i$ should
be added to the final list of consequences.


\subsection{Balancing}\label{subsec:balancing}

Claim \ref{claim:aent-tight} can be used to decrease both the dimension of
the polyhedral cone $\mathcal Q$ and the number of constraints (rows in
matrices $P$ and $Q$) in (\ref{eq:PQ}). The linear inequality $\mathbf
e\cdot \mathbf x\ge 0$ is a (valid) \emph{entropic inequality} if it holds
for all vectors $\mathbf x$ in the entropy region $\GaN$. Conic
(non-negative linear) combination of entropy inequalities is also an entropy
inequality, thus these vectors form a (closed) cone $\mathcal E$. It is the
dual cone of both $\GaN$ and $\clGa N$, meaning

$$
  \mathbf e \in  \mathcal E \Longleftrightarrow \mathbf e\cdot \mathbf x \ge 0 
   \mbox{ for all } \mathbf x\in\GaN \Longleftrightarrow
       \mathbf e\cdot \mathbf x \ge 0 
   \mbox{ for all } \mathbf x\in\clGa N .
$$

For $i\in N$ the inequality $\mathbf s_i\in\mathcal E$ expresses $\H(N) \ge
\H(N\sm i)$, that is, the vector $\mathbf s_i$ has two non-zero coordinates:
plus one at $N$, and negative one at $N\sm i$. Inequalities $\mathbf s_i$
are just the basic Shannon inequalities listed in (\bref{B1}). Recall that
the entropic vector $\r_i$, defined in Section \ref{subsec:aent-cone}, is the
characteristic vector of the relation $i\in J$. The inequality $\mathbf
e\in\mathcal E$ is called \emph{balanced} if the scalar product $\mathbf
e\cdot \r_i$ is zero for all $i\in N$. The following claim is phrased as ``every
information inequality can be strengthened to a balanced one'' in
\cite{Chan} and in \cite{Kaced} .

\begin{claim}\label{claim:balanced}
Each $\mathbf e\in\mathcal E$ can be written as a unique conic combination
of basic Shannon inequalities $\{\mathbf s_i: i\in N\}$ and
a balanced inequality $\mathbf e^\circ\in\mathcal E$.
\end{claim}

\begin{proof}

Let $\mathbf e\in\mathcal E$, and for $i\in N$ define $\mu_i=\mathbf
e\cdot\r_i$. The numbers $\mu_i$ are non-negative as $\r_i$ is an entropic
vector. Let $\mathbf e^\circ=\mathbf e -
\sum_{i\in N}\mu_i\mathbf s_i$. Then $\mathbf e^\circ$ is balanced as the
scalar product $\mathbf s_i\cdot \r_j = \delta_{ij}$ where $\delta_{ij}$ is
the Kronecker delta, and it is the unique way to write $\mathbf e$ as a linear
combination of a balanced vector and vectors from $\{\mathbf s_i:i\in N\}$.

To finish the proof we need to show that $\mathbf e^\circ$ is in $\mathcal E$.
To this end pick an arbitrary almost entropic vector $\mathbf x\in\clGa N$.
Its \emph{tight} version, as defined in (\ref{eq:tight}), is
$$
   \mathbf x\dn=\mathbf x - \ssum_{i\in N}\lambda_i\r_i,
$$
where $\lambda_i = \mathbf x(i\|N\sm i)$, in other words, $\lambda_i =\mathbf s_i
\cdot \mathbf x$.
Since $\mathbf x\dn$ is almost entropic by Claim \ref{claim:aent-tight}
and $\mathbf e$ is a valid information inequality, the scalar product $\mathbf
e\cdot(\mathbf x\dn)$ is non-negative, thus
\begin{align*}
  0 &\le \mathbf e\cdot\big(\mathbf x-\ssum_{i\in N}\lambda_i\r_i \big)
    = (\mathbf e\cdot \mathbf x) - \ssum_{i\in N}(\mathbf s_i\cdot \mathbf x)(\mathbf e\cdot
\r_i) = {} \\
    &=\big(\mathbf e - \ssum_{i\in N}(\mathbf e\cdot\r_i)\mathbf s_i\big)
\cdot \mathbf x = \mathbf e^\circ\cdot\mathbf x,
\end{align*} 
showing that $\mathbf e^\circ\in\mathcal E$ indeed.
\end{proof}

Balanced vectors form a linear subspace of the $(2^{|N|}\m-1)$-dimensional
Euclidean space; this subspace has dimension $|N|$ less. By Claim
\ref{claim:balanced} it suffices to concentrate on balanced entropy
inequalities only, thus the cone $\mathcal Q$ can be restricted to this
subspace. This can be done, for example, by adding the constraints $\mathbf
s_i\cdot\mathbf x=0$ to the defining matrices in (\ref{eq:PQ}), which then
would reduce the number of main variables in $\mathbf x$ by $|N|$.

\smallskip

The next claim shows that if only balanced consequences are required from
the Copy Lemma, then it suffices to use only balanced inequalities for
$f^*$. In particular, among basic Shannon inequalities only those in
(\bref{B2}) should be used to form the matrices $P$ and $Q$, and none from
the set (\bref{B1}).

\begin{claim}

Let $f^*$ be an $A$-copy of $(f,N)$ over $D$ as in Definition
\ref{def:copy}. Let $\mathcal E$ be a set of inequalities specified for
$f^*$\!, each valid for polymatroids, and $\mathbf e$ be a balanced inequality
for $f$. If $\mathbf e$ is a consequence of $\mathcal E$, then it is a
consequence of $\mathcal E^\circ=\{\mathbf g^\circ: \mathbf g\in\mathcal
E\}$. Additionally, every consequence of $\mathcal E^\circ$ on $f$ is
balanced.
\end{claim}

\begin{proof}

The base of $f$ is partitioned as $N=E\cups D$, $A\subseteq E$, and $f^*$ is
on $A'N$. If $i\in A$ then we let $i'=\pi(i)\in A'$ be its copied instance.
Since the inequality $\mathbf g\in\mathcal E$ valid for polymatroids, it is
a conic combination of $\mathbf g^\circ$ and $\mathcal S=\{\mathbf s_i:i\in
A'N\}$. Thus consequences of $\mathcal E$ are also consequences of $\mathcal
E^\circ \cup \mathcal S$. Therefore it suffices to show that every balanced
consequence of this set is also a consequence of $\mathcal E^\circ$, and
that $\mathcal E^\circ$ has only balanced consequences on $N$. Observe that
whether an inequality is balanced or not does not depend on the base set, so
$\mathbf e$ is balanced in $N$ if and only if it is balanced in $A'N$.

For an equivalent condition for consequence we will use Farkas' lemma
\cite{ziegler} again, but in this case without eliminating auxiliary
variables. The inequality $\mathbf e$ is a consequences of the Copy Lemma if
it is a linear combination of the inequalities in $\mathcal E$ with
non-negative coefficients, and the additional equality constraints with
arbitrary coefficients. The equality constraints are $f^*(K'L)-f^*(KL)=0$
for $K\subseteq A$ and $L\subseteq D$, $K$ not empty, plus the single
equality expressing that $A'$ and $E$ are conditionally independent:
\begin{equation}\label{eq:bal1}
  f^*(A',E\|D) =  f^*(A'D)+f^*(ED)-f^*(D)-f^*(A'ED) = 0.
\end{equation}
Write every condition as a linear combination of its balanced part and
vectors from $\mathcal S=\{\mathbf s_i:i\in A'N\}$. We are focusing on the
components from $\mathcal S$. Constraint (\ref{eq:bal1}) is balanced, thus
it has no components from $\mathcal S$. The $\mathcal S$-part of the
equality constraint $f^*(K'L)-f^*(KL)=0$ is $\sum \{ (\mathbf s_{i'}-
\mathbf s_i):i\in K\}$. Inequalities in $\mathcal E^\circ$ have no $\mathcal
S$-components. Finally, we have the inequalities $\mathbf s_i\ge 0$ for
$i\in A'N$ with this single component. The linear combination must contain
variables from $N$ only, and must be balanced. Consequently, in this
combination the coefficient of $\mathbf s_i$ for $i\in N$ must be zero
(balanced), and the coefficient of $\mathbf s_{i'}$ for $i'\in A'$ must also
be zero (the result should contain variables from $N$ only). The
contribution of the equality constraints $f^*(K'L)-f^*(KL)=0$ is a linear
combination of the differences $\{(\mathbf s_{i'}-\mathbf s_i):i\in A\}$.
Other contributions come from the inequalities $\mathbf s_i\ge 0$ for $i\in
A'N$; these, however, with non-negative weights $\mu_i$. Therefore the
coefficients of $\mathbf s_{i'}$ and $\mathbf s_i$ for $i\in A$ in the total
sum are $\lambda_i+\mu_{i'}$ and $-\lambda_i+\mu_i$, respectively for some
real number $\lambda_i$, and the coefficient of $\mathbf s_j$ for $j\in N\sm
A$ is $\mu_j$. All these coefficients are zero, which means $\mu_i=0$ for
all $i\in A'N$, which proves the first part of the claim.

For the second part observe that, starting from balanced inequalities, only
the equality constrains $f^*(K'L)-f^*(KL)=0$ contain components from
$\mathcal S$, and these components are linear combinations of $\{(\mathbf
s_{i'}- \mathbf s_i):i\in A\}$. In the final sum only $\mathbf s_{i'}$ and
$\mathbf s_i$ can appear with the same value but with opposite signs. Since the
sum must be about $N$, the coefficient of $\mathbf s_{i'}$ must be zero, and
then the coefficient of $\mathbf s_i$ is zero as well. It means that the
result is balanced as was claimed.
\end{proof}


\subsection{Preconditions}\label{subsec:copy-cond}

This Section presents some general, and easy to check, conditions on the
polymatroid $(f,N)$ that guarantee the existence of a copy extension. If any
of these ``preconditions'' hold, then the expensive polyhedral computation
discussed in Section \ref{subsec:new-ineq} is guaranteed to yield nothing
new. Each of these preconditions is about the set $D$ of ``over'' variables.

\begin{claim}\label{claim:Dmodular}
$(f,N)$ has a polymatroid $A$-copy over $D$ if $f(D)=\sum\{f(i):i\in D\}$.
\end{claim}

\begin{proof}

By Proposition \ref{prop:cp-red} it suffices to prove the existence of a
polymatroid full copy. Observe that the condition implies
$f(L)=\sum\{f(i):i\in L\}$ for every $L\subseteq D$, and then
\begin{equation}\label{eq:modular}
   f(L_1)+f(L_2)=f(L_1\cup L_2)+f(L_1\cap L_2)
\end{equation}
for arbitrary subsets $L_1$ and $L_2$ of $D$. So let $E=N\sm D$, $E'$ be a 
copy of $E$ and $\pi:EE'D\leftrightarrow E'ED$ be the canonical
map as in Definition \ref{def:copy}. For
subsets $I,J\subseteq E$, and $K\subseteq D$ define
$$
   f^*(I\pi(J)K) \eqdef \min\nolimits_L \big\{ f(IL)+f(JL)-f(L)\,:~ 
  K\subseteq L\subseteq D \big\}.
$$
If $J$ is empty, then $f^*(IK)=f(IK)$, and if $I$ is empty, then
$f^*(\pi(J)K)=f(JK)$, thus $f^*$ is an extension of $f$, and $f^*\restr E'D$
is isomorphic to $f\restr ED$. Also, $f^*(EE'D)=f(ED)+f(ED)-f(D)$, thus $E'$
and $E$ are conditionally independent given $D$. Thus $f^*$ is a copy if it
is a polymatroid. Monotonicity of $f^*$ follows easily from the monotonicity
of $f$ and from the fact that the minimum of a set is $\le$ than the minimum
of a subset. Therefore, we only need to check submodularity. Let $I_i$, $J_i$, $K_i$
for $i=1$ and $i=2$ be two triplets so that in the definition of
$f^*(I_i\pi(J_i)K_i)$ the minimal value is taken at $K_i\subseteq
L_i\subseteq D$. Let $I_3=I_1\cup I_2$ and $I_4=I_1\cap I_2$, and similarly
for $J$, $K$, and $L$. Using submodularity of $f$ we have
\begin{align*}
   & f^*(I_1\pi(J_1)K_1)+f^*(I_2\pi(J_2)K_2) = {} \\[2pt]
   & ~~{} = \big(f(I_1L_1)+f(J_1L_1)-f(L_1)\big) + 
            \big(f(I_2L_2)+f(J_2L_2)-f(L_2)\big) ={} \\[2pt]
   & ~~{} = \big(f(I_1L_1)+f(I_2L_2)\big) + \big(f(J_1L_1)+f(J_2L_2)\big) 
             - \big(f(L_1)+f(L_2)\big) \ge {} \\[2pt]
   & ~~{} \ge \big(f(I_3L_3)+f(I_4L_4) \big) +
              \big(f(J_3L_3)+f(J_4L_4) \big) -\big(f(L_1)+f(L_2)\big) = {}
\\[2pt]
   & ~~{} = \big( f(I_3L_3)+f(J_3L_3)-f(L_3)\big) +
            \big( f(I_4L_4)+f(J_4L_4)-f(L_4)\big ) \ge {} \\[2pt]
   & ~~{} \ge f^*(I_3\pi(J_3)K_3) + f^*(I_4\pi(J_4)K_4),
\end{align*}
where we used (\ref{eq:modular}). This proves that $f^*$ is indeed submodular.
\end{proof}

A slightly more general statement is true, covering the case when $D$ has
only two elements $d_1$, $d_2$, and $f(d_1)=f(d_1d_2)$. The exact condition
for this general case is rather technical, while the proof is almost
identical. For details, please consult \cite{Csirmaz.oneadhesive,M.fmadhe}.

\begin{claim}\label{claim:smallD}
$(f,N)$ has a polymatroid copy over $D$ if $|D|=1$, or $|D|=|N|-1$.
\end{claim}
\begin{proof}

$|D|=1$ is a special case of Claim \ref{claim:Dmodular}. For $|D|=|N|-1$ let
$a$ be the single element of $N$ that is not in $D$, and let $\lambda=f(a\|D)$.
After tightening at $a$ the polymatroid $f$ becomes $ f\dn a = f - \lambda
\r_a $ as defined in Section \ref{subsec:tightening}. Add $a'$ parallel to
$a$ in $f\dn a$ to get $f'$, and finally let $f^*$ be $f'+\lambda \r_a +
\lambda\r_{a'}$. It is clear that both $f^*\restr aD$ and $f^*\restr a'D$
are isomorphic to $f=f\dn a+\lambda\r_a$, moreover
\begin{align*}
   f^*(aa'D)&=f'(aa'D)+2\lambda = (f\dn a)(aD)+2\lambda = f(aD)+\lambda
={}\\[2pt]
   &=f(aD)+f(aD)-f(D)=f^*(aD)+f^*(a'D)-f^*(D).
\end{align*}
This shows that $a$ and $a'$ are independent over $D$, thus the polymatroid
$f^*$ is indeed an $a$-copy of $f$ over $D$.
\end{proof}

A subset $F\subseteq N$ is a \emph{flat} if every $F'\supsetneqq F$ has a
strictly larger rank. The next two claims show that it suffices to
take copies over flat subsets.

\begin{claim}
Suppose $d\notin A$, $f(Dd)=f(D)$, and $f^*$ is an $A$-copy over $Dd$.
Then $f^*$ is also an $A$-copy over $D$.
\end{claim}

\begin{proof}
$f^*\restr A'D$ and $f\restr AD$ are isomorphic as they are isomorphic even
adding $d$. Thus we need to check independence only. Let $E=N\sm Dd$. Since
$f^*(Dd)=f^*(D)$ implies $f^*(A',Ed\|D)=f^*(A',E\|Dd)$, we are done.
\end{proof}

\begin{claim}
Suppose $f(Dd)=f(D)$, $f^*$ is an $A$-copy over $Dd$, and $d'$ is
parallel to $d$. Then $A'd'$ is an $Ad$-copy over $D$.
\end{claim}

\begin{proof}
Since $AdD$ and $A'dD$ are isomorphic and $d'$ is parallel to $d$ in $f^*$\!,
$AdD$ and $A'd'D$ are also isomorphic. For the independence
\begin{align*}
  (A'd',E\|D) &= A'd'D+ED-A'd'ED-D={}\\[2pt]
              & = A'Dd+ED-A'EDd-D=(A',E|Dd) = 0
\end{align*}
as $f^*$ is a copy and $f(Dd)=f(D)$. This proves the claim.
\end{proof}

\subsection{Iterating}\label{subsec:iterating}

Starting from an entropic (or aent) polymatroid, the Copy Lemma guarantees
an entropic (or aent) extension. Thus the Lemma can be applied iteratively
to this new polymatroid, obtaining a larger extension of the original one, and
so on. Such an \emph{iterated copy} is usually specified by the sequence of subset
pairs, where the pairs are separated by semicolons, and a pair is separated by a
colon, such as
$$
   \langle\, A_1:D_1; ~~ A_2:D_2; ~~\cdots; ~~
   A_k:D_k\,\rangle.
$$ 
Here each copied set $A_j$ and the over set $D_j$ is a subset of the base of 
the previous copy.
Since the copy of a restriction is a restriction of a copy by Proposition
\ref{prop:cp-red}, this iterated copy is a restriction of the \emph{iterated
full copy} determined solely by the sequence $\langle D_1;D_2; \cdots\rangle$
of the over sets.

\def\ta#1{{\tt a}_{#1}}
\def\tb#1{{\tt b}_{#1}}
\def\tc#1{{\tt c}_{#1}}
\def\td#1{{\tt d}_{#1}}

As an example, let the base of the initial polymatroid be $\{\ta1 \tb1 \tc1 \td1\}$.
Take a $\td1$-copy over $\ta1\tb1$, and mark the ``mirror image''
of $\td1$ by $\td2$. That is, the canonical map $\pi_1$ of this copy step
gives $\pi_1(\td1)=\td2$. Thus $\pi_1:\td1\ta1\tb1 \leftrightarrow \td2\ta1\tb1$
describes the copy isomorphism, while the conditional independence is
$(\td2,\tc1\td1\| \ta1\tb1)=0$. The next iteration is an $\ta1$-copy over
$\tb1\td1\td2$ naming the image as $\pi_2(\ta1)=\ta2$. The final iteration is a
$\tb1$-copy over $\ta1\ta2\td1\td2$ with $\pi_3(\tb1)=\tb2$. This copy
sequence can be written succinctly as
$$
   \td2\m=\td1 \m: \ta1\tb1; ~~ \ta2\m=\ta1 \m: \tb1\td1\td2; ~~
\tb2\m=\tb1\m:\ta1\ta2\td1\td2.
$$
The third copy in this sequence is a polymatroid on seven points, as $\tc1$
is carried over. When
taking full copies the first iteration results in a six-point
polymatroid, the second iteration in a nine-point polymatroid, and the
third iteration in a 14-point polymatroid. The full-copy iteration can be
specified by the over sets $\ta1\tb1$, $\tb1\td1\td2$, and
$\ta1\ta2\td1\td2$ only; the complete specification with explicit naming of
the new elements is
\begin{equation}\label{eq:copyseq}
  \begin{array}{r@{\m=}l}
  \tc2\td2 & \tc1\td1\m:\ta1\tb1; \\[2pt]
  \ta2\tc3\tc4 & \ta1\tc1\tc2\m:\tb1\td1\td2; \\[2pt]
  \tb2\tc5\tc6\tc7\tc8 & \tb1\tc1\tc2\tc3\tc4 \m: \ta1\ta2\td1\td2.
  \end{array}
\end{equation}
Restricting this iterated full copy $f^*$ to the seven element subset
$\ta1\ta2\tb1\tb2\tc1\td1\td2$ gives the result of the previous iteration sequence.
Observe that $f^*$ contains eight isomorphic instances of the original
polymatroid; these are at $\ta i \tb j\tc k\td\ell$ where $i$, $j$, and
$\ell$ are in $\{1,2\}$, and $k$ goes from $1$ to $8$.

\smallskip

Consequences of the iterated Copy Lemma can be determined similarly to that
of a single application which was discussed in Section
\ref{subsec:new-ineq}. First, the iterated copy is translated to an LP
instance, and then to a vertex enumeration problem. The \emph{main} LP variables
are indexed by the non-empty subsets of the initial polymatroid base, while
\emph{auxiliary} LP variables are indexed by the additional subsets of the
base of the last copy. Equality constraints expressing isomorphism induced
by the canonical maps as well as the consequences of the conditional
independences are written separately for each iteration. Finally, inequality
constraints are added, these are the basic Shannon inequalities (from
(\bref{B2}) if only balanced inequalities are required) written for the base
elements of the last iteration. (Note that these include the same
inequalities for smaller base sets.) From this LP instance the translation
to a vertex enumeration problem of a polyhedral cone $\mathcal Q$ analogous
to the one defined in (\ref{eq:cone-Q}) is straightforward. The dimension of
$\mathcal Q$ is determined by the number of main variables, which remains
the same as for a single copy-step. However, the number of auxiliary variables
and constraints can grow significantly. While the complexity of vertex
enumeration algorithms grows exponentially with the dimension of $\mathcal
Q$, it grows quite modestly with the number of constraints. These vertex
enumeration problems are, therefore, typically tractable for a couple of
iterations, and have been used successfully to generate several hundreds
of new four-variable Shannon inequalities \cite{DFZ11}.

\smallskip

A $k$-iterated copy has the same consequences as the very first copy
instance when the matrices $P$ and $Q$ of that copy step are supplemented by
those linear inequalities on $A'N$ that guarantee that the subsequent
$(k\m-1)$ iterations can be applied successfully. The converse is also true:
adding a single linear inequality stated for $A'N$ provided by some
(possibly iterated) application of the Copy Lemma, can be replaced by adding
an instance (or iterated instances) of the Copy Lemma on top of $A'N$ which
forces that inequality to hold. If more than one such an inequality is
required, further iterations can be added to the chain. While the
two methods:
\begin{itemize}\setlength\itemsep{0pt}
\item[a)] applying a single copy step augmented with additional (non-Shannon) 
          inequalities on $A'N$;
\item[b)] using an iterated copy sequence
\end{itemize}
are theoretically equivalent, in practice almost always the second one is
used \cite{G.Rom}. Determining the inequalities that are needed to be added
requires solving a significantly higher dimensional vertex enumeration
problem, which is almost always numerically infeasible. Also, even a single
instance of the Copy Lemma could produce a large number of new inequalities.
Adding all those inequalities can increase the number of
constraints (the rows in matrices $P$ and $Q$) beyond a manageable level,
while presumably only a small fraction of these inequalities are
actually needed.

\smallskip

Finally, preconditions from Section \ref{subsec:copy-cond} apply for each
iteration. If any of the conditions listed there holds for an intermediate
polymatroid, the iteration step can be omitted safely (or replaced by a
larger over set) without losing any of the consequences.


\subsection{Symmetries}\label{subsec:symmetry}

Symmetries of the polymatroid can be used to decrease the number of main and
auxiliary variables, thus increasing the efficiency of the polyhedral
algorithms. Symmetries of the copy polymatroid could provide additional
restrictions to the cone $\mathcal Q$ and can lead to new consequences. We
start exploring these symmetries with some definitions.

\begin{definition}\label{def:auto}
(i) An \emph{automorphism} of the polymatroid $(f,N)$ is a permutation
$\sigma$ of
the base $N$ such that $f(\sigma(A))=f(A)$ for all $A\subseteq N$.

(ii) A \emph{symmetry}, or \emph{partial automorphism}, of $(f,N)$ is a 
permutation $\sigma$ of a subset of $N$ such that $f(\sigma(A))=f(A)$ for all
$A\subseteq\dom(\sigma)$.
\end{definition}

Note that an automorphism of $f$ is also a symmetry. 
For a permutation $\sigma$ with $\dom(\sigma)\subseteq N$, the
\emph{restriction} of $\sigma$ to the subset $K\subseteq N$, denoted by
$\sigma\restr K$, 
is the maximal subpermutation of $\sigma$ on $K$. That is,
$$
   \dom(\sigma\restr K) = \bigcup \{L : L \subseteq K\cap
\dom(\sigma), \mbox{ and } \sigma(L)=L \},
$$
and $(\sigma\restr K)(x)=\sigma(x)$ when $x\in\dom(\sigma\restr K)$. The
restriction $\sigma\restr K$ can be empty even if
$\dom(\sigma)\cap K$ is not empty. For an illustration let $\pi$ be the
canonical map of the full copy of $f$ over $D$. The canonical map of
the $A$-copy over $D$, which was denoted by $\pi_A$, is just the
restriction $\pi\restr A'N$.

\smallskip

The proof of the Copy Lemma \ref{lemma:copy} actually yielded a stronger
result involving symmetries.

\begin{claim}\label{claim:symm1}
Every entropic (or aent, or linear) polymatroid has an entropic (or aent, or
linear, respectively) $A$-copy over $D$ in which the canonical map
$\pi_A$ is a symmetry.
\end{claim}

\begin{proof}

In the construction of the full copy, the sets $E$ and $E'$ are completely
symmetric, therefore the canonical permutation $\pi$, swapping $E$ and $E'$
and keeping $D$ fixed pointwise, is an automorphism of the copy polymatroid.
The $A$-copy can be obtained as a restriction of the full copy. The
restriction preserves those symmetries where both $J$ and $\pi(J)$ are in the
restriction. These subsets are just the ones that are preserved by the
canonical map $\pi_A$ of the $A$-copy.
\end{proof}

Adding further equality constraints implied by the symmetries reduces the
number of auxiliary variables, and also could lead to new consequences.
However, in the case of a full copy, adding the canonical symmetry $\pi$ provided
by Claim \ref{claim:symm1} does not increase the strength of the method.

\begin{claim}

Suppose $f$ has a polymatroid (or aent, or linear) full copy $f^*$ over $D$.
Then $f$ also has a $\pi$-symmetric polymatroid (or aent, or linear,
respectively) full copy over $D$.
\end{claim}

\begin{proof}
The average of $f^*$ and $\pi(f^*)$ is $\pi$-symmetric, and it is the
required full copy of $f$ over $D$.
\end{proof}

The analogous statement for an $A$-copy is not true. There are polymatroids
that have an $A$-copy but have no $A$-copy that would be $\pi_A$ symmetrical.
Using the existence of $\pi_A$-symmetrical $A$-copies can actually yield
stronger consequences than $A$-copies without that symmetry\footnote{%
This fact has been observed by the authors of \cite{DFZ11}. Some
of the non-Shannon inequalities reported in their paper actually required
additional symmetry assumptions. (Personal communication from
R.~Dougherty.)%
}.

\smallskip

Certain symmetries of the original polymatroid $f$ extend to symmetries of
its full copy. Then, as it was observed in the proof of Claim
\ref{claim:symm1}, the existence of a full copy with such symmetries
guarantees the existence of an $A$-copy with inherited
symmetries. So let $(f,N)$ be a polymatroid, $D\subseteq N$, and $\sigma$ be
a permutation on a subset of $N$ such that $\sigma(D)=D$, implying
$D\subseteq \dom(\sigma)$. This $\sigma$ can be extended to the full copy of
$f$ over $D$ as follows. Let $\pi$ be the canonical map of the full copy,
and define $\sigma^*$ as the extension of $\sigma$ to
$\dom(\sigma^*)=\dom(\sigma)\cup \pi(\dom(\sigma))$ by
$$
\left.\begin{array}{r@{\;=\;}l}
   \sigma^*(x)&\sigma(x)  \\
   \sigma^*(\pi(x))&\pi(\sigma(x))
  \end{array}
  \right\} ~~ \mbox{ for all $x\in\dom\sigma$}.
$$

\begin{proposition}\label{prop:manysymm}
Suppose $\sigma_k$ are a symmetries of $f$, and $\sigma_k(D)=D$ for all
$k$.
\begin{itemize}\setlength\itemsep{0pt}
\item[\upshape(i)]
Let $f^*$ be a $\pi$-symmetrical copy of $f$ over $D$.
Then $f^*\!$ is also $\sigma^*_k$-sym\-met\-ri\-cal for all $k$.
\item[\upshape(ii)]
There is an $A$-copy of $f$ over $D$ which has symmetries $\pi_A$ and
$\sigma^*_k\restr A'N$.
\end{itemize}
\end{proposition}

\begin{proof}
(i) By assumption, $\pi$ is an automorphism of $f^*$, meaning $f^*(J)=
f^*(\pi(J))$ for all $J\subseteq E'N$. Since $\sigma^*$ and $\pi$ commute on
$E'N$, the statement follows.

(ii) Take a full copy of $f$ as in (i). Its restriction to $A'N$ clearly
works.
\end{proof}

Proposition \ref{prop:manysymm} can be used to maintain a set of symmetries
during an iterated application of the Copy Lemma. In each step the actual
set of symmetries is updated and then applied to the next copy polymatroid.
The gain from maintaining these symmetries is twofold. On one hand, symmetries
can reduce the number of auxiliary variables significantly, so that the
associated vertex enumeration problem becomes numerically tractable. On
the other hand, symmetries of $A$-copies can force additional consequences,
leading to more, or stronger inequalities.

For an illustration, consider the iterated copy sequence in
(\ref{eq:copyseq}). The canonical map $\pi_1$ of the first
$\tc2\td2\m=\tc1\td1\m:\ta1\tb1$ copy is
$$
\pi_1 =  \left( \begin{array}{c*{5}{@{\,}c}}
         \ta1 & \tb1 & \tc1 & \tc2 & \td1 & \td2 \\
         \ta1 & \tb1 & \tc2 & \tc1 & \td2 & \td1
           \end{array}\right) ,
$$
which swaps $\tc1\td1$ and $\tc2\td2$, and keeps everything else fixed.
The next copy step is the full $\ta1\tc1\tc2$-copy over $\tb1\td1\td2$ where
the new elements are named $\ta2\tc3\tc4$. The canonical map $\pi_2$ and
the inherited symmetry $\pi^*_1$ are
\begin{align*}
  \pi_2 = \left(\begin{array}{c*{8}{@{\,}c}}
           \ta1 & \ta2 & \tb1 & \tc1 & \tc2 & \tc3 & \tc4 & \td1 & \td2 \\
           \ta2 & \ta1 & \tb1 & \tc3 & \tc4 & \tc1 & \tc2 & \td1 & \td2
                \end{array}\right ), \\[2pt]
  \pi^*_1 = \left(\begin{array}{c*{8}{@{\,}c}}
           \ta1 & \ta2 & \tb1 & \tc1 & \tc2 & \tc3 & \tc4 & \td1 & \td2 \\
           \ta1 & \ta2 & \tb1 & \tc2 & \tc1 & \tc4 & \tc3 & \td2 & \td1
                \end{array}\right ).
\end{align*}
The iteration is a full $\tb1\tc1\tc2\tc3\tc4$-copy over $\ta1\ta2\td1\td2$,
and the canonical map and the inherited symmetries are
\begin{align*}
  \pi_3 = \left(\begin{array}{c*{13}{@{\,}c}}
  \ta1&\ta2&\tb1&\tb2&\tc1&\tc2&\tc3&\tc4&\tc5&\tc6&\tc7&\tc8&\td1&\td2 \\
  \ta1&\ta2&\tb2&\tb1&\tc5&\tc6&\tc7&\tc8&\tc1&\tc2&\tc3&\tc4&\td1&\td2
          \end{array}\right), \\[2pt]
\pi^*_2 = \left(\begin{array}{c*{13}{@{\,}c}}
  \ta1&\ta2&\tb1&\tb2&\tc1&\tc2&\tc3&\tc4&\tc5&\tc6&\tc7&\tc8&\td1&\td2 \\
  \ta2&\ta1&\tb1&\tb2&\tc3&\tc4&\tc1&\tc2&\tc7&\tc8&\tc5&\tc6&\td1&\td2
          \end{array}\right), \\[2pt]
\pi^{{*}{*}}_1= \left(\begin{array}{c*{13}{@{\,}c}}
  \ta1&\ta2&\tb1&\tb2&\tc1&\tc2&\tc3&\tc4&\tc5&\tc6&\tc7&\tc8&\td1&\td2 \\
  \ta1&\ta2&\tb1&\tb2&\tc2&\tc1&\tc4&\tc3&\tc6&\tc5&\tc8&\tc7&\td2&\td1
          \end{array}\right).
\end{align*}
Without considering conditional independences, these symmetries alone reduce
the total number of LP variables for the $14$-element base set (the main and
auxiliary variables together) from $2^{14}\m-1$ to $2351$. One of the consequences of
this arrangement is that the following strengthening of the Zhang-Yeung
inequality (\ref{eq:ZY}) is also a valid entropy inequality:
$$
  [a,b,c,d\, ] + 0.8\tsp (a,b\|c)+(a,c\|b)+(b,c\|a) \ge 0.
$$


\section{Maximum Entropy method}\label{sec:MEP}

The Copy Lemma \ref{lemma:copy} can be recovered by an ingenious application
of the \emph{principle of maximum entropy} conceived in its modern form by
\cite{jaynes}. The principle can be formulated as ``if a probability
distribution is specified only partially, take the one with the largest
entropy'', see \cite{maxentp}. In this case ``partial specification'' means that some of the
marginal distributions agree with those of a fixed distribution. So let $N$
be the base of some fixed distribution $\xi=\{\xi_i:i\in N\}$ defined on the
alphabet $\mathcal X=\prod\{ \mathcal X_i:i\in N\}$. The family of subsets
of $N$ where the marginal distributions will be required to be the same will
be denoted by $\mathcal F$. Clearly, if for some distribution $\eta$, the
marginal $\eta_F$ on some $F\subseteq N$ agrees with the marginal $\xi_F$,
then the marginals on the subsets of $F$ will be the same. Thus $\mathcal F$
can be extended by subsets of its elements, or, equivalently, can be
restricted to its maximal elements. So let $\Upsilon_{\!\mathcal F}$ denote
the family of those distributions $\eta=\{\eta_i:i\in N\}$ that are on the
same alphabet $\mathcal X$ as $\xi$, and, for ever $F\in \mathcal F$, have
the same marginal on $F$ that $\xi$ does:
$$
   \Upsilon_{\!\mathcal F} = \{ \eta: \eta_F = \xi_F ~\mbox{ for all } F\in\mathcal F\}.
$$
Since the probability weights in a marginal are sums of the probability
weights of the original distribution, the condition on the marginals is
actually a large number of linear constraints on the probability weights of
$\eta$. The total Shannon entropy $\H(\eta)$ is a strictly concave
function of these weights, thus finding a distribution in
$\Upsilon_{\!\mathcal F}$ with maximum entropy is a concave optimization
problem on a compact space (the space of all distributions on
$\mathcal X$) with linear constraints. Consequently, see
\cite{convex-optim}, there is always a unique optimal
$\hat\eta\in\Upsilon_{\!\mathcal F}$ where the entropy is maximum. While
structural properties of such maximum entropy distributions are mainly
unknown, they certainly satisfy some independence statements.

To describe these statements, let us introduce some notions. A
\emph{3-partition} is a partition of the base set into three pieces as
$N=X\cups Y\cups D$. Occasionally such a 3-partition is written as
$\<X,Y\|D\>$ to emphasize the distinct role of the $D$ component. When it is
clear from the context, the simpler form $N=XY\!D$ is used. The
3-partition $XY\!D$ \emph{separates} the collection $\mathcal F$ of
subsets of $N$ if no $F\in\mathcal F$ intersects both $X$ and $Y$. Thus
$\mathcal F$ is separated into two (not necessarily disjoint) parts: those
$F\in\mathcal F$ which are subsets of $X\cup D$, and those $F\in\mathcal
F$, which are subsets of $Y\cup D$.

\begin{claim}\label{claim:me1}

Let $\mathcal F$ be a collection of subsets of $N$, and let $XY\!D$ be a
3-partition of $N$ that separates $\mathcal F$. In the maximal entropy
distribution $\hat\eta\in \Upsilon_ {\!\mathcal F}$, $X$ and $Y$ are
conditionally independent given $D$, that is, $\I_{\hat\eta}(X,Y\|D)=0$.
\end{claim}
\begin{proof}
If $\I_{\eta}(X,Y\|D)$ is not zero for some distribution
$\eta\in\Upsilon_{\!\mathcal F}$, then
$$
   \H_\eta(XY\!D) < \H_\eta(XD)+\H_\eta(Y\!D)-\H_\eta(D).
$$
Denote the probabilities in $\eta\in\Upsilon_{\!\mathcal F}$ by $\Prob_
\eta(xyd)$ where $x$, $y$, $d$ 
runs over the alphabets $\mathcal X_X$, $\mathcal X_Y$, and $\mathcal X_D$,
respectively. Consider the distribution $\eta^*$ on the same alphabet
$\mathcal X$ with probabilities
$$
   \Prob_{\eta^*}(xyd) \eqdef \frac{\;\Prob_\eta(xd)\cdot\Prob_\eta(yd)\;}
        {\Prob_\eta(d)}.
$$
The marginals of $\eta^*$ and $\eta$ on $XD$, $Y\!D$ and $D$ are the same,
therefore $\eta^*\in\Upsilon_{\!\mathcal F}$ as each $F\in\mathcal F$ is a
subset of either $XD$ or $Y\!D$, moreover $\I_{\eta^*}(X,Y\|D)\allowbreak=0$.
Thus
\begin{align*}
  \H_{\eta^*}(XY\!D) &= \H_{\eta^*}(XD)+\H_{\eta^*}(Y\!D)-\H_{\eta^*}(X) ={}\\
     &= \H_\eta(XD)+\H_\eta(Y\!D)-\H_\eta(D)> \H_\eta(XY\!D).
\end{align*}
Since $\H_{\hat\eta}(XY\!D)$ is maximal, we cannot have
$\I_{\hat\eta}(X,Y\|D)\ne 0$, as was claimed.
\end{proof}

\begin{lemma}[Maximum Entropy Method]\label{lemma:MEM}

Suppose $(f,N)$ is entropic (or aent), and $\mathcal F$ is a family of
subsets of $N$. There is an entropic (or aent, respectively) polymatroid
$(f^*\!,N)$, written as $\MEM_{\mathcal F}(f)$, such that
\begin{itemize}\setlength\itemsep{0pt}
\item[\upshape(i)]
the restrictions of $f$ and $f^*$ to elements of $\mathcal F$ are the same:
$f^*\restr F = f\restr F$ for all $F\in\mathcal F$;
\item[\upshape(ii)]
$f^*(X,Y\|D)=0$ for every 3-partition $XY\!D$ of $N$ which separates $\mathcal F$.
\end{itemize}
\end{lemma}

\begin{proof}
For entropic polymatroids it follows immediately from Claim \ref{claim:me1}.
For aent polymatroids apply continuity.
\end{proof}

Although the maximum entropy distribution in $\Upsilon_{\!\mathcal F}$ is
determined uniquely, Lemma \ref{lemma:MEM} does not completely specify the
polymatroid $\MEM_{\mathcal F}(f)$. It would be interesting to see a family
$\mathcal F$ and two distributions with the same entropy profile so that the
entropy profiles of the corresponding maximum entropy distributions are
different.

Similarly to the Copy Lemma \ref{lemma:copy}, existence of $\MEM_{\mathcal
F}(f)$ is also oblivious to tightening. This can be seen by adding or
subtracting from $f^*$ a constant multiple of the polymatroid $\r_z$.

An illustration of using the Lemma \ref{lemma:MEM} is a quick proof of a
pair of non-Shannon entropy inequalities. The top one was discovered by
Makarichev et al.~\cite{MMRV}.

\begin{claim}[\cite{MMRV}]\label{claim:MMRV}
These are 5-variable entropy inequalities:
\begin{equation}\label{eq:mmineq}
  \begin{array}{c}
   {[a,b,c,d\,]} + (a,b\|z)+(a,z\|b)+(b,z\|a) \ge 0, \\[3pt]
   {[a,c,b,d\,]} + (a,b\|z)+(a,z\|b)+(b,z\|a) \ge 0.
  \end{array}
\end{equation}
\end{claim}
\begin{proof}

Let $f$ be an entropic polymatroid on $N=\{abcdz\}$. Set $\mathcal
F=\{abcd,abz\}$, and let $f^*$ be the entropic polymatroid provided by Lemma
\ref{lemma:MEM}. By (i), all four terms on the left hand side of
(\ref{eq:mmineq}) have the same value for $f$ and $f^*$. By (ii),
$f^*(cd,z\|ab)=0$ as the 3-partition $\< cd,z\|ab\>$ separates $\mathcal F$.
Inequality (\ref{eq:MMRV}) and its variant
$$
  \begin{array}{c}
    {[a,b,c,d\,]}+(a,b\|z)+(a,z\|b)+(b,z\|a)+3\tsp(cd,z\|ab) \ge 0, \\[3pt]
    {[a,c,b,d\,]}+(a,b\|z)+(a,z\|b)+(b,z\|a)+3\tsp(cd,z\|ab) \ge 0
  \end{array}
$$
hold in any polymatroid, thus they hold for $f^*$ as well. This implies that
inequalities in (\ref{eq:mmineq}) are true for $f^*$, and then for $f$,
as required.
\end{proof}

A more involved application of \MEM{} provides infinitely many
non-Shannon inequalities.
\begin{claim}\label{claim:5k}
For each $k\ge 0$ these are 5-variable entropy inequalities:
\begin{align*}
 k\tsp\left\{{[a,b,c,d\,] \atop [a,c,b,d\,]}\right\}
  ~+~ \frac{k(k-1)}2\big(
      (a,c\|b)+(b,c\|a)\big) ~+~{} & \\
  ~+~ (a,b\|z)~+~ k\big( (a,z\|b)+(b,z\|a)\big)  &\ge 0.
\end{align*}
\end{claim}
\begin{proof}
For $k=0$ both are the Shannon inequality $(a,b\|z)\ge 0$; for $k=1$ the
inequalities are stated as (\ref{eq:mmineq}) in Claim \ref{claim:MMRV}. The
proof goes by induction on $k$. By \MEM, we can also assume that
$(cd,z\|ab)=0$ holds. Use the induction hypothesis for $k$ and the 
five variables $az$, $bz$, $cz$, $d$, and $cz$ to get
\begin{align*}
  k\tsp\left\{{[az,bz,cz,d\,] \atop [az,cz,bz,d\,]}\right\}
   + \frac{k(k-1)}2\big(
      (az,cz\|bz)+(bz,cz\|az)\big) &+ {}\\
 {}+ (az,bz\|cz) + k\big( (az,cz\|bz)+(bz,cz\|az)\big) &\ge 0.
\end{align*}
The following Shannon inequalities hold in any polymatroid:
\begin{align*}
 [a,b,c,d\,] + (a,b\|z)+(a,z\|b)+(b,z\|a) & \ge 
          (az,bz\|cz) - 3\tsp(cd,z\|ab), \\[3pt]
 [a,b,c,d\,] + (a,z\|b)+(b,z\|a) & \ge
           [az,bz,cz,d\,] - 3\tsp(cd,z\|ab), \\[3pt]
  (a,c\|b) & \ge (az,cz\|bz) - (cd,z\|ab), \\[3pt]
  (b,c\|a) & \ge (bz,cz\|az) - (cd,z\|ab),
\end{align*}
and similarly when $[a,b,c,d\,]$ is replaced by $[a,c,b,d\,]$. Take the
first inequality once, the second $k$ times, the third and fourth
$k(k\m+1)/2$ times, and add them up. The left hand side of the sum is the
claimed inequality for $k+1$, while the right hand side is $\ge\tsp 0$ by
induction and by the assumption $(cd,z\|ab)=0$.
\end{proof}

\begin{claim}\label{claim:equiv2}
The Copy Lemma \ref{lemma:copy} is equivalent to the special case of
Lemma \ref{lemma:MEM} when independence is guaranteed for a single
3-partition $XY\!D$ only.
\end{claim}

\begin{proof}

First, let $N=XY\!D$ be a single partition that separates $\mathcal F$, and
take an $X$-copy of $f$ over $D$. Denote the copy polymatroid by $(f^*,X'N)$
This $f^*$ is an extension of $f$, moreover $f^*\rest X'D \cong f\restr XD$,
and $f^*(X',Y\|D)=0$. Therefore $f^*\restr X'Y\!D$ satisfies the conditions
in Lemma \ref{lemma:MEM}.

In the other direction it suffices to show the existence of a full copy of
$f$ over $D$ where $N$ is partitioned as $N=ED$. Take an extension $f'$ of
$f$ to the base set $N'=E'ED$ such that $E'$ is parallel to $E$. Apply Lemma
\ref{lemma:MEM} to $f'$, the family $\mathcal F=\{ED,E'D\}$ and the single
3-partition $N'=E'ED$. The provided polymatroid $f^*$ is clearly the
required full copy.
\end{proof}

While clearly more general than the Copy Lemma, and can encompass some
special cases of the iterated Copy Lemma, it is not known whether it can
actually generate an entropy inequality that the iterated Copy Lemma cannot.
Finding a reasonable and numerically amenable \emph{iterated} \MEM{} is
also a challenging open problem. It is not known either if \MEM{} preserves
linear polymatroids.


\subsection{Consequences}\label{subsec:me-consequences}

This Section discusses how all consequences of a \MEM{} instance can be
extracted. There are two interrelated notions in a \MEM{} setting: the
family $\mathcal F$ of marginals to be preserved, and the family $\mathcal
P$ of 3-partitions of $N$ which separate $\mathcal F$. For any family
$\mathcal F$ of subsets, and any family $\mathcal P$ of 3-partitions
define $\mathcal F^\bot$ as a family of 3-partitions, and $\mathcal
P^\bot$ as a family of subsets as follows:
\begin{align*}
   \mathcal F^\bot &\eqdef \{
     XY\mkern-3mu D : \mbox{ for all $F\in\mathcal F$, either $X\cap F$ or $Y\cap F$ is
empty}\},
\\[2pt]
   \mathcal P^\bot &\eqdef \{
     F: \mbox{ for all $XY\!D\in\mathcal P$, either $X\cap F$ or $Y\cap F$ is
empty} \}.
\end{align*}
Clearly, $\mathcal F^\bot$ is just the set of 3-partitions that separate
$\mathcal F$. These operations provide a Galois connection between these
families. In particular, both $\mathcal F$ and its closure, $\mathcal
F^{\bot\bot}\supseteq\mathcal F$, have the same set of 3-partitions that
separates them. Thus one can take $\mathcal F$ to be closed, i.e., $\mathcal
F=\mathcal F^{\bot\bot}$; in this case $\mathcal F$ is also closed downwards.
Similarly, the \MEM{} instance can be specified equally by a family
$\mathcal P$ of 3-partitions providing the conditional independence
statements. This family $\mathcal P$ can also be assumed to be closed. If
$\mathcal P$ consists of a single element $XY\!D$, then it is closed, and
$\mathcal P^\bot$ contains all non-empty subsets of $XD$ and $Y\!D$. This
latter set is the same as all subsets of $N$ excluding those that
intersect both $X$ and $Y$.

Assume that the \MEM{} instance is specified by the families $\mathcal F$
and $\mathcal P$ such that $\mathcal F^\bot=\mathcal P$ and $\mathcal
P^\top=\mathcal F$; this implies that both $\mathcal F$ and $\mathcal P$ are
closed. The balanced inequality $\mathbf e\cdot\mathbf x\ge 0$ is a
consequence of this \MEM{} instance if $\mathbf e$ contains non-zero
coefficients only at positions whose labels are elements of $\mathcal F$ (as
only these polymatroid values are guaranteed to be preserved by $f^*$), and
$\mathbf e\cdot \mathbf x\ge 0$ holds in every polymatroid on $N$ that
satisfies all conditional independence statements in $\mathcal P$.

The value of the conditional independence statement $(X,Y\|D)$, generated by
the 3-par\-ti\-tion $N=X Y\! D$, is zero if and only if
each of the following set of basic Shannon inequalities, denoted by $\mathcal
S(XY\!D)$, evaluates to zero:
$$
   \mathcal S(XY\!D) \eqdef 
   \big\{ (a,b\|K): ~ K\supseteq D,~ a\in X\sm K,~ b\in Y\sm K \big\}.
$$
Referring to the Farkas' lemma \cite{ziegler} again, the balanced inequality
$\mathbf e\cdot \mathbf x\ge 0$ holds in polymatroids which satisfy the
conditional independence statements in $\mathcal P$ if and only if $\mathbf
e$ is a linear combination of the basic Shannon inequalities in (\bref{B2})
such that inequalities \emph{not} in
$$
   \mathcal S_{\mathcal P} =\bigcup\big\{\mathcal S(XY\!D): XY\!D\in \mathcal
P\big\}
$$
have non-negative coefficients, while those in $\mathcal S_{\mathcal P}$ can
have arbitrary coefficients. Since we want a minimal generating set, it
suffices to consider cases when combining coefficients of inequalities in
$\mathcal S_{\mathcal P}$ are non-positive, and coefficients of other
inequalities are non-negative.

As was done in Section \ref{subsec:new-ineq}, the complete algorithm works
as follows. Create LP variables indexed by the non-empty subsets of $N$. The
\emph{main variables}, collected into the vector $\mathbf x$, are indexed by
subsets in $\mathcal F$. The remaining \emph{auxiliary variables}, with
indices not in $\mathcal F$, are collected into $\mathbf y$. Create a matrix
with columns indexed by the LP variables, and rows indexed by the basic
Shannon inequalities in (\bref{B2}). A matrix row contains the coefficients
of the corresponding basic inequality, it is negated if the inequality
is in $\mathcal S_{\mathcal P}$. Write the matrix as $(P,Q)$ separating the
main and auxiliary variables. The minimal set of inequalities, of which
everything else follows, can be
recovered from the extremal rays of the polyhedral cone
$$
     \mathcal Q = \{\mathbf h P: \mathbf h Q=\mathbf 0, ~\mathbf h\ge 0 \}.
$$
Similarly to the Copy Lemma case, rows, where the $Q$-part is all zero, can
be deleted, and rows with the same $Q$-part can be merged. Symmetries of the
problem, as in Section \ref{subsec:symmetry} for the iterated Copy Lemma,
can also be applied to reduce the computational complexity. The average of
polymatroids $\sigma(f^*)$ where the permutation $\sigma$ of $N$
preserves $\mathcal F$ (elementwise, not pointwise) provides a polymatroid
$\MEM_{\mathcal F}(f)$ which is symmetrical for all of these
permutations. Thus, without loss of generality, the polymatroid $f^*$
provided by Lemma \ref{lemma:MEM} can be assumed to be symmetrical for all
of these symmetries. In spite of this reduction, the polyhedral computation
is significantly more demanding, and only a few cases have been settled, see
\cite{Csirmaz.oneadhesive,M.fmadhe}. Each of them handles a single-element
3-partition set $\mathcal P$. Claim \ref{claim:i-iv}, which is presented here
without proof, summarizes these results.

\begin{claim}\label{claim:i-iv}
{\upshape(i)} No new inequality is generated when $|D|=1$.

\noindent{\upshape(ii)}
No new inequality is generated for the 3-partition $\<x,y\|ab\>$.

\noindent{\upshape(iii)}
The 3-partition $\<x,y\|abc\>$ generates the non-Shannon inequalities
\begin{align*}
  &(a,x\|c) + (a,b\|x)+(a,b\|y)+(c,y) + {} \\
  &{}+ \left\{{(b,x\|ac)\atop (b,y\|ac)}\right\} + 
       \left\{{(c,x\|ab)\atop (c,y\|ab)}\right\} \ge (a,b),
\end{align*}
(choosing either the top or the bottom line in the curly brackets), plus
all permutations of the elements $xyabc$ keeping the partition.

\noindent{\upshape(iv)}
The 3-partition $\<cd,z\|ab\>$ generates the non-Shannon inequalities
$$
   \left\{{[a,b,c,d\,]\atop [a,c,b,d\,]} \right\} + (a,b|z)+(a,z|b)+(b,z|a)\ge 0,
$$
plus all permutations of the elements $abcdz$ keeping the partition fixed.
\end{claim}

Part (iv) recovers the inequalities stated in Claim \ref{claim:MMRV}.

Any automated inequality prover, such as \cite{ITIP}, can be used to quickly
check that these inequalities are indeed consequences of the stated
arrangements. However, proving that there are no more, is more challenging.
A partial reason stems from polyhedral geometry. Given all extremal rays and
all facets of a polyhedral cone (such as $\mathcal Q$ above), it is not
known whether the maximality of the specified extremal rays can be checked
faster than recreating them from scratch.


\subsection{Many instances}

The iterated Copy Lemma, discussed in Section \ref{subsec:iterating},
creates many instances of the initial polymatroid embedded into the final
polymatroid. A variant of \MEM, called \GMEM, does the same embedding in a
single step. To describe the method we need some preparations. Let
$(f,N)$ be the polymatroid to be embedded, and let $\phi:M \tto N$ be a
many-to-one map with range $N$. The subset $T\subseteq M$ is a
\emph{transversal} if $\phi$, restricted to $T$, is a one-to-one map
between $T$ and $N$. Finally, let $\mathcal T$ be a collection of
transversals that mark the embedded instances in $M$.

\begin{lemma}[Generalized Maximum Entropy Method]\label{lemma:MEM-many}

Suppose $(f,N)$ is entropic (or aent), and let $\phi:M\tto N$ and the
collection $\mathcal T$ of transversals
as above. There is an entropic (or aent, respectively) polymatroid $(g,M)$, such that 
\begin{itemize}\setlength\itemsep{0pt}
\item[\upshape(i)] for all transversals $T\in \mathcal T$, the map $\phi\restr
T$ is an isomorphism between $g\restr T$ and $f$;
\item[\upshape(ii)] if the 3-partition $\<X,Y\|D\>$ of $M$ separates
$\mathcal T$, then $g(X,Y\|D)=0$.
\end{itemize}
\end{lemma}
\begin{proof}
Take $|M|$ many parallel extensions of $(f,N)$ such that $i\in M$ is parallel
to $\phi(i)$, and then restrict the extension to $M$. The resulting
polymatroid $(f',M)$ is clearly entropic (or aent) and satisfies (i).
Apply the \MEM{} Lemma \ref{lemma:MEM} to $f'$ and $\mathcal T$ to get the
required polymatroid $g$.
\end{proof}

Keeping track of instances of the original polymatroid during a sequence of
full copies is relatively easy. Each original instance, which is not a
subset of the current over set, doubles. From this list both the map $\phi$
and the transversal set $\mathcal T$ marking these instances can be
constructed easily. For example, the iterated full copy specified in
(\ref{eq:copyseq}) resulted in a 14-element polymatroid $g$ on $M=\{\ta i,
\tb j, \tc k, \td\ell\}$ with $i$, $j$, and $\ell$ in $\{1,2\}$, and $k$ in
$\{1,\dots,8\}$. The map $\phi$ sends elements from $M$ to $\{\tt a, \tt
b, \tt c, \tt d\} $ simply by removing their indices. The eight
instances of the original polymatroid on $\tt a\tt b\tt c\tt d$ is specified 
by the transversals
\begin{equation}\label{eq:T}
  \begin{array}{cccc}
 \ta1\tb1\tc1\td1, & \ta1\tb1\tc2\td2, & \ta2\tb1\tc3\td1, & \ta2\tb1\tc4\td2, \\
 \ta1\tb2\tc5\td1, & \ta1\tb2\tc6\td2, & \ta2\tb2\tc7\td1, & \ta2\tb2\tc8\td2.
  \end{array}
\end{equation}
LP variables corresponding to subsets of these transversals are set to be
equal to one of the main variables. LP variables, corresponding to other
subsets of $M$ are, or are equal to, some auxiliary variables. Additionally,
the iterated Copy Lemma imposes various independence statements on $M$
(these statements are also multiplied in the copy steps). Taking instead the
maximum entropy extension of $(g,M)$ while fixing the original instances as
marked by the transversals in $\mathcal T$, would hopefully introduce a
nicer, more structured set of independencies, which could then produce
additional non-Shannon inequalities. As stated in Lemma
\ref{lemma:MEM-many}, independencies in the maximum entropy extension arise
from 3-partitions separating $\mathcal T$. In this example they are just the
consequences of the highly symmetrical collection
$$
  \begin{array}{c}
   (\tb1\tc1\tc2\tc3\tc4, \tb2\tc5\tc6\tc7\tc8 \| \ta1\ta2\td1\td2)=0, \\[2pt]
   (\ta1\tc1\tc2\tc5\tc6, \ta2\tc3\tc4\tc7\tc8 \| \td1\td2\tb1\tb2)=0, \\[2pt]
   (\td1\tc1\tc3\tc5\tc7, \td2\tc2\tc4\tc6\tc8 \| \tb1\tb2\ta1\ta2)=0.
  \end{array}
$$
Somewhat surprisingly, this \GMEM{} instance with the transversal set in
(\ref{eq:T}) yields no new inequality at all. It is an immediate consequence
of Claim \ref{claim:no4} below, hinting that, in general, the transversal
set $\mathcal T$ cannot be too dense.

\begin{claim}\label{claim:no4}
Suppose $|N|=4$, and $\mathcal T$ is a set of transversals in $M$ so that,
in every 3-partition $\<X,Y\|D\>$ of $M$ separating $\mathcal T$, the
last part $D$ has at least three elements. Then, for every polymatroid
$(f,N)$, there is another polymatroid $(g,M)$ that satisfies the conditions of
Lemma \ref{lemma:MEM-many}.
\end{claim}
\begin{proof}

Note that if the claim holds for a collection of polymatroids $f$, then
it also holds for their conic combination. Also, the claim holds when $f$ is
entropic by Lemma \ref{lemma:MEM-many}. As discussed in Section
\ref{subsec:case234}, extremal rays of $\Gamma_{\!4}$ are entropic with the
exception of the six V\'amos vectors defined in (\ref{eq:vamos}). Thus this
claim holds for all polymatroids on four points if it holds for the V\'amos
vectors. 

So let $\mathbf v$ be a V\'amos vector. Recall that $\mathbf v$ is $2$ at
singletons, $3$ or $4$ on two-element subsets, and $4$ on the remaining
subsets with three or four elements. Define $g$ on subsets of $M$ as
follows. If $A\subseteq T$ for some transversal $T\in \mathcal T$ or $A$ is
a singleton, then let $g(A)=\mathbf v(\phi(A))$, otherwise let $g(A)=4$.
Since $\phi$ is onto, $g(a)=2$ for every singleton $a\in M$; $g(A)=4$ when
$A$ has at least three elements, and $g(A)=3$ or $g(A)=4$ when $A$ has two
elements. From this it is easy to check that $g$ is a polymatroid, and
$\phi(g\restr T)$ is just $\mathbf v$ for every $T\in\mathcal T$. Finally,
if $\<X,Y\|D\>$ is a 3-partition of $M$ and $|D|\ge 3$, then each of
$g(XD)$, $g(YD)$, $g(D)$ and $g(XYD)$ equals $4$, thus $g(X,Y\|D)=0$. This
proves that condition (ii) of Lemma \ref{lemma:MEM-many} also holds for $g$.
\end{proof}

Suppose $\mathbf v$ is the V\'amos vector $\mathbf v_{cd}$ from
(\ref{eq:vamos}). The proof above also shows that $g(X,Y\|D)=0$ whenever
$g(D)=4$. This is the case for any two-element set $D$ that is either not a
subset of any transversal; or $D$ is in a transversal and $\phi(D)=cd$. A
consequence of this observation is that there is a polymatroidal full copy
of $\mathbf v_{cd}$ over $cd$. Therefore, showing that $\mathbf v_{cd}$ is
not aent cannot be done by taking a copy over $cd$. Taking a copy of $\mathbf
v_{cd}$ over any other two-element subset of $abcd$ leads to the desired
result, similarly to the one discussed after the proof of 
Lemma \ref{lemma:copy}, via either the (\ref{eq:MMRV}) inequality,
or one of its variants, such as
$$
   [a,b,c,d\tsp] + (b,c\|z)+(b,z\|c)+(c,z\|b) + 3\tsp(ad,z\|bc) \ge 0,
$$
see also part (iv) of Claim \ref{claim:i-iv}.

\smallskip

There are several open problems concerning the method of this Section. The
full copy of $f$ is clearly equivalent to an instance of the Generalized
Maximum Entropy method with two transversals. The trivial attempt to
simulate the iterated Copy Lemma in a single \GMEM{} instance failed
spectacularly by Claim \ref{claim:no4}. Can it be done in some other way, or
does the iterated Copy Lemma have some consequence that cannot be obtained
by a single instance of \GMEM?

The particular case of \GMEM{} when $\mathcal T$ is a \emph{sunflower}, i.e., any two
elements of $\mathcal T$ intersect in the same set $D$, and the differences
$T\sm D$ (the ``petals'') are disjoint, is known as the \emph{book
extension} \cite{Csirmaz.book}. Depending on the number of petals, this
arrangement produces larger and larger sets of non-Shannon inequalities.
This case, however, can be simulated easily by the iterated Copy
Lemma. Is it true in general that the consequences of \GMEM{} are also
consequences of the iterated Copy Lemma? Finally, we remark that there has
been no systematic numerical exploration of \GMEM{} with a small number of
transversals.


\section{The Ahlswede-K\"orner method}\label{sec:AK}

This section discusses a third technique to create non-Shannon type
information inequalities. It was used by Makarychev et al.~in \cite{MMRV} to
derive an infinite family of non-Shannon inequalities. In \cite{Kaced} the
method was called MMRV from the initials of the authors of \cite{MMRV}, but it is also
known as the Ahlswede-K\"orner method since it is based on Lemma
\ref{lemma:AK} attributed to them. As proved in Claim \ref{claim:AKpoly},
polymatroids are closed even for the Generalized Ahlswede-K\"orner
operation, thus direct applications of Ahlswede-K\"orner Lemma
\ref{lemma:AK} cannot yield new inequalities. To recap, Lemma \ref{lemma:AK}
says that if $f$ is aent, then so is the function
$$
   f^*(A) = \begin{cases}
     f(A)  & \mbox{ if $z\notin A$}, \\[3pt]
     f(A)-f(z\|N\sm z) & \mbox{ if $z\in A$}.
   \end{cases}
$$
This function $f^*$ is the same as the tightening $f\dn z$, as seen in the
proof of Claim \ref{claim:aent-tight}. The actual statement, which is used in \cite{Kaced}
and in \cite{MMRV} to derive a non-Shannon inequality, is formalized as
Lemma \ref{lemma:AK2} below. Its proof, beyond the Ahlswede-K\"orner lemma,
must use additional machinery.

\begin{lemma}\label{lemma:AK2}
Suppose $(f,N)$ is aent and $N$ is partitioned as $N=X\cups Y\cups \{z\}$.
Then there is an aent function $f^*$ on $N$ such that
$$
   f^*(A) = \begin{cases}
      f(A) & \mbox { if $z\notin A$}, \\[3pt]
      f(A)-f(z\|Y) & \mbox{ if $z \in A\subseteq Yz$}.
   \end{cases}
$$
\end{lemma}

The Ahlswede-K\"orner Lemma \ref{lemma:AK} is a special case of this Lemma
when $X$ is the empty set. For non-empty $X$ Theorem \ref{thm:AK} with
$Z=\{z\}$ and $\alpha=f(z,Y)$ provides an aent polymatroid that gives the
correct values for subsets $z\in A\subseteq Yz$, but that polymatroid fails
to preserve the original values of $f$ on subsets that do not contain $z$ and
intersect $X$.

\begin{proof}

The first step is to use the Copy Lemma to make a one-point extension of $(f,N)$ to
$(f,N\cup\{z'\})$ which is a $z$-copy over $Y$. Since it is  $z$-copy, $f$
restricted to $Y\!z$ and restricted to $Y\!z'$ are isomorphic, moreover
$f(z',Xz\|Y)=0$. The independence gives
$$
   f(XY\!zz')-f(XY\!z) = f(Y\!z')-f(Y)=f(Y\!z)-f(Y)=f(z|Y).
$$
This, and submodularity implies that for all $A\subseteq XY$,
\begin{equation}\label{eq:pfAK2}
   f(Az')-f(A) \ge f(XY\!zz')-f(XY\!z) = f(z|Y).
\end{equation}
To get the aent polymatroid $(f^*,Nz')$ apply the Generalized 
Ahlswede-K\"orner operation, stated in Theorem
\ref{thm:AK}, to the aent polymatroid $(f,Nz')$ with $Z=\{z'\}$ and 
$\alpha=f(Y,z)$. Since $\alpha\le f(z)$ for $J\subseteq Y$, we have
\begin{align*}
  f^*(Jz')&=\min\{f(Jz'),\alpha+f(Jz'|z')\} = f(Jz)-f(z)+\alpha \\
          &= f(Jz)-f(z|Y),
\end{align*}
and, similarly, for $A\subseteq XY$ we have
\begin{align*}
  f^*(A) &= \min\{f(A),\alpha+f(A|z')\} \\
         &= f(A)+\min\{0, f(Az')-f(A)+f(Y)-f(Y\!z) \}\\
         &= f(A),
\end{align*}
since $f(z)=f(z')$, and the second term in $\min\{0,{\cdot}\}$ is non-negative
by (\ref{eq:pfAK2}). Restricting $f^*$ to $XY\!z'$ and renaming $z'$ to $z$ 
proves Lemma \ref{lemma:AK2}.
\end{proof}

The proof gave a slightly stronger result that the stated aent
polymatroid $f^*$ can be recovered as the restriction of a one-point
extension of $f$. This latter statement, however, implies the existence of
an aent $z$-copy of $f$ over $X$. Thus Lemma \ref{lemma:AK2} strengthened
this way and the Copy Lemma for copying a singleton are equivalent.

The following example from \cite{Kaced} illustrates how Lemma
\ref{lemma:AK2} can be used to produce the Makarichev et al.~inequality from
Claim \ref{claim:MMRV}. Consider an aent polymatroid $f$ on the five-element
set $N=\{abcdz\}$, and apply the Lemma for $z\in N$ and $Y=\{ab\}$ to get
the aent polymatroid $f^*$. Polymatroids $f^*$ and $f$ take the same value
on subsets not containing $z$, moreover $f^*(Az)=f(Az)-f(z\|ab)$ for each
$A\subseteq \{ab\}$. Plugging these values to (\ref{eq:MMRV}) inequality
$$
   [a,b,c,d\,]+(a,b\|z)+(a,z\|b)+(b,z\|a)+3\tsp(cd,z\|ab) \ge 0,
$$
which holds for $f^*$\!, the first four terms have the same value for $f^*$ and
$f$, while the last term is $f^*(cd,z\|ab) \le f^*(z\|ab)=f^*(abz)-f(ab)=0$.
Thus the inequality
$$
   [a,b,c,d\,]+(a,b\|z)+(a,z\|b)+(b,z\|a)\ge 0
$$
holds in $f$. Since $f$ was arbitrary, we recover the Makarichev et 
al.~inequality of Claim \ref{claim:MMRV}.

\smallskip

Unlike the Copy Lemma, Lemma \ref{lemma:AK2} cannot directly prove that the
V\'amos vector $\mathbf v=\mathbf v_{cd}$ is not almost entropic. Plugging
$N=\{abcd\}$, $X=\{d\}$, $Y=\{ab\}$, and $z=c$ to Lemma \ref{lemma:AK2}, the
partially defined
function $\mathbf v^*$ can be entropic. Indeed, $\mathbf v^*$ is not defined
on subsets containing $cd$, and on other subsets it takes the same value as
the following linear distribution: give $a$, $b$, and $d$ one-one-one
independent random bits, and give the same fourth random bit to each of $a$,
$b$, $c$, and $d$.

\smallskip

In Lemma \ref{lemma:AK2} the singleton $\{z\}$ can be replaced with a larger
set $Z$. The statement is slightly different, but the proof is similar and is
omitted. In this case even the strengthened version (saying that $f^*$ is a
restriction of an extension of $f$) is strictly weaker than
the Copy Lemma for stating the existence of a $Z$-copy.

\begin{claim}

Suppose $(f,N)$ is aent and $N$ is partitioned as $N=X\cups Y\cups Z$. There
is an aent function $f^*$ on $N$ such that
$$
   f^*(A) = \begin{cases}
      f(A) & \mbox { if $A\cap Z=\emptyset$, and} \\[3pt]
      \min\{f(A),f(AZ)-f(Z\|Y)\} & \mbox{ if $A\subseteq YZ$}.
   \end{cases}
$$
\end{claim}


\section{Conclusions and open problems}\label{sec:conclusion}

Only a few methods and techniques are known to generate non-Shannon
information inequalities. This paper explored these methods and their
relationship, and investigated some of their variants. 

\subsection{Foundation}

The foundation was established in Sections \ref{sec:entreg} and
\ref{sec:polymatroids}.
Basic facts about the entropy region $\GaN$, determined by $N$ jointly
distributed discrete random variables, were recalled and proved in Section
\ref{sec:entreg}. Natural manipulations on probability distributions give
rise to natural operations on the entropy region. These operations include
restricting, deleting, factoring, and independent drawing. Two sophisticated
operations on distributions, called \emph{thinning} and \emph{compressing},
were defined in Section \ref{subsec:fmpe}. They arise as limits of a certain
type of random coding. Similar constructions can be found in the literature
as auxiliary tools in converse theorems of coding problems, but Theorems
\ref{thm:AK} and \ref{thm:princ}, expressing properties of the thinning and
compressing operations, have not been stated in this general form. The
corresponding proofs in the \bref{Appendix} separate typicality
considerations from random coding, thus achieving an error term of order
$(\log n)/n$ instead of the best possible error term $(\log n)/\sqrt n$ when
typical sequences are used.

\emph{Polymatroids}, the second pillar of the foundation, are the abstract
structures satisfying all basic Shannon inequalities. Section
\ref{sec:polymatroids} provided basic facts and specific terminology
concerning polymatroids. Many natural operations on the entropy region
$\GaN$ extend to the set $\GN$ of polymatroids. However, the terminology
sometimes differs. For example, the operation corresponding to
compressing (Section \ref{subsec:fmpe}) is called \emph{principal extension}
(Section \ref{subsec:principal}).

A rich source of probability distributions is linear codes over finite
fields. The corresponding polymatroid class, discussed in Section
\ref{sec:linear}, originates from linear subspace arrangements of finite
dimensional vector spaces. The definition of \emph{linear polymatroid}
is more subtle, because representations over fields with different
characteristics do not mix well.

\subsection{Copy Lemma}

Finding a non-Shannon information inequality is equivalent to separating
regions $\GaN$ and $\GN$. A natural strategy is to look for an operation
that is applicable to entropic vectors but not to every polymatroid. In the
case of the analogous problem of separating linear and general polymatroids,
it would be an operation that applies to linear polymatroids but not to
general ones. Such an operation, defined in Section \ref{sec:linear}, is
the extraction of common information. Using this operation, Claim \ref{claim:4ing}
proves that the Ingleton expression (\ref{eq:ing}) is necessarily
non-negative for linear polymatroids, providing a separating inequality. An
important takeaway is that operations applicable equally to almost entropic
and general polymatroids cannot contribute to a linear non-Shannon
inequality. Table \ref{table:2} summarizes the operations discussed previously,
including all standard operations, indicating which polymatroid classes
they preserve. Clearly, none of them can produce a separating non-Shannon
inequality. Table \ref{table:3} supplements this list with two entries: the
Copy Lemma from Section \ref{sec:copylemma} and the Maximum Entropy Method
from Section \ref{sec:MEP}, indicating that these operations could
create separating inequalities. The question mark indicates that the status
of linear polymatroids is unknown.

\begin{table}[!ht]
\def\mstr#1{\rule{0pt}{1#1pt}}%
\def\yes{$\checkmark$}%
\def\no{$\times$}%
\centering\begin{tabular}{lcccc}
\mstr2 \bf Operation & \bf linear &\bf entropic &\bf aent &\bf polymatroid \\
\hline
\mstr3 Copy Lemma         & \yes & \yes & \yes & \no \\
\mstr1 \MEM               & ?    & \yes & \yes & \no \\
\end{tabular}
\caption{Polymatroid classes preserved by operations}\label{table:3}
\end{table}

The first method, discovered by Zhang and Yeung \cite{ZhY.ineq,Zh.gen.ineq} and later
formalized and named \emph{Copy Lemma} by Dougherty et al.~\cite{DFZ11}, was
described in Section \ref{sec:copylemma}. Proposition \ref{prop:cp-red}
shows that a (partial) $A$-copy is a minor of a full copy and thus cannot have
more consequences. Therefore, it is theoretically sufficient to consider full copies.
In practice, however, due to their smaller base size, a partial copy can be
numerically tractable, while a full copy may not be. Section \ref{subsec:new-ineq}
describes an algorithm that extracts all consequences of a single instance
of the Copy Lemma by transforming it into a vertex enumeration problem
\cite{fukuda-prodon,hamel-loehne-rudloff}.
Section \ref{subsec:balancing} proves that the new inequalities are
necessarily \emph{balanced}, and to generate balanced consequences, only
balanced assumptions are needed. In particular, only the balanced basic Shannon
inequalities in (\bref{B2}) should be used by the extraction algorithm in
Section \ref{subsec:new-ineq}.

Section \ref{subsec:copy-cond} presented three ``preconditions'' that should
be checked before applying the Copy Lemma. If any of them holds, then either
the Copy Lemma yields no result, or the set of ``over'' variables can
be extended to a larger set (while keeping or reducing the computational
complexity of the associated algorithm).

The iterated Copy Lemma, as described in Section \ref{subsec:iterating},
mimics adding one or more non-Shannon entropy inequalities (generated by
other applications of the Copy Lemma) to the copy polymatroid while
providing significantly lower computational complexity. Symmetries, both
from the original problem and from the intrinsic symmetry of the copy
polymatroid when applied iteratively, can further reduce the computational
complexity of the extraction algorithm, as described in Section
\ref{subsec:symmetry}.

\subsection{Maximum Entropy Method}

The Maximum Entropy Method uses the fact that the probability masses of a
distribution can be modified so that some fixed marginals have the same
distribution, while the total entropy, within this constraint, is maximum.
Such a maximum entropy distribution satisfies independence statements
that depend on which marginals are fixed, see Claim \ref{claim:me1}. Two
flavors of this method were presented. The first one, dubbed \MEM,
requires an aent polymatroid $f$, and provides another aent polymatroid
$f^*$ on the same base set, which has the same restrictions as $f$ on a
collection of subsets of the base set, and additionally satisfies
conditional independence statements, see Lemma \ref{lemma:MEM}. Claim
\ref{claim:equiv2} proves that the Copy Lemma is equivalent to the special
case of \MEM{} with a single independence statement. Section
\ref{subsec:me-consequences} describes an algorithm that extracts all
consequences of a \MEM{} instance. The algorithm can easily incorporate the
symmetries of the problem. Results of the algorithm are presented for a
couple of cases in Claim \ref{claim:i-iv}.

The second flavor of the Maximum Entropy Method, called \GMEM, is given as
Lemma \ref{lemma:MEM-many}. It encompasses many instances of the same base
polymatroid marked by a collection of \emph{transversals}, mimicking the
work of the iterated Copy Lemma, while letting the maximum entropy principle
stretch the distribution to satisfy many conditional independence
statements. Defying intuition, high transversal density destroys the
strength of \GMEM, see Claim \ref{claim:no4}.
 
\subsection{Ahlswede-K\"orner method}

The Ahlswede-K\"orner method was used in \cite{MMRV} to derive an infinite
family of non-Shannon inequalities. The method was described explicitly in
\cite{Kaced}. Proven in Claim \ref{claim:AKpoly}, and reiterated in Table
\ref{table:2}, polymatroids are closed under the $\GAK$ operation.
Consequently, direct applications of the Ahlswede-K\"orner lemma (or any of
its consequences) cannot produce new inequalities. The actual statement,
stated as Lemma \ref{lemma:AK2}, that was used, requires additional
machinery. In Section \ref{sec:AK} Lemma \ref{lemma:AK2} is shown to be a
consequence of the Copy Lemma when copying a singleton and to be strictly
weaker (as it cannot prove that the V\'amos vector is not aent).

\subsection{Open Problems}

Theorem \ref{claim:3} showed that $\clGa N$, the closure of the entropy
region, is a full-di\-men\-sion\-al convex cone, and its internal points are
entropic. The structure of the boundary is mostly unknown
\cite{2faces,M.piece,thakor}.

\begin{problem}
Investigate the boundary of $\Ga3$, and the boundary of $\GaN$ in general.
\end{problem}

The similar problem of which faces (intersection of facets) of the outer
boundary cone $\GN$ have entropic points in their relative interior, is
closely related to the inference of conditional independence statements.
This problem has been settled for $|N|=3$ and for $|N|=4$ in \cite{studeny},
while it remains open when $|N|$ is at least five.

For $|N|=5$ the complete list of the extremal rays of $\Gamma_{\!5}$ is
known, see Table \ref{table:1}, but it is not known which of them are
entropic and which are almost entropic.

\begin{problem}
Determine which extremal rays of $\Gamma_{\!5}$ are entropic, and which are
almost entropic.
\end{problem}

The Shannon bound $\GN$ is a full-dimensional cone in the
$d=(2^n\m-1)$-di\-men\-sion\-al space where $|N|=n$, and has
$$
    f=n+{n\choose2}2^{n-2}
$$
bounding hyperplanes (facets, the number of basic Shannon inequalities). As every
extremal ray of $\GN$ is at the intersection of $d$ or more bounding hyperplanes,
$\GN$ has no more than $f \choose d$ extremal rays.
For $n\le 5$ this upper bound is significantly larger
than the actual value.

\begin{problem}
Give reasonable lower and upper bounds on the number of extremal rays in
$\GN$ (see Table \ref{table:1}).
\end{problem}

Linear polymatroids allow the extraction of common information (Table
\ref{table:2}) and are closed for the Copy Lemma (Lemma \ref{lemma:copy}).

\begin{problem}
Does the Copy Lemma have a consequence for linear polymatroids (a rank
inequality) that cannot be derived from the extractability of CI?
\end{problem}

\begin{problem}
Is the class of linear polymatroids closed for \MEM? (see Table
\ref{table:3})
\end{problem}

An affirmative answer to this question does not require that the maximum
entropy distribution be linear. It only requires the existence of a
linear polymatroid that satisfies the independence statements stated in
Lemma \ref{lemma:MEM} and keeps the dimension of the subspaces specified in
$\mathcal F$.

The iterated Copy Lemma has been applied successfully to generate many
non-Shannon inequalities for \emph{four} variables \cite{multiobj,DFZ11}.

\begin{problem}
Start a systematic numerical exploration of the iterated Copy Lemma for five
variables.
\end{problem}

There are many open problems concerning the Maximum Entropy Method. The
trivial attempt to show that \MEM{} can simulate the iterated Copy Lemma
failed spectacularly by Claim \ref{claim:no4}.

\begin{problem}
By Claim \ref{claim:equiv2} the Copy Lemma is equivalent to a special case of
\MEM, while \MEM{} can prove inequalities that are not consequences of a
single application of the Copy Lemma. Compare \MEM{} to the iterated Copy
Lemma.
\begin{itemize}\setlength\itemsep{0pt}
\item[(i)]
Is \MEM{} actually stronger than the iterated Copy Lemma? In other words, can a
single \MEM{} (or \GMEM) prove an entropic inequality that cannot be proved
by an iterated Copy Lemma?

\item[(ii)]
Conversely, is the iterated Copy Lemma stronger than \MEM{}?
\end{itemize}
\end{problem}

\section*{Funding}
\addcontentsline{toc}{section}{Funding}

The research reported in this paper was partially supported by the ERC 
Advanced Grant ERMiD.


\section*{Appendix}\hypertarget{Appendix}{}
\addcontentsline{toc}{section}{Appendix}

The formal proof of Theorems \ref{thm:AK} and \ref{thm:princ} follows the
general idea discussed in Section \ref{subsec:fmpe}. In general terms both
theorems state that some piecewise linear and continuous transformation
$\Phi(\alpha,f)$ (which also depends on a fixed subset $Z$ of the base set
$N$) maps almost entropic points $f$ into almost entropic points for every
non-negative value of $\alpha$. Since the map $\Phi$ is continuous and
homogeneous in the sense that for every $\lambda\ge 0$ we have $\Phi(
\lambda\alpha, \lambda f)=\lambda\Phi(\alpha,f)$,
it suffices to prove the theorems for
a collection of ``special'' aent elements whose multiples are
dense in $\clGa N$. Indeed, if $g$ is such a ``special'' function 
and $\Phi(\alpha',g)$ is aent for every $\alpha'\ge 0$, then
$$
   \Phi(\alpha,\lambda g) = \lambda \Phi({\textstyle\frac{\alpha}
{\lambda}}, g)
$$
is also aent for every non-negative $\alpha$ and $\lambda$ (since $\clGa N$
is closed for multiplication). By assumption, any $f\in\clGa N$ can be
approximated by a sequence of multiples of special elements, say $\lim \lambda_i
g_i = f$. Continuity of $\Phi$ gives $\lim \Phi(\alpha,\lambda_i g_i) =
\Phi(\alpha,f)$. Thus $\Phi(\alpha,f)$ is a limit of aent points, therefore
it is also aent, as required.

For these ``special'' elements we will use entropy functions arising from
group-based distributions; their multiples form a dense subset by Theorem
\ref{thm:group}. So let $\xi$ be a group-based distribution on $N$, its
entropy profile be $g$, and fix both $\alpha\ge 0$ and $Z\subseteq N$.
Choose $n$ to be a large integer, and let $\xi^n$ be the distribution
consisting of $n$ i.i.d.~copies of $\xi$. By Claim \ref{claim:qu}, $\xi^n$
is also a group-based distribution on $N$ (thus quasi-uniform), and its
entropy profile is $n\cdot g$. The random constructions described in Section
\ref{subsec:fmpe} will be performed using the distribution $\xi^n$
(guaranteeing ``large enough'' alphabets) and the selection value will be $n
\tsp \alpha$. The following bound on the distance between the entropy
profile of the generated distribution $\eta$ and the required value
$\Phi(n\tsp\alpha, n\tsp g)$ suffices to conclude both Theorem \ref{thm:AK}
and Theorem \ref{thm:princ}.

\begin{claim}\label{claim:ax-bound}
There is a constant $c$, depending on the distribution $\xi$ only, so that
whenever $n>c$ then with positive probability the generated distribution $\eta$
satisfies
\begin{equation}\label{eq:ap1}
     \big| \Phi(n\tsp\alpha,n\tsp g)(A) - \H_\eta(A)\big| \le 1+\log n
\end{equation}
for all non-empty subsets $A$.
\end{claim}

Indeed, dividing (\ref{eq:ap1}) by $n$ shows that there are aent
points arbitrarily close to $\Phi(\alpha,g)$ (since $\frac1n\H_\eta$, being
a multiple of an entropic vector, is aent). Consequently $\Phi(\alpha,g)$
is aent as well, from which both Theorems follow.

The rest of this Section is devoted to the proof of Claim
\ref{claim:ax-bound} for both random constructions. The
probability of the event $e$ will be denoted by $\prob(e)$; $\E(\zeta)$
denotes the expected value of the random variable $\zeta$, and the entropy
$\H$ is defined using natural logarithm.

First, we repeat the setting from Section \ref{subsec:fmpe}. $Z\subset N$ is
the fixed subset as specified in Theorems \ref{thm:AK} and \ref{thm:princ},
while $A\subseteq N$ is arbitrary. The quasi-uniform distribution is
$\{\xi^n_i:i\in N\}$; $\mathcal Z$ is the support of $\xi^n_Z$ (all elements
of $\mathcal Z$ have equal probability), and, similarly, $\mathcal A$ is the
support of $\xi^n_A$. Clearly, $\H_{\xi^n}(J)= n\H_\xi(J)$ for all
$J\subseteq N$; in the rest of this section $\H(J)$ abbreviates $\H_\xi(J)$.
Since the entropy is based on natural logarithm, we have
$$
   |\mathcal Z| = \exp(n\H(Z)), ~~~\mbox{ and }~~~ 
   |\mathcal A| = \exp(n\H(A)).
$$
Columns of the matrix $\mathbb M$ are labeled by elements of $\mathcal Z$,
its rows are labeled by elements of $\mathcal A$, and the entry at $\mathbb
M[a,z]$ is the probability $\Prob(\xi^n_A=a$, $\xi^n_Z=z)$. From
quasi-uniformity it follows that each matrix entry is either zero or equals
$\exp(-n\H(AZ))$. Probabilities in a row add up to $\exp(-n\H(A))$, and in a
column add up to $\exp(-n\H(Z))$. Thus there are exactly $\exp(n\H(Z\|A))$
non-zero elements in each row, and exactly $\exp(n\H(A\|Z))$ non-zero elements
in each column.

\smallskip

The operation \emph{thinning}, leading to Theorem \ref{thm:AK}, creates the
distribution $\eta$ as follows: choose a column of the matrix $\mathbb M$
randomly and independently $\exp(n\alpha)$ times (the same column can
be chosen multiple times). Form a new matrix $\mathbb N$ from the chosen
columns, which, after normalization, specifies the probability masses of the
distribution $\eta$. The new variable $z'$ will take values as the columns
of $\mathbb N$, and the variable corresponding to $A$ will take the
rows of $\mathbb N$. Observe that $\mathbb N$ contains
$\exp[n(\alpha+\H(A\|Z))]$ non-zero elements, each equal to
$\exp(-n\H(AZ))$.

\begin{claim}\label{claim:AK-1}
For $n$ large enough the following inequality holds with
probability exponentially close to $1$
\begin{equation}\label{eq:AK-1}
 \H_\eta(A) \le  n\tsp\min\{\H(A),\alpha+\H(A\|Z)\} \le \H_\eta(A)+
(1+\log n).
\end{equation}
\end{claim}

Instead of proving the estimate (\ref{eq:AK-1}) for each $A$ simultaneously,
it suffices to prove it for each $A$ separately. The reason is that the
estimate holds for each $A$ with overwhelming probability and there are only
a fixed number of subsets to consider. Consequently the same estimate holds 
for all subsets simultaneously with overwhelming probability as well.

The proof of Claim \ref{claim:AK-1} uses an axiliary lemma and a
Chernoff-type bound. The latter is stated without proof as Lemma
\ref{chernoff}. For a proof of the Chernoff bound and additional information
please consult \cite{chernoff}.

\begin{lemma}\label{ax:aux}
Suppose that the discrete random variable $\xi$ takes every value with
probability at most $p$. Then $\H(\xi)\ge -\log p$.
\end{lemma}
\begin{proof}
The $\log$ function is strictly increasing, thus $p_i\le p$ implies $-\log
p_i\ge -\log p$. Then
$$
  \H(\xi)=\ssum_i -p_i\log p_i \ge \ssum_i -p_i\log p = -\log p,
$$
as claimed.
\end{proof}

\begin{lemma}[Chernoff bound]\label{chernoff}
Let $\zeta_k$ be independent $0\,$--$\,1$ random variables with expected
value $\E(\sum_k
\zeta_k)=\mu$. The following inequality holds for all $0<\eps$:
\begin{equation}\label{eq:chernoff-1}
\prob\mkern-2mu\big(\,{\textstyle\sum_k} \zeta_k \pge
(1\m+\eps)\mu\big) 
  \le \exp[\,\mu\eps-\mu(1\m+\eps)\log(1\m+\eps)].
\end{equation}
\end{lemma}

\begin{proof}[Proof of Claim \ref{claim:AK-1}]

First remark that the matrix $\mathbb N$ has $|\mathcal A|=\exp(n\H(A))$
rows, therefore $\H_\eta(A)\le n\H(A)$. Also, the matrix has
$\exp[n(\alpha+\H(A\|Z))]$ non-zero elements, thus
$\eta_A$ cannot take more than that many values with non-zero probability,
implying $\H_\eta(A)\le n(\alpha+\H(A\|Z))$. This proves the lower
bound in (\ref{eq:AK-1}); it holds with probability 1.

For the upper bound we will use Lemma \ref{ax:aux} to estimate
$\H_\eta(A)$ from below, which requires to estimate the probability of each
row in the matrix $\mathbb N$. To this end pick a row $a\in\mathcal A$, and
for $0<k<\exp(n\alpha)$ let $\zeta^a_k$ be the indicator whether the column
chosen at the $k$-th step intersects the row $a$ in a non-zero value.
Since there are $\exp(n\tsp\H(Z))$ columns in $\mathbb M$, and there are
$\exp(n\tsp\H(Z\|A))$ non-zero elements in a row of $\mathbb M$, the expected
value of $\zeta^a_k$ is
$$
   \E(\zeta^a_k)=\frac{\exp(n\tsp\H(Z\|A))}{\exp(n\tsp\H(Z))} =\exp(-n\tsp\I(A,Z)),
$$
independently of $a$ and $k$. Since $\mathbb N$ has $\exp[n( \alpha+
\H(A\|Z))]$ many non-zero elements of the same value, row $a$ has
probability
$$
   p_a = \frac{\ssum_k \zeta^a_k}{\exp[n(\alpha+\H(A\|Z))]}.
$$
If, for some number $M$, all enumerators are smaller than $M$, then, by
Lemma \ref{ax:aux}, we have
$$
   \H_\eta(A) \ge -\log M + n\tsp(\alpha+\H(A\|Z)).
$$
A quick calculation shows that choosing
$$
   \log M = n\tsp\max\{\alpha-\I(A,Z), 0\} + (1+\log n)
$$
gives the required upper bound in (\ref{eq:AK-1}). To finish the proof we
need to show that, with probability exponentially close to $1$, all sums
$\sum_k\zeta^a_k$ are below $M$. To this end we will use the Chernoff bound
(\ref{eq:chernoff-1}) with $M=(1\m+\eps)\mu$ where
$$
    \mu=\E(\ssum_k\zeta^a_k)=\exp(n\alpha)\tsp\E(\zeta^a_k)
        = \exp[n(\alpha-\I(A,Z))].
$$
The right hand size of (\ref{eq:chernoff-1}) estimates the probability of
a single wrong event; it must be exponentially small even after multiplying
it by the number of rows in $\mathbb N$, which happens to be $\exp(n\H(A))$.
Taking the logarithm, the probability of the wrong event can be estimated as
\begin{align*}
   & \mu\eps-\mu(1\m+\eps)\log(1\m+\eps) + n\H(A) = {}\\
   &~~~~=  M-\mu - M\log(M/\mu) + n\H(A) \le {}\\
   &~~~~\le M-M(1+\log n) + n\H(A) \le {}\\
   &~~~~\le n\big(\H(A)-\log n\big),
\end{align*}
since $\log M-\log \mu \ge 1+\log n$, and $M>n$ by the choice of $M$. This
value is clearly less than $-n$ for $n$ large enough, thus the error
probability is exponentially small, as required.
\end{proof}

\smallskip

The operation \emph{compressing} creates a new variable $z'$ in alphabet
$\mathcal Z'$ with $\exp(\alpha n)$ different values, which we also call
\emph{buckets}. Each column of the matrix $\mathbb M$ is placed randomly and
independently into one of the buckets, creating the new matrix $\mathbb N$
with $\mathcal Z'$ as columns and $\mathcal A$ as rows. The corresponding
distribution is $\eta$. Clearly, the marginal $\eta_A$ remains unchanged.
The following claim provides the missing part of the proof of Theorem
\ref{thm:princ}.

\begin{claim}\label{claim:princ-1}
For $n$ large enough, the following inequality holds with probability 
exponentially close to $1$:
\begin{equation}\label{eq:princ-1}
 \H_\eta(Az') \le n\tsp\min\{\H(A)+\alpha,\H(AZ)\} \le
    \H_\eta(Az')+ (1+\log n).
\end{equation}
\end{claim}

\begin{proof}

Since $Az'$ can take at most $|\mathcal A|\cdot e^{\alpha n}$ different
values, we have $\H_\eta(Az')\le n\tsp\H(A)+n\tsp \alpha$. Similarly, $Az'$
cannot take more values with positive probability than those present in the
matrix $\mathbb M$, which means $\H_\eta(Az') \le n\tsp\H(AZ)$. This
establishes the first inequality in (\ref{eq:princ-1}).

For the upper bound, similarly to the proof of Claim \ref{claim:AK-1}, we
use Lemma \ref{ax:aux} with an upper estimate on the size of the buckets. So let
$a\in \mathcal A$, $z'\in \mathcal Z'$, and $\zeta^{a,z'}_k$ be the
indicator value that the column $k\in \mathcal Z$ with a non-zero value at
row $a$ goes into the bucket $z'$. The probability that $\zeta^{a,z'} _k=1$
is $1/|\mathcal Z'|=\exp(-n\tsp\alpha)$, and the expected size of the bucket
$z'$ at $a$ is
$$
   \mu=\E(\ssum_k \zeta^{a,z'}_k) = \exp[n(\H(Z\|A)-\alpha)]
$$
as the matrix $\mathbb M$ has $\exp(n\tsp\H(Z\|A))$ many non-zero values in
row $a$. If $M$ is an upper bound on the bucket size at each row, then the
entropy of $\eta$ can be bounded as
$$
    \H_\eta(Az') \ge -\log M + n\tsp\H(AZ),
$$
as the total number of non-zero elements in $\mathbb M$ is $\exp(n\tsp
\H(AZ))$, and each goes to some bucket. Choosing
$$
    \log M = n\tsp \max\{\H(Z\|A)-\alpha, 0\} + (1+\log n)
$$
gives the required upper bound in (\ref{eq:princ-1}, so we only need to
check that, with probability exponentially close to $1$, all buckets at
all positions have size less than $M$. Since this number is
$\sum_k\zeta^{a,z'}_k$, using the Chernoff bound from 
Lemma \ref{chernoff} and taking logarithm, the probability that this 
is not the case for any $a$ or $z'$ can be estimated in the exponent as
\begin{align*}
   & \mu\eps-\mu(1\m+\eps)\log(1\m+\eps) + n\tsp\H(AZ) = {}\\
   &~~~~=  M-\mu - M\log(M/\mu) + n\tsp\H(AZ) \le {}\\
   &~~~~\le M-M(1+\log n) + n\tsp\H(AZ) \le {}\\
   &~~~~\le n\big(\H(AZ)-\log n\big),
\end{align*}
as $\log M-\log\mu\ge 1+\log n$ and $M>n$ by the choice of $M$. As this
is clearly smaller than $-n$ for $n$ large enough, Claim \ref{claim:princ-1}
is proved.
\end{proof}


\end{document}